\documentclass[enabledeprecatedfontcommands,
11pt, a4paper
]{scrartcl}

\usepackage[utf8]{inputenc}
\usepackage[T1]{fontenc}
\usepackage{parskip}
\usepackage[scaled=.85]{roboto}
\usepackage{titlesec}

\titlespacing\section{0pt}{12pt plus 4pt minus 2pt}{0pt plus 2pt minus 2pt}
\titlespacing\subsection{0pt}{12pt plus 4pt minus 2pt}{0pt plus 2pt minus 2pt}
\titlespacing\subsubsection{0pt}{12pt plus 4pt minus 2pt}{0pt plus 2pt minus 2pt}
\usepackage[style=authoryear, backend=biber, maxcitenames=2, maxnames=2,maxbibnames=6, giveninits=true, uniquename=false,uniquelist=false,%uniquename=init,
			sortcites=true,doi=false,url=false, bibencoding=utf8]{biblatex} % load the package, 
\DeclareLanguageMapping{english}{english-apa}
\renewbibmacro{in:}{}
\usepackage{csquotes}
\usepackage{microtype}

\usepackage{amsmath}
\usepackage{mathtools}
\usepackage{amssymb}
\usepackage{mathrsfs} 
\usepackage{amsthm}
\usepackage{dsfont}
\allowdisplaybreaks

\newcommand\numberthis{\addtocounter{equation}{1}\tag{\theequation}}

\usepackage{xcolor}
\usepackage{colortbl}
\definecolor{lightgray}{gray}{0.8}
\definecolor{lightblue}{RGB}{0, 51, 153}%{0, 51, 102}
\definecolor{lightgreen}{RGB}{184,255,176}
\definecolor{marineblue}{RGB}{79, 118, 181}
\usepackage[colorlinks, linkcolor=lightblue, urlcolor=lightblue, bookmarks=false, citecolor=marineblue]{hyperref}

\newtheoremstyle{stheorem}% hnamei
{11pt}% hSpace abovei
{11pt}% hSpace belowi
{\itshape}% hBody fonti
{}% hIndent amounti
{\bfseries}% hTheorem head fonti
{.}% hPunctuation after theorem headi
{11pt}% hSpace after theorem headi
{}% hTheorem head spec (can be left empty, meaning `normal')i

\newtheoremstyle{sdef}% hnamei
{11pt}% hSpace abovei
{11pt}% hSpace belowi
{}% hBody fonti
{}% hIndent amounti
{\bfseries}% hTheorem head fonti
{.}% hPunctuation after theorem headi
{11pt}% hSpace after theorem headi
{}% hTheorem head spec (can be left empty, meaning `normal')i

\newtheoremstyle{srem}% hnamei
{11pt}% hSpace abovei
{11pt}% hSpace belowi
{\slshape}% hBody fonti
{}% hIndent amounti
{\bfseries\slshape}% hTheorem head fonti
{.}% hPunctuation after theorem headi
{11pt}% hSpace after theorem headi
{}% hTheorem head spec (can be left empty, meaning `normal')i

\newtheoremstyle{sex}% hnamei
{21pt}% hSpace abovei
{21pt}% hSpace belowi
{\slshape}% hBody fonti
{}% hIndent amounti
{\sffamily\bfseries}% hTheorem head fonti
{.}% hPunctuation after theorem headi
{11pt}% hSpace after theorem headi
{}% hTheorem head spec (can be left empty, meaning `normal')i

\theoremstyle{stheorem}
\newtheorem{theorem}{Theorem}
\newtheorem{lemma}{Lemma}
\newtheorem{proposition}{Proposition}
\newtheorem{corollary}{Corollary}
\theoremstyle{srem}
\newtheorem{remark}{Remark}

\theoremstyle{sdef}
\newtheorem{definition}{Definition}
\newtheorem{assumption}{Assumption}

\usepackage{authblk}

\usepackage{fancyhdr}
\def\T{{ \mathrm{\scriptscriptstyle T} }}

\newcommand{\pr}{\mathrm{pr}}

\title{Large Sample Asymptotics of the Pseudo-Marginal Method}
\author[1]{S. M. Schmon}
\author[1]{G. Deligiannidis}
\author[1]{A. Doucet}
\author[2]{M. K. Pitt}
\affil[1]{Department of Statistics, University of Oxford, UK}
\affil[2]{Department of Mathematics, King's College London, UK}
\setkomafont{author}{\Large\sffamily}
\setkomafont{date}{\large}

\begin{document}

\maketitle

\begin{abstract}
The pseudo-marginal algorithm is a variant of the Metropolis--Hastings algorithm which samples asymptotically from a probability distribution when it is only possible to estimate unbiasedly an unnormalized version of its density. Practically, one has to trade-off the computational resources used to obtain this estimator against the asymptotic variances of the ergodic averages obtained by the pseudo-marginal algorithm.
Recent works optimizing this trade-off rely on some strong assumptions which can cast doubts over their practical relevance. In particular, they all assume that the distribution of the difference between the log-density and its estimate is independent of the parameter value at which it is evaluated. Under regularity conditions we show here that, as the number of data points tends to infinity, a space-rescaled version
of the pseudo-marginal chain converges weakly towards another pseudo-marginal chain for which this assumption indeed holds. A study of this limiting chain allows us to provide parameter dimension-dependent guidelines on how to optimally scale a normal random walk proposal and the number of Monte Carlo samples for the pseudo-marginal method in the large-sample regime. This complements and validates currently available results.
\end{abstract}

\begin{refsection} % refsection environment

\section{Introduction}

The pseudo-marginal algorithm is a variant of the popular Metropolis--Hastings algorithm where an unnormalized version of the target density is replaced by a non-negative unbiased estimate. The algorithm first appeared in the physics literature \autocite{linliuSloan2000} and has become popular in Bayesian statistics as many intractable
likelihood functions can be estimated unbiasedly using importance
sampling or particle filters \autocite{beaumont2003estimation,andrieu2009pseudo,andrieu:doucet:holenstein2010}.

Replacing the true likelihood in the Metropolis-Hastings algorithm with an estimate results in a trade-off: the asymptotic variance of an ergodic average of a pseudo-marginal chain typically decreases as the number of Monte Carlo samples, $N$, used to obtain the likelihood estimator increases, as established by \textcite{andrieuvihola2016} for importance sampling estimators; however, this comes at the cost of a higher computational burden.
An important task in practice is thus to choose $N$ such that the computational resources required to obtain a given asymptotic variance are minimized.
This problem has already been investigated by \textcite{PittSilvaGiordaniKohn}, \textcite{doucet2015efficient} and \textcite{Sherlock2015efficiency} where guidelines
have been obtained under various assumptions either on the proposal \autocite{PittSilvaGiordaniKohn,doucet2015efficient} or on the proposal and target distribution \autocite{Sherlock2015efficiency}.

Additionally, all these contributions make the assumption that the noise in the log-likelihood estimator, that is the difference between this estimator and the true log-likelihood, is Gaussian with variance inversely proportional to $N$,
its mean and variance being independent of the parameter value at
which it is evaluated. A similar assumption has also been used by \textcite{Nemeth2016} for the analysis of a related algorithm.
This assumption can cast doubts over the practical relevance of the guidelines provided in these contributions.
The normal noise assumption was motivated by \textcite{PittSilvaGiordaniKohn}, \textcite{doucet2015efficient} and \textcite{Sherlock2015efficiency}
by the fact that the error in the log-likelihood estimator for state-space
models computed using a particle filter is asymptotically normal of
variance proportional to $\gamma$ as $T\rightarrow\infty$ with $N=T/\gamma$ \autocite{berarddelmoraldoucet2014} while the constant variance assumption over the parameter space was motivated in \textcite{PittSilvaGiordaniKohn} and \textcite{doucet2015efficient}
by the fact that the posterior typically concentrates as $T$ increases.
However, no formal argument justifying why the pseudo-marginal chain would behave as a Markov chain for which these assumptions hold has
been provided.

We carry out here an original weak convergence analysis
of the pseudo-marginal algorithm in a Bayesian setting which not only justifies these assumptions
but also allows us to obtain novel guidelines on how to optimally tune this algorithm as a function of the parameter dimension $d$.
Weak convergence techniques have become very popular in the Markov chain Monte Carlo literature since their introduction in the
seminal paper of \textcite{RobertsGelmanGilks1997}. To the recent exception of \textcite{deligiannidis2015}, all these analyses have been performed in the asymptotic regime where the parameter dimension $d\rightarrow\infty$. Results of this type typically require making strong structural assumptions on the target distribution such as having $d$ independent and identically distributed components as in \textcite{Sherlock2015efficiency}. We analyse here the pseudo-marginal scheme in the large-sample asymptotic regime where the number of data  points $T$ goes to infinity while $d$ is fixed. Under weak regularity conditions, we show that a space-rescaled version of the pseudo-marginal chain converges to a pseudo-marginal chain targeting a normal distribution for which the noise in the log-likelihood estimator is indeed also normal of constant mean and variance.
We provide numerical results to optimally scale normal random walk proposals and the noise variance to optimize the performance of this limiting Markov chain as a function of $d$.
These guidelines complement and validate the results obtained in \textcite{doucet2015efficient} and \textcite{Sherlock2015efficiency}. All proofs can be found in the supplementary material.

\section{The Pseudo-Marginal Algorithm\label{sec:PseudoMarginal}}

\subsection{Background\label{sec:ModelKernel}}
Consider a Bayesian model on the Borel space $\left\{\Theta,\mathcal{B}\left(\Theta\right)\right\}$
where $\Theta\subseteq\mathbb{R}{}^{d}$. The parameter $\theta\in\Theta$
follows a prior distribution $p(\mathrm{d}\theta)$ while  $\theta\mapsto p(y\mid\theta)$ denotes the likelihood
function, where
$y=(y_{1},\ldots,y_{T})$ denotes the vector of observations. When
the likelihood arises from a complex latent variable model an analytic expression of $p(y\mid \theta)$ might not be available. Hence, the standard Metropolis\textendash Hastings algorithm cannot be used to sample the posterior distribution
$p(\mathrm{d}\theta\mid y)\propto p(\mathrm{d}\theta)\thinspace p(y\mid\theta)$
as the likelihood ratio $p(y\mid\theta^{\prime})/p(y\mid\theta)$
appearing in the Metropolis--Hastings acceptance probability, when at parameter $\theta$
and proposing $\theta^{\prime}$, cannot be computed.
Assume we have access to an unbiased positive estimator $\hat{p}(y\mid\theta,U)$
of the intractable likelihood $p(y\mid\theta)$, where $U\sim m_{\theta}$
represents the auxiliary variables on $\left\{\mathcal{U},\mathcal{B\left(U\right)}\right\}$
used to compute this estimator. We introduce the following probability
measure on $\left\{\Theta\times\mathcal{U},\mathcal{B}\left(\Theta\right)\times\mathcal{B\left(\mathbb{\mathcal{U}}\right)}\right\}$
\[
\pi(\mathrm{d}\theta,\mathrm{d}u)=p(\mathrm{d}\theta\mid y)\frac{\hat{p}(y\mid\theta,u)}{p(y\mid\theta)}m_{\theta}\left(\mathrm{d}u\right),
\]
which satisfies $\pi(\mathrm{d}\theta)=p(\mathrm{d}\theta\mid y)$. The pseudo-marginal algorithm is a Metropolis--Hastings scheme targeting $\pi(\mathrm{d}\theta,\mathrm{d}u)$,
hence marginally $p(\mathrm{d}\theta\mid y)$, using a proposal distribution
$Q\left(\theta,u;\mathrm{d}\theta',\mathrm{d}u'\right)=q(\theta,\mathrm{d}\theta')m_{\theta'}\left(\mathrm{d}u'\right)$.
This yields the acceptance probability
\begin{equation*}
\alpha(\theta,u;\theta',u')=\min\left\{ 1,r(\theta,\theta')\frac{\hat{p}(y\mid\theta',u')/p(y\mid\theta')}{\hat{p}(y\mid\theta,u)/p(y\mid\theta)}\right\},\:\text{where}~~  r(\theta,\theta')=\frac{\pi(\mathrm{d}\theta')}{\pi(\mathrm{d}\theta)}\frac{q(\theta',\mathrm{d}\theta)}{q(\theta,\mathrm{d}\theta')}. %\label{eq:pmmh_accept}
\end{equation*}
As in previous contributions \autocite{andrieu2009pseudo,PittSilvaGiordaniKohn,andrieuvihola2015,doucet2015efficient,Sherlock2015efficiency},
we analyse the pseudo-marginal algorithm using additive noise in the log-likelihood estimator, writing $Z(\theta)=\mathrm{log}\thinspace\hat{p}(y\mid\theta,U)-\mathrm{log}\thinspace p(y\mid\theta)$. This parameterization allows us to write the target distribution as a measure on $\left\{\Theta\times\mathbb{R},\mathcal{B}\left(\Theta\right)\times\mathcal{B\left(\mathbb{R}\right)}\right\}$
with
\[
\pi(\mathrm{d}\theta,\mathrm{d}z)=p(\mathrm{d}\theta\mid y)\mathrm{exp}\left(z\right)g\left(\mathrm{d}z\mid\theta\right),
\]
where $Z(\theta)\sim g\left(\cdot \mid\theta\right)$ when $U\sim m_{\theta}$ and the pseudo-marginal kernel is
\[
P\left(\theta,z;\mathrm{d}\theta',\mathrm{d}z'\right)=q(\theta,\mathrm{d}\theta')g(\mathrm{d}z'\mid\theta')\alpha\big(\theta,z;\theta',z'\big)+\rho(\theta,z)\delta_{(\theta,z)}(\mathrm{d}\theta',\mathrm{d}z'),
\]
with acceptance probability
\begin{align*}
\alpha\big(\theta,z;\theta',z'\big) & =\min\left\{ 1,r(\theta,\theta')\mathrm{exp}\left(z'-z\right)\right\} ,
\end{align*}
and corresponding rejection probability $\rho(\theta,z)$.

\subsection{Literature review\label{sec:related}}

We review here recent research motivating this work. To this
end, we need to introduce a few additional notations. Let $\mu$
be a probability measure on $\left\{\mathbb{R}^{n},\mathcal{B}(\mathbb{R}^{n})\right\}$
and ${\Pi\colon\mathbb{R}^{n}\times\mathcal{B}(\mathbb{R}^{n})\rightarrow[0,1]}$
a Markov transition kernel. For any measurable function $f$ and measurable
set $A$, we write $\mu(f)=\int f(x)\mu(\mathrm{d}x)$, $\mu(A)=\mu\left\{\mathbb{I}_{A}\left(\cdot\right)\right\}$ and ${\Pi f\left(x\right)=\int\Pi\left(x,\mathrm{d}y\right)f\left(y\right)}$.
We consider the Hilbert space $L^{2}(\mu)$ with inner product ${\langle f,g\rangle_{\mu}=\int f(x)g(x)\mu(\mathrm{d}x)}$.
For a function $f\in L^{2}(\mu)$, the asymptotic variance of averages
of a stationary Markov chain $\left(X_{k}\right)_{k\geqslant1} $ of $\mu$-invariant
transition kernel $\Pi$ is defined as
\begin{equation*}
\mathrm{var}(f,\Pi)=\lim_{M\rightarrow\infty}\frac{1}{M}E\left\{\sum_{k=1}^{M}f(X_k)-\mu(f)\right\}^2,
\end{equation*}
and $\mathrm{var}(f,\Pi)=\mathrm{var}_\mu(f)\thinspace\textsc{iat}(f,\Pi)$ when
the integrated autocorrelation time given by
\begin{equation*}
\textsc{iat}(f,\Pi)=1+2\sum_{k=1}^{\infty}\frac{\mathrm{cov}\left\{f(X_0), f(X_k)\right\}}{\mathrm{var}\left\{f(X_0) \right\}}
\end{equation*}
is finite. We denote by $\varphi(x; m,\Lambda)$ the normal density of argument $x$, mean $m$ and covariance $\Lambda$.

In order to obtain guidelines to balance computational cost and accuracy of the likelihood estimator \textcite{PittSilvaGiordaniKohn}, \textcite{doucet2015efficient} and \textcite{Sherlock2015efficiency} make the simplifying assumption that $g\left(\mathrm{d}z\mid\theta\right)=\varphi(\mathrm{d}z;-\sigma^{2}/2,\sigma^{2})$,
that $\sigma^{2}\propto1/N$, and focus on functions $f\in L^{2}(\pi)$
such that $f\left(\theta,z\right)=f\left(\theta,z'\right)$ for any
$z,z'$. Under these assumptions, it was first proposed by \textcite{PittSilvaGiordaniKohn} to minimize
\begin{equation}
\textsc{ct}(f,P_{\sigma})=\frac{\textsc{iat}(f,P_{\sigma})}{\sigma^{2}},\label{eq:ct}
\end{equation}
with respect to $\sigma$ where
\begin{equation}
P_{\sigma}\left(\theta,z;\mathrm{d}\theta',\mathrm{d}z'\right)=q(\theta,\mathrm{d}\theta')\varphi(\mathrm{d}z;-\sigma^{2}/2,\sigma^{2})\alpha\big(\theta,z;\theta',z'\big)+\rho_{\sigma}(\theta,z)\delta_{(\theta,z)}(\mathrm{d}\theta',\mathrm{d}z'),\label{eq:KernelPMMHlimit}
\end{equation}
$\rho_{\sigma}(\theta,z)$ being the corresponding rejection probability.
The criterion (\ref{eq:ct}) arises from the fact that the computational
time required to evaluate the likelihood is typically proportional
to $N$. Under the additional assumption that $q(\theta,\mathrm{d}\theta')=\pi(\mathrm{d}\theta')$,
the minimizer of $\textsc{ct}(f,P_{\sigma})$ is $\sigma=0$$\cdot$92 \autocite{PittSilvaGiordaniKohn}.
For general proposal distributions \textcite{doucet2015efficient} minimize upper bounds on $\textsc{ct}(f,P_{\sigma})$.
This results in guidelines stating that one should indeed select $\sigma$ around
1$\cdot$0 when the Metropolis--Hastings algorithm using the exact likelihood would provide
an estimator having a small integrated autocorrelation time and around 1$\cdot$7 when this autocorrelation time is
very large \autocite{doucet2015efficient}. In practical scenarios, the
integrated autocorrelation time of the Metropolis--Hastings algorithm using the exact likelihood is unknown and the results in \textcite{doucet2015efficient}
suggest to select $\sigma$ around 1$\cdot$2 as a robust default choice. A slightly different approach is taken by \textcite{Sherlock2015efficiency}.
In addition to similar noise assumptions, it is assumed that the posterior
factorizes into $d$ independent and identically distributed components
and that one uses an isotropic normal random walk proposal of jump
size proportional to $\ell$. In this context, one maximizes with respect to $(\sigma,\ell)$ the expected squared jump
distance associated to the pseudo-marginal sequence of the
first parameter component $(\vartheta_{1,k})_{k\geqslant0}$
divided by the noise variance as $d\rightarrow\infty$.
In this asymptotic regime, a time-rescaled version of $(\vartheta_{1,k})_{k\geqslant0} $ converges weakly to a diffusion process and the adequately rescaled expected squared jumping distance converges to the squared diffusion coefficient of this process. Maximizing this squared jump distance is asymptotically equivalent to minimizing $\textsc{ct}(f,P_{\sigma})$ irrespective of $f$  \autocite[see ][]{RobertsRosenthal2014} and its maximizing arguments are $\sigma=1$$\cdot$8 and $\ell=2\cdot$56 \autocite[Corollary 1]{Sherlock2015efficiency}.

In practice, the standard deviation of the log-likelihood estimator varies over the parameter space and one selects $N$ such that this standard deviation is approximately equal to the desired $\sigma$ for a parameter value around the mode of the posterior obtained through a preliminary run.

The strong assumptions made in those contributions can bring into question the merits of the guidelines provided within these papers.
Our original weak convergence analysis of the pseudo-marginal algorithm justifies this assumption in the large sample regime,  as $T\rightarrow\infty.$
This convergence occurs under fairly weak regularity assumptions on
the posterior distribution. The resulting limiting algorithms can be optimized to provide guidelines for random walk proposals without relying on any upper bound as in \textcite{doucet2015efficient}.

\section{Large Sample Asymptotics of the Pseudo-Marginal Algorithm\label{sec:WeakCV}}

\subsection{Notation and assumptions\label{sec:Assumptions}}
Our analysis of the pseudo-marginal algorithm relies on the assumption that the posterior concentrates (Assumption \ref{ass1}) which is most commonly formulated using convergence in probability with respect to the data distribution, denoted $\mathbb{P}^Y$.
For our result to hold under this weak assumption we take into account the randomness induced by the data, resulting in a random Markov chain and requiring us to deal with weak convergence of random probability measures. To make this more precise we introduce the following notation.

The observations $(Y_t)_{t\geqslant1}$ are regarded as
random variables defined on a probability space $\left\{\mathbb{\mathsf{Y}}^{\mathbb{N}},\mathcal{B}(\mathsf{Y})^{\mathbb{N}},\mathbb{P}^{Y}\right\}$,
where $\mathcal{B}(\mathsf{Y})^{\mathbb{N}}$ denotes the Borel $\sigma$-algebra
and we write $\Omega=\mathsf{Y}{}^{\mathbb{N}}$ for brevity. For
$T\geqslant1$ we can define the random variables $Y_{1:T}=\left(Y_{1},\ldots,Y_{T}\right)$
as the coordinate projections to $\mathsf{Y}^{T}$. Then, for $\omega=\left( y_{t}\right)_{t\geqslant1} \in\Omega$,
$\pi_{T}^{\omega}(\mathrm{d}\theta)=p(\mathrm{d}\theta\mid y_{1:T})$
denotes a regular version of the target posterior distribution and,
for any $\theta\in\Theta$, $g_{T}^{\omega}\left(\mathrm{d}z\mid\theta\right)$
the conditional distribution of the error in the log-likelihood estimator given observations $y_{1:T}$. The measures $\pi_{T}^\omega$ and $g_{T}^\omega$
can be interpreted as random measures. Relevant results for random
measures are briefly discussed in Section \ref{sec:proofoftheorem}
and in more detail in the supplementary material. In the following we will use a superscript $\omega$ to highlight that a certain quantity depends on the data.
All probability densities considered hereafter are with respect to the
Lebesgue measure and we use the same symbols for distributions and
densities, for example $\mu\left(\mathrm{d}\theta\right)=\mu\left(\theta\right)\mathrm{d}\theta$.

In this context, the target distribution of the pseudo-marginal algorithm
is
\begin{equation*}
\pi_{T}^{\omega}(\mathrm{d}\theta,\mathrm{d}z)=\pi_{T}^{\omega}(\mathrm{d}\theta)\mathrm{exp}\left(z\right)g_{T}^{\omega}\left(\mathrm{d}z\mid\theta\right),\label{eq:randomtarget}
\end{equation*}
and its transition kernel is
\begin{equation*}
P_{T}^{\omega}\left(\theta,z;\mathrm{d}\theta',\mathrm{d}z'\right)=q_{T}(\theta,\mathrm{d}\theta')g_{T}^{\omega}(\mathrm{d}z'\mid\theta')\alpha_{T}^{\omega}\big(\theta,z;\theta',z'\big)+\rho_{T}^{\omega}(\theta,z)\delta_{(\theta,z)}\left(\mathrm{d}\theta',\mathrm{d}z'\right),\label{eq:kernelPseudoMarginal}
\end{equation*}
where
\begin{align*}
\alpha_{T}^{\omega}\big(\theta,z;\theta',z'\big) & =\min\left\{ 1,\frac{\pi_{T}^{\omega}(\mathrm{d}\theta')}{\pi_{T}^{\omega}(\mathrm{d}\theta)}\frac{q_{T}(\theta',\mathrm{d}\theta)}{q_{T}(\theta,\mathrm{d}\theta')}\thinspace\mathrm{exp}\left(z'-z\right)\right\},
\end{align*}
$\rho_{T}^{\omega}(\theta,z)$ is the corresponding rejection
probability.

Our first assumption is that the posterior distributions concentrate
towards a normal at rate $1/\surd{T}$. We denote by $\mathcal{Y}_{T}$ the $\sigma$-algebra spanned by $Y_{1:T}$.
\begin{assumption}
\label{ass1} The posterior distributions $\{ \pi_{T}^{\omega}\left(\mathrm{d}\theta\right)\}_{T\geqslant1} $
admit Lebesgue densities and there exists a ${d\times d}$ positive definite
matrix $\Sigma$, a parameter value $\bar{\theta}\in\Theta$ and a
sequence $(\hat{\theta}_{T}^{\omega})_{T\geqslant1}$ of $\mathcal{Y}_{T}$-adapted random variables such that as $T\rightarrow\infty$
\begin{equation}
\int\left|\pi_{T}^{\omega}(\theta)-\varphi\big(\theta;\hat{\theta}_{T}^{\omega},\Sigma/T\big)\right|\mathrm{d}\theta\rightarrow 0,\quad\quad\hat{\theta}_{T}^{\omega}\rightarrow \bar{\theta},\label{eq:bernmis}
\end{equation}
both limits being in $\mathbb{P}^{Y}$-probability.
\end{assumption}

Assumption \ref{ass1} is satisfied if a Bernstein-von
Mises theorem holds; see  \textcite[Theorem 10.1]{vandervaart2000} and \textcite{kleijn2012}.
Our second assumption is that we use random walk proposal distributions
with appropriately scaled increments.
\begin{assumption}
\label{ass2} The proposal distributions $\{q_{T}(\theta,\mathrm{d}\theta')\}_{T\geqslant1}$
admit densities of the form
\[
q_{T}(\theta,\theta')=\surd{T}\nu\left\{\surd{T}(\theta'-\theta)\right\},
\]
where $\nu$ is a continuous density on $\mathrm{\mathbb{R}}^{d}$.
\end{assumption}

Finally, we assume that the error in the log-likelihood estimator
satisfies a central limit theorem conditional upon $\mathcal{Y}_{T}$
and that this convergence holds uniformly in a neighbourhood of $\bar{\theta}$.
\begin{assumption}
\label{ass3} There exists an $\varepsilon$\textup{-ball} $B(\bar{\theta})$
around $\bar{\theta}$ such that the distributions of the error in
the log-likelihood estimator $\left\{ g_{T}^{\omega}(\mathrm{d}z\mid\theta)\right\}_{T\geqslant1}$
satisfy as $T\rightarrow \infty$
\begin{equation}
\sup_{\theta\in B(\bar{\theta})}d_{\mathrm{BL}}\left[g_{T}^{\omega}\left(\,\cdot\mid\theta\right),\varphi\left\{\,\cdot\,;-\sigma^{2}(\theta)/2,\sigma^{2}(\theta)\right\} \right]\rightarrow 0,\quad \text{in } \mathbb{P}^Y\text{-probability}, \label{eq:uniformCLT}
\end{equation}
where $d_{\mathrm{BL}}(\cdot,\cdot )$ denotes the bounded Lipschitz metric and the function $\sigma\colon\Theta\rightarrow[0,\infty)$ is continuous at $\bar{\theta}$ with $0<\sigma(\bar{\theta})<\infty$.
An analogous result holds for $\bar{g}_{T}^{\omega}(\mathrm{d} z\mid\theta)=\exp(z)g_{T}^{\omega}(\mathrm{d}z\mid\theta)$,
the distribution of this error at equilibrium, that is as $T\rightarrow \infty$
\begin{equation}
\sup_{\theta\in B(\bar{\theta})}d_{\mathrm{BL}}\left[\bar{g}_{T}^{\omega}\left(\,\cdot\mid\theta\right),\varphi\left\{\cdot;\sigma^{2}(\theta)/2,\sigma^{2}(\theta)\right\}\right] \rightarrow 0\quad \text{in } \mathbb{P}^Y\text{-probability.}\label{eq:uniformCLTequilibrium}
\end{equation}
\end{assumption}

We will refer to convergence in probability with respect to the bounded Lipschitz metric as weak convergence in probability.
In Section \ref{sec:clt}, we provide sufficient conditions under which Assumption 3 is satisfied for random effect models where the likelihood estimator is a product of $T$ independent importance sampling estimators. This differs from scenarios where the likelihood estimator is given by one single importance sampling estimator studied in \textcite{sherlock2017pseudo}. Empirical evidence in \autocite{PittSilvaGiordaniKohn} and \autocite{doucet2015efficient} also suggests that Assumption 3 might hold for a large class of state-space models when the likelihood is estimated using particle filters. Under strong assumptions, a standard central limit theorem has been established in \autocite{berarddelmoraldoucet2014} for $g_{T}^{\omega}\left(\,\cdot\mid\theta\right)$. However, it would be technically very challenging to provide weak sufficient conditions under which Assumption 3 holds in this context.

\subsection{Weak convergence in the large sample regime\label{sec:WeakCVstatement}}

Denote by $(\vartheta_{T,k}^{\omega},Z_{T,k}^{\omega})_{k\geqslant0}$ the stationary Markov chain defined by the pseudo-marginal kernel,  $(\vartheta_{T,0}^{\omega},Z{}_{T,0}^{\omega})\sim\pi_{T}^{\omega}$
and $(\vartheta_{T,k}^{\omega},Z_{T,k}^{\omega})\sim P_{T}^{\omega}(\vartheta_{T,k-1}^{\omega},Z_{T,k-1}^{\omega};\cdot)$
for $k\geqslant1$. Let $\chi_{T}^{\omega}=(\tilde{\vartheta}_{T,k}^{\omega},Z_{T,k}^{\omega})_{k\geqslant0}$
where $\tilde{\vartheta}_{T,k}^{\omega}=\surd{T}(\vartheta_{T,k}^{\omega}-\hat{\theta}_{T}^{\omega})$
is the Markov chain arising from rescaling the parameter component
of the pseudo-marginal chain. Its transition kernel is thus
\begin{equation}
\tilde{P}_{T}^{\omega}(\tilde{\theta},z;\mathrm{d}\tilde{\theta}^{\prime},\mathrm{d}z')=\tilde{q}_{T}(\tilde{\theta},\mathrm{d}\tilde{\theta}^{\prime})\tilde{g}_{T}^{\omega}(\mathrm{d}z'|\tilde{\theta}^{\prime})\tilde{\alpha}_{T}^{\omega}\big(\tilde{\theta},z;\tilde{\theta}^{\prime},z'\big)+\tilde{\rho}_{T}^{\omega}(\tilde{\theta},z)\delta_{(\tilde{\theta},z)}(\mathrm{d}\tilde{\theta}^{\prime},\mathrm{d}z'),
\label{eq:Markov}
\end{equation}
where
\begin{align*}
\tilde{\alpha}_{T}^{\omega}(\tilde{\theta},z;\tilde{\theta}^{\prime},z'\big) & =\min\left\{ 1,\frac{\tilde{\pi}_{T}^{\omega}(\mathrm{d}\tilde{\theta}^{\prime})}{\tilde{\pi}_{T}^{\omega}(\mathrm{d}\tilde{\theta})}\frac{\tilde{q}_{T}(\tilde{\theta}^{\prime},\mathrm{d}\tilde{\theta})}{\tilde{q}_{T}(\tilde{\theta},\mathrm{d}\tilde{\theta}^{\prime})}\thinspace\mathrm{exp}\left(z'-z\right)\right\} ,
\end{align*}
$\tilde{\rho}_{T}^{\omega}(\theta,z)$ is the corresponding rejection
probability, $\tilde{\pi}_{T}^{\omega}(\tilde{\theta})=\pi_{T}^{\omega}(\hat{\theta}_{T}^{\omega}+\tilde{\theta}/\surd{T})/\surd{T}$,
$\tilde{q}_{T}(\tilde{\theta},\tilde{\theta}^{\prime})=q_{T}(\hat{\theta}_{T}^{\omega}+\tilde{\theta}/\surd{T},\hat{\theta}_{T}^{\omega}+\tilde{\theta}^{\prime}/\surd{T})/\surd{T}$
and $\tilde{g}_{T}^{\omega}(z\mid\tilde{\theta})=g_{T}^{\omega}(z\mid\hat{\theta}_{T}^{\omega}+\tilde{\theta}/\surd{T})$.
Under Assumption \ref{ass2}, we have $\tilde{q}_{T}(\tilde{\theta},\tilde{\theta}^{\prime})=\nu(\tilde{\theta}^{\prime}-\tilde{\theta})=\tilde{q}(\tilde{\theta},\tilde{\theta}^{\prime})$.
We now state the main result of this paper.
\begin{theorem}
\label{theorem} Under Assumptions \ref{ass1}, \ref{ass2} and \ref{ass3},
the sequence of stationary Markov chains $(\chi_{T}^{\omega})_{T\geqslant1}$
converges weakly in \textup{$\mathbb{P}^{Y}$}-probability as $T\rightarrow\infty$
to the law of a stationary Markov chain of initial distribution
\begin{equation}
\tilde{\pi}(\mathrm{d}\tilde{\theta},\mathrm{d}z)=\varphi(\mathrm{d}\tilde{\theta};0,\Sigma)\varphi\left(\mathrm{d}z;\sigma^{2}/2,\sigma^{2}\right)\label{eq:invariantlimit}
\end{equation}
and transition kernel
\begin{equation}
\tilde{P}(\tilde{\theta},z;\mathrm{d}\tilde{\theta}^{\prime},\mathrm{d}z')=\tilde{q}(\tilde{\theta},\mathrm{d}\tilde{\theta}^{\prime})\varphi\left(\mathrm{d}z';-\sigma^{2}/2,\sigma^{2}\right)\tilde{\alpha}(\tilde{\theta},z;\tilde{\theta}^{\prime},z')+\tilde{\rho}(\tilde{\theta},z)\delta_{(\tilde{\theta},z)}(\mathrm{d}\tilde{\theta}^{\prime},\mathrm{d}z')\label{eq:limitingMarkov}
\end{equation}
where $\sigma=\sigma(\bar{\theta})$,
\begin{align*}
\tilde{\alpha}(\tilde{\theta},z;\tilde{\theta}^{\prime},z')=\min\left\{ 1,\frac{\varphi(\tilde{\theta}^{\prime};0,\Sigma)}{\varphi(\tilde{\theta};0,\Sigma)}\frac{\tilde{q}(\tilde{\theta}^{\prime},\tilde{\theta})}{\tilde{q}(\tilde{\theta},\tilde{\theta}^{\prime})}\mathrm{exp}\left(z'-z\right)\right\}  & ,
\end{align*}
and $\tilde{\rho}(\theta,z)$ is the corresponding rejection probability.
\end{theorem}

Under this asymptotic regime, the limiting transition kernel $\tilde{P}$ in (\ref{eq:limitingMarkov})
is also a pseudo-marginal kernel where the noise distribution is $\varphi\left(\mathrm{d}z;-\sigma^{2}/2,\sigma^{2}\right)$
as assumed in previous analyses \autocite{PittSilvaGiordaniKohn,doucet2015efficient,Sherlock2015efficiency}. As Theorem \ref{theorem} is a weak convergence
result, it does not imply that the integrated autocorrelation time of the pseudo-marginal kernel $\tilde{P}_{T}^{\omega}$ converges to the one of $\tilde{P}$.
However, for large $T$, this suggests that some characteristics of $\tilde{P}_{T}^{\omega}$ can indeed be captured by those of the kernel (\ref{eq:KernelPMMHlimit})
which can be obtained from $\tilde{P}$ by using the
change of variables $\theta=\hat{\theta}_{T}^{\omega}+\tilde{\theta}/\surd{T}$
and substituting the true target for its normal approximation $\varphi(\theta;\hat{\theta}_{T}^{\omega},\Sigma/T)$,
hence removing a level of approximation.%

\section{Outline of the Proof of the Main Result\label{sec:proofoftheorem}}

\subsection{Random Markov chains}

The proof of Theorem \ref{theorem} follows from a slightly more general
result on weak convergence of random Markov chains on Polish spaces
given in Theorem \ref{theo:appendix} below. We introduce here some notation and recall some definitions concerning
random probability measures needed to define random
Markov chains; see the supplementary material or \textcite{crauel2003} for more details.

Let $(\Omega,\mathcal{F},\mathbb{P})$ be a probability space and
$S$ a Polish space endowed with its Borel $\sigma$-algebra $\mathcal{B}(S)$.
We equip the product space $\Omega\times S$ with the product $\sigma$-algebra
$\mathcal{F}\otimes\mathcal{B}(S)$. We denote by $\mathcal{P}(S)$
the space of Borel probability measures which is itself endowed with
the Borel $\sigma$-algebra $\mathcal{B}\{\mathcal{P}(S)\}$ generated
by the weak topology. Finally, $C_{b}(S)$, respectively $\mathrm{BL}(S)$,
denote the sets of continuous bounded functions, respectively the
set of bounded Lipschitz functions.
\begin{definition}%[\emph{Random probability measure}]
\emph{\label{def:random_measure} }A \emph{random probability measure}
is a map $\mu\colon\Omega\times\mathcal{B}\left(S\right)\rightarrow[0,1]$,
$(\omega,B)\mapsto\mu(\omega,B)=\mu^{\omega}(B)$, such that for
every $B\in\mathcal{B}\left(S\right)$ the map $\omega\mapsto\mu(\omega,B)$
is measurable while $\mu^{\omega}\in\mathcal{P}(S)$ $\mathbb{P}-$almost
surely.
\end{definition}

For all bounded and measurable functions $g\colon\Omega\times S\rightarrow\mathbb{R}$,
$\omega\mapsto\int_{S}g(\omega,x)\mu^{\omega}(\mathrm{d}x)$ is measurable
\autocite[Proposition 3.3]{crauel2003} and thus the map $\omega\mapsto\mu^{\omega}(f)$
is a random variable for bounded measurable functions $f\colon S\rightarrow\mathbb{R}$.
Consequently, $\mu^{\omega}\colon\Omega\rightarrow\mathcal{P}(S)$
is a Borel measurable map. Conversely,
it can be shown that any random element of $\left\{\mathcal{P}(S),\mathcal{B}(\mathcal{P}(S))\right\}$
fulfils the conditions set out in Definition \ref{def:random_measure};
see \textcite[Remark 3.20 (i)]{crauel2003} or \textcite[Lemma 1.37]{Kallenberg2006}.
\begin{definition}%[\emph{Random Markov kernel}]
 A \emph{random Markov kernel }is a map $K\colon\Omega\times S\times\mathcal{B}(S)\rightarrow[0,1],\quad(\omega,x,B)\mapsto K(\omega,x,B)=K^{\omega}(x,B)$,
such that
\end{definition}
\begin{itemize}
\item[\emph{(i)}]  $(\omega,x)\mapsto K^{\omega}(x,B)$ is $\mathcal{F\otimes\mathcal{B}}(S)$-measurable
for every $B\in\mathcal{B}(S),$
\item[\emph{(ii)}] $K^{\omega}(x,\cdot)\in\mathcal{P}(S)$ $\mathbb{P}-$almost surely
for every $x\in S$.
\end{itemize}
\begin{lemma}%[\emph{Random Markov chain}]
\label{lem:Random-Markov-chain}\textup{ Given a random probability
measure} $\mu^{\omega}$\textup{ and random Markov kernel $K^{\omega}$,
there exists an almost surely unique random probability measure $\mu^{\mathbb{N},\omega}$
on $S^{\mathbb{N}}$ such that
\[
\mu^{\mathbb{N},\omega}(A_{1}\times\ldots\times A_{k}\times E_{k+1})=\int_{A_{1}}\mu^{\omega}(\mathrm{d}x_{1})\int_{A_{2}}K^{\omega}(x_{1},\mathrm{d}x_{2})\ldots\int_{A_{k}}K^{\omega}(x_{k-1},\mathrm{d}x_{k})
\]
for any }$A_{i}\in\mathcal{B}(S)$ $(i=1,\ldots,k)$, $k\in \mathbb{N}$ and $E_{k+1}=\boldsymbol{\times}_{i=k+1}^{\infty}S$.
\end{lemma}

\subsection{Convergence of random Markov chains}

For a sequence of random probability measures $(\mu_{n}^\omega)_{n\geqslant1}$,
respectively a sequence of random Markov kernels $(K_{n}^\omega)_{n\geqslant1}$,
converging in a suitable sense towards a probability measure $\mu$,
respectively a Markov kernel $K$, we show here that the distributions
of the associated Markov chains $(\mu_{n}^{\mathbb{N,\omega}})_{n\geqslant1}$
defined in Lemma \ref{lem:Random-Markov-chain} converge weakly in
probability to the distribution $\mu^{\mathbb{N}}$ of the homogeneous
Markov chain of initial distribution $\mu$ and Markov
kernel $K$.
\begin{theorem}%[Weak convergence of random Markov chains]
\label{theo:appendix} If the following assumptions hold,
\begin{itemize}
\item[(T.1)]\label{t1} the random probability measures $\left(\mu_{n}^{\omega}\right)_{n\geqslant1}$ converge weakly in probability to a probability measure $\mu$ as $n\rightarrow \infty$,
\item[(T.2)] the random Markov transition kernels $\left(K_n^\omega\right)_{n\geqslant1}$ satisfy
\begin{equation*}
\int\left|K_{n}^{\omega}f(x)-Kf(x)\right|\mu_{n}^{\omega}(\mathrm{d}x)\rightarrow0\label{eq:ass2}
\end{equation*}
in probability as $n\rightarrow \infty$ for all $f\in\mathrm{BL}(S)$ where $K$ is a Markov transition kernel
,
\item[(T.3)] the transition kernel $K$ is such that \textup{$x\mapsto Kf(x)$
is continuous for any $f\in C_{b}(S)$, }
\end{itemize}
then, as $n\rightarrow \infty$, the measures $(\mu_{n}^{\mathbb{N},\omega})_{n\geqslant1}$ on $S^\mathbb{N}$ converge weakly in probability to
the measure $\mu^\mathbb{N}$ induced by the Markov chain with initial distribution $\mu$ and transition kernel $K$.
\end{theorem}

\subsection{Application to the pseudo-marginal algorithm}

Theorem \ref{theorem}  follows from Theorem \ref{theo:appendix} by showing that, under Assumptions \ref{ass1}, \ref{ass2} and \ref{ass3}, all conditions set out in Theorem \ref{theo:appendix} are
fulfilled. Firstly, as we increase the number of data points, the stationary distribution of the Markov chain will converge weakly to the limiting stationary distribution of Theorem \ref{theo:appendix}.
\begin{proposition}
\label{prop:4.1}Under Assumptions \ref{ass1} and \ref{ass3},
we have
\[
\tilde{\pi}_{T}^{\omega}(\mathrm{d}\tilde{\theta},\mathrm{d}z)\rightarrow \tilde{\pi}(\mathrm{d}\tilde{\theta},\mathrm{d}z),
\]
weakly in $\mathbb{P}^Y$-probability as $T\rightarrow\infty$ where $\tilde{\pi}_{T}^{\omega}(\mathrm{d}\tilde{\theta},\mathrm{d}z)=\tilde{\pi}_{T}^{\omega}(\mathrm{d}\tilde{\theta})\mathrm{exp}\left(z\right)\tilde{g}_{T}^{\omega}(\mathrm{d}z\mid\tilde{\theta})$. $\tilde{\pi}_{T}^{\omega}(\mathrm{d}\tilde{\theta})$ and $\tilde{g}_{T}^{\omega}(\mathrm{d}z)$ are defined in Section \ref{sec:WeakCVstatement} and $\tilde{\pi}(\mathrm{d}\tilde{\theta},\mathrm{d}z)$ in equation \eqref{eq:invariantlimit}.
\end{proposition}
This follows as the marginal $\pi_T^\omega(\mathrm{d}\theta)$ concentrates around the limiting parameter value $\bar{\theta}$ while the noise uniformly converges towards a normal distribution in a neighbourhood around $\bar{\theta}$. The next proposition ensures the stability of the transition and can be proven using similar arguments.
\begin{proposition}
\label{prop:4.2}Under Assumptions \ref{ass1}, \ref{ass2} and \ref{ass3}, as $T\rightarrow\infty$
we have for any $f\in\mathrm{BL}(\mathbb{R}^{d+1})$
\[
\mathbb{\int}|\tilde{P}_{T}^{\omega}f(\theta,z)-\tilde{P}f(\theta,z)|\tilde{\pi}_{T}^{\omega}(\mathrm{d}\theta,\mathrm{d}z)\rightarrow 0,\quad \text{in } \mathbb{P}^Y\text{-probability},
\]
where the transition kernels $\tilde{P}_{T}^{\omega}$ and $\tilde{P}$ are defined in equations \eqref{eq:Markov} and \eqref{eq:limitingMarkov}.
\end{proposition}
A further requirement to ensure the stability of the transition is that the application of the transition operator conserves continuity.
\begin{proposition}
\label{prop:4.3}Under Assumption \ref{ass2},
the map $(\theta,z)\mapsto\tilde{P}f(\theta,z)$ is continuous for every $f\in C_b(\mathbb{R}^{d+1})$.
\end{proposition}
Theorem \ref{theorem} now follows from a direct application of Theorem
\ref{theo:appendix} as the assumptions \emph{(T.1)}, \emph{(T.2)} and \emph{(T.3)} hold by
Proposition \ref{prop:4.1}, \ref{prop:4.2} and \ref{prop:4.3},
respectively.

\section{Random effects models\label{sec:clt}}

\subsection{Statistical model and likelihood estimator}

We provide here sufficient conditions under which weak convergence
of the pseudo-marginal algorithm is verified for an important class
of latent variable models. Consider the model
\begin{equation}
X_{t} \sim f(\cdot\mid\theta)\text{,}\qquad Y_{t}\mid X_{t}\sim g(\cdot\mid X_{t},\theta),\label{eq:independentlatentvariablemodels}
\end{equation}
where $(X_{t})_{t\geqslant 1}$ are independent $\mathbb{R}^{k}$-valued
latent variables, $f(x\mid\theta)$ is a density
with respect to Lebesgue measure and $\left(Y_{t}\right)_{t\geqslant1} $
are $\mathsf{Y}$-valued observations distributed according to a conditional
density $g(y\mid x,\theta)$ with respect to a dominating measure,
$\mathsf{Y}$ being a topological space. For observations $Y_{1:T}=y_{1:T}$
the likelihood is
\[
p(y_{1:T}\mid\theta)=\prod_{t=1}^{T}p(y_{t}\mid\theta)=\prod_{t=1}^{T}\int g(y_{t}\mid x_{t},\theta)f(x_{t}\mid\theta)\mathrm{d}x_{t}.
\]
In many scenarios, this likelihood is not available analytically.
If one wants to perform Bayesian inference about the parameter $\theta$, we can thus use the pseudo-marginal algorithm as it is
possible to obtain an unbiased non-negative estimator of $p(y_{1:T}\mid\theta)$
using importance sampling. Indeed, we can consider $\hat{p}(y_{1:T}\mid\theta,U)=\prod_{t=1}^{T}\hat{p}(y_{t}\mid\theta,U_{t})$ where $U=\left(U_{1},...,U_{T}\right),$ $U_{t}=\left(U_{t,1},...,U_{t,N}\right)$,
$U_{t,i}$ is $\mathbb{R}^{k}$-valued, $N$ denotes the number of
Monte Carlo samples and $\hat{p}(y_{t}\mid\theta,U_{t})$ is an
importance sampling estimator of $p(y_{t}\mid\theta)$ is
\[
\hat{p}(y_{t}\mid\theta,U_{t})=\frac{1}{N}\sum_{i=1}^{N}w(y_{t},U_{t,i},\theta),\text{\hspace{1cm}}w(y_{t},U_{t,i},\theta)=\frac{g(y_{t}\mid U_{t,i},\theta)f(U_{t,i}\mid\theta)}{h(U_{t,i}\mid y_{t},\theta)},
\]
where $U_{t,i}\sim h(\cdot\mid y_{t},\theta)$, $h(\cdot\mid y_{t},\theta)$ being a probability density on $\mathbb{R}^{k}$
with respect to Lebesgue measure. In this case the joint density $m_{T,\theta}\left(u\right)$ of
all the auxiliary variates used to obtain the likelihood estimator
is given by the product over $t=1,...,T$ and $i=1,...,N$ of $h(u_{t,i}\mid y_{t},\theta)$.
We will assume subsequently that the true observations are independent and identically distributed samples taken from a probability measure $\mu$ so the joint data distribution is the product measure $\mathbb{P}^{Y}(\mathrm{d}\omega)=\prod_{t=1}^{\infty}\mu(\mathrm{d}y_{t})$.

\subsection{Verifying the assumptions}

The Bernstein--von Mises theorem holds under weak regularity assumptions; see 
\textcite[Theorem 10.1]{VanderVaart2000} and the supplementary material (Section S3$\cdot$2) for the case of generalized linear mixed models presented in Section \ref{sec:glmm}. %and Section \ref{subsec:Random-Effects-Models}.
This ensures Assumption \ref{ass1} is satisfied while Assumption
\ref{ass2} is easy to satisfy, selecting for example a multivariate normal
proposal of covariance scaling as $1/\surd{T}$. Assumption \ref{ass3}
is more complicated as it requires to establish uniform conditional
central limit theorems for $\hat{p}(Y{}_{1:T}\mid\theta,U)$ in scenarios where
$U\sim m_{T,\theta}$ arise from the proposal, so $Z\sim g_{T}^{\omega}\left(\cdot\mid\theta\right),$
or at stationarity where $U\sim\pi_{T}^{\omega}(\cdot\mid\theta)$
with
\[
\pi_{T}^{\omega}(u\mid\theta)=\frac{\hat{p}(y_{1:T}\mid\theta,u)}{p(y_{1:T}\mid\theta)}m_{T,\theta}(u),
\]
implying that $Z\sim\bar{g}_{T}^{\omega}\left(\cdot\mid\theta\right)$.
We denote
\begin{align*}
\sigma^{2}(y,\theta) & =\mathrm{var}\left\{\overline{w}(y,U_{1,1},\theta)\right\}, \quad \sigma^{2}(\theta) =E\left\{\sigma^{2}(Y_{1},\theta)\right\},
\end{align*}
with $U_{1, 1} \sim h(\cdot \mid y, \theta)$, $Y_{1}\sim\mu$ and
\begin{equation}
\overline{w}(Y_{t},U_{t,i},\theta)=\frac{w(Y_{t},U_{t,i},\theta)}{p(Y_{t}\mid\theta)}.\label{eq:normalizedweight}
\end{equation}
However, under the following assumption, we show here that Assumption \ref{ass3} holds.
\begin{assumption}
\label{ass4} There exists a closed $\varepsilon$-ball $B(\bar{\theta})$
around $\bar{\theta}$ and a function $g$ such that the normalized
weight $\overline{w}(y,U_{1,1},\theta)$ defined in (\ref{eq:normalizedweight})
satisfies for some $0<\Delta<1$
\begin{equation*}
\sup_{\theta\in B(\bar{\theta})}E\left\{\overline{w}(y,U_{1,1},\theta)^{2+\Delta}\right\} \leq g(y),\label{eq:unif:bound}
\end{equation*}
where $U_{1,1}\sim h(\,\cdot\mid y,\theta)$ and $\mu(g)<\infty$. Additionally, $\theta\mapsto \sigma^2(y,\theta)$ is continuous in
$\theta$ on $B(\bar{\theta})$ for all $y\in\mathsf{Y}$.
\end{assumption}

\begin{theorem}%[Uniform Central Limit Theorem]
\label{thm:uniform_CLT}
Under Assumption \ref{ass4}, Assumption \ref{ass3} is satisfied.
\end{theorem}
Theorem \ref{thm:uniform_CLT} strengthens earlier results of \textcite[Theorem 1]{deligiannidis2015} which obtain standard central limit theorems for the error in the log-likelihood estimator.

\subsection{Generalized linear mixed models}
\label{sec:glmm}
A common example of random effects models is the class of generalized linear mixed models \autocite[see ]{mcculloch2005generalized}, where the observation density is a member of the exponential family and the latent variable follows a centred Gaussian distribution. The densities with respect to some dominating measure can be written as
\begin{equation}\label{eq:exp_model}
	g(y \mid x, \theta) = \prod_{j=1}^J m(y_j) \exp\left[\eta_j(x) T(y_j) - A\{\eta_j(x)\}\right], \quad f(x \mid \theta) = \varphi(x; 0, \tau^2),
\end{equation}
where $\eta_j(x) = c_j^{\T}\beta + x$, $c$ is a vector of covariates with corresponding parameter vector $\beta$, $A(\eta)$ denotes the $\log$-partition function and $m(y)$ is a base measure. In section S3$\cdot$2 of the supplementary material, we show that for many such models the assumptions of Theorem \ref{theorem} can be verified. In particular, we show that Assumption 4 holds thus Assumption 3 holds by Theorem \ref{thm:uniform_CLT}.

\section{Efficient Implementation of the Pseudo-Marginal Random Walk Algorithm\label{sec:optimalscaling}}

\subsection{Optimal tuning}

We optimize the performance of the limiting pseudo-marginal
chain identified in Theorem \ref{theorem} as a proxy for the optimization
of the original pseudo-marginal chain. We assume that the limiting
covariance matrix $\Sigma$ in \eqref{eq:bernmis} is the identity matrix $I_{d}$ with $d$
denoting the parameter dimension. For general covariance matrices,
we can use a Cholesky decomposition and a change of variables as in \autocite{Sherlock2015efficiency, Nemeth2016}. We
denote by $\tilde{P}_{\ell,\sigma}$ the transition kernel (\ref{eq:limitingMarkov})
using the proposal density
\begin{equation*}
q(\theta,\theta')=\varphi\left(\theta';\theta,\ell^{2}I_{d}/d\right).
\end{equation*}
As in \textcite{PittSilvaGiordaniKohn} and \textcite{doucet2015efficient}, we propose to minimize $\textsc{ct}(f,\tilde{P}_{\ell,\sigma})$, as defined in \eqref{eq:ct},
with respect to the noise standard deviation $\sigma$ but, contrary
to these contributions, also with respect to
the scale parameter $\ell$. We restrict attention here to the case
where $f\left(\theta,z\right)=\theta_{1}$, the first component of
$\theta$, and write $\textsc{ct}(f,\tilde{P}_{\ell,\sigma})=\textsc{ct}(\ell,\sigma)$
in this case. As this criterion is not available in closed-form, we
simulate the limiting Markov chain initialized in its stationary regime
with different noise levels $\sigma$ and scales $\ell$
on a fine grid to obtain empirical estimates of $\textsc{ct}(\ell,\sigma)$
computed using the overlapping batch mean estimator. %Other estimators did not provide significantly different results.
This simulation is straightforward as the target and noise distributions in this limiting
case are both Gaussian. We then find the approximate minimizer $(\hat{\ell}_{\mathrm{opt}},\hat{\sigma}_{\mathrm{opt}})$
of $\textsc{ct}(\ell,\sigma)$ over this grid. This
set-up is applied for parameter dimension
$d$ ranging from $1$ to 50. The results are summarized in Table~\ref{tab:multi_dim}.

\begin{table}
\centering
\begin{tabular}{crrrr}
Dimension $d$  & $\hat{\ell}_{\mathrm{opt}}$   & $\hat{\sigma}_{\mathrm{opt}}$   & $\textsc{ct}(\hat{\ell}_{\mathrm{opt}},\hat{\sigma}_{\mathrm{opt}})$  & $\pr_{\mathrm{acc}}(\hat{\ell}_{\mathrm{opt}},\hat{\sigma}_{\mathrm{opt}})$ \tabularnewline
$d=1$  & 2$\cdot$05 (0$\cdot$25)  & 1$\cdot$16 (0$\cdot$07) &  8$\cdot$47 & 25$\cdot$73\%\tabularnewline
$d=2$  & 1$\cdot$97 (0$\cdot$14)  & 1$\cdot$21 (0$\cdot$06) & 12$\cdot$71 & 22$\cdot$92\%\tabularnewline
$d=3$  & 2$\cdot$11 (0$\cdot$07)  & 1$\cdot$24 (0$\cdot$05) & 16$\cdot$79 & 19$\cdot$97\%\tabularnewline
$d=5$  & 2$\cdot$17 (0$\cdot$12)  & 1$\cdot$30 (0$\cdot$05) & 23$\cdot$18 & 17$\cdot$35\%\tabularnewline
$d=10$ & 2$\cdot$20 (0$\cdot$08)  & 1$\cdot$44 (0$\cdot$05) & 37$\cdot$93 & 14$\cdot$27\%\tabularnewline
$d=15$ & 2$\cdot$33 (0$\cdot$08)  & 1$\cdot$50 (0$\cdot$00) & 53$\cdot$43 & 12$\cdot$07\%\tabularnewline
$d=20$ & 2$\cdot$34 (0$\cdot$10)  & 1$\cdot$54 (0$\cdot$05) & 65$\cdot$62 & 11$\cdot$44\%\tabularnewline
$d=30$ & 2$\cdot$36 (0$\cdot$11)  & 1$\cdot$61 (0$\cdot$03) & 90$\cdot$46 & 10$\cdot$41\%\tabularnewline
$d=50$ & 2$\cdot$41 (0$\cdot$10)  & 1$\cdot$74 (0$\cdot$05) &136$\cdot$38  & 8$\cdot$66\%\tabularnewline
\end{tabular}
\caption{Optimal values for scaling $\ell$ and noise $\sigma$ and associated
value of computing time and average acceptance probability: mean and standard deviation of the minimizers over $10$ runs.}
\label{tab:multi_dim}
\end{table}

Table \ref{tab:multi_dim} also lists the computing time at these values and the average acceptance probability of the proposal under $\tilde{P}_{\ell,\sigma}$
at stationarity by using 5 million iterates of the chain. The obtained results
are consistent with those in \textcite{doucet2015efficient} and \textcite{Sherlock2015efficiency}.
For low dimensions, $1\leq d\leq5$, the ideal Metropolis--Hastings algorithm mixes well and $\hat{\sigma}_{\mathrm{opt}}$
is around 1$\cdot$1-1$\cdot$3 as suggested by \textcite{doucet2015efficient} and
it increases slowly as $d$ increases to the values ${\left(\ell_{\infty},\sigma_{\infty}\right)=(2{\cdot}56,1{\cdot}81)}$
obtained by the diffusion limit \autocite{Sherlock2015efficiency}. For
example, for $d=50$, we obtain $(\hat{\ell}_{\mathrm{opt}},\hat{\sigma}_{\mathrm{opt}})=(2{\cdot}41,1{\cdot}74)$
and the resulting optimal computing time $\textsc{ct}(\hat{\ell}_{\mathrm{opt}},\hat{\sigma}_{\mathrm{opt}})$
is close to $\textsc{ct}(\ell_{\infty},\sigma_{\mathrm{\infty}})$.
For lower dimensions, however, the performance in terms of computing
time can be increased by reducing $\sigma$ and $\ell$ in comparison to $\sigma_{\infty}$ and $\ell_{\infty}$;
see Table \ref{tab:comparison}. We also observed empirically that
the cost function $\ell\mapsto\textsc{ct}(\ell,\sigma)$ is fairly
flat as noticed in the limiting case by \textcite{Sherlock2015efficiency}.

\subsection{Implementation}
We now show how to exploit the results of the last section in practice to design an efficient implementation of the pseudo-marginal algorithm. Using a preliminary run, we compute estimates $\hat{\theta}$, $\hat{\Sigma}$  of the posterior mean and posterior covariance matrix. For the parameter dimension $d$, we choose $\ell$ according to Table~\ref{tab:multi_dim} and use a Gaussian random walk proposal with covariance matrix $\hat{\ell}_{\mathrm{opt}}^2\hat{\Sigma}/d$. Finally we select the number of Monte Carlo samples $N$ such that the sample standard deviation of the log-likelihood estimate at $\hat{\theta}$ matches the optimal value $\hat{\sigma}_{\mathrm{opt}}$ listed in Table~\ref{tab:multi_dim}. This approach is similar to the one followed in \textcite{Sherlock2015efficiency} except for the dimension dependence of the recommended parameters $(\hat{\ell}_{\mathrm{opt}},\hat{\sigma}_{\mathrm{opt}})$.

\begin{table}
\centering
\begin{tabular}{crrr}
Dimension $d$ & $\textsc{ct}(\ell_{\infty},\hat{\sigma}_{\mathrm{opt}})$  & $\textsc{ct}(\ell_{\mathrm{\infty}},\sigma=1.2)$  & $\textsc{ct}(\ell_{\mathrm{\infty}},\sigma_\infty)$ \tabularnewline
$d=1$  &  9$\cdot$04 (0$\cdot$25) & 9$\cdot$05  (0$\cdot$21) & 17$\cdot$10 (1$\cdot$34)\tabularnewline
$d=2$  & 13$\cdot$48 (0$\cdot$32) & 13$\cdot$37 (0$\cdot$28) & 22$\cdot$45 (0$\cdot$81)\tabularnewline
$d=3$  & 17$\cdot$63 (0$\cdot$28) & 17$\cdot$43 (0$\cdot$26) & 26$\cdot$71 (0$\cdot$64)\tabularnewline
$d=5$  & 24$\cdot$38 (0$\cdot$44) & 24$\cdot$72 (0$\cdot$31) & 34$\cdot$14 (0$\cdot$88)\tabularnewline
$d=10$ & 40$\cdot$17 (0$\cdot$71) & 41$\cdot$60 (0$\cdot$24) & 47$\cdot$08 (1$\cdot$03)\tabularnewline
$d=15$ & 53$\cdot$69 (0$\cdot$72) & 58$\cdot$01 (0$\cdot$50) & 59$\cdot$08 (0$\cdot$79)\tabularnewline
$d=20$ & 67$\cdot$15 (0$\cdot$53) & 74$\cdot$34 (0$\cdot$36) & 71$\cdot$41 (1$\cdot$48)\tabularnewline
$d=30$ & 91$\cdot$36 (0$\cdot$95) & 106$\cdot$08 (0$\cdot$34) & 93$\cdot$73 (1$\cdot$08)\tabularnewline
$d=50$ &136$\cdot$49 (1$\cdot$18) & 167$\cdot$83 (0$\cdot$75) & 135$\cdot$92 (1$\cdot$27)\tabularnewline
\end{tabular}
\caption{Comparison of the computing time for different noise levels. $\hat{\sigma}_{\mathrm{opt}}$
denotes the minimizer of the estimated integrated autocorrelation
time, as shown in Table \ref{tab:multi_dim}.}
\label{tab:comparison}
\end{table}

\section{Simulation study: Random Effects Model}
\label{sec:sim_rem}

We now illustrate how the guidelines derived from the limiting pseudo-marginal chain compare to a practical implementation of the pseudo-marginal algorithm. We consider a logistic mixed effects model applied to a real data set. Mixed models are popular in econometrics, survey
analysis and medical statistics amongst others and are often used
to describe heterogeneity between groups. Here we consider a subset
of a cohort study of Indonesian preschool children.
This dataset was previously analysed using Bayesian mixed models by \textcite{zeger1991}.
It contains 1200 observations of 275 children.
We model the probability of a respiratory infection based on the following covariates: age, sex, height, an indicator for presence of vitamin deficiency, an indicator for subnormal height and two seasonal components.
Including the intercept we have 8 covariates. Cluster
effects due to repeated measurements of the same children are modelled with
individual random intercepts. In this case the linear predictor of
a regression model based on covariates $c_{t,j}$~$(t=1,\ldots,T,j=1,\ldots J)$
reads $\eta_{t,j}=c_{t,j}^{\T}\beta+X_{t}$ where $X_{t}\sim\mathcal{N}(0,\tau)$
denotes the random intercept for children $t=1,\ldots,T$ and $\beta$ the
regression parameters. For every child, we have an observation vector
$y_{t}=\left(y_{t,1}\ldots,y_{t,J}\right)\in\{0,1\}^{J}$.
The unknown parameter is $\theta=(\beta,\tau\text{)}\in\mathbb{R}{}^{d}$
where $d=9$. The observations are
assumed conditionally independent given the random effects and are
modelled through
\begin{equation*}
	g(y_t \mid x_t, \theta) = \prod_{j=1}^{J}\frac{\exp(y_{t,j}\eta_{t,j})}{1+\exp(\eta_{t,j})}, \quad f(x_t \mid \theta) = \varphi(x_t;0,\tau), \quad t=1,\ldots, T.
\end{equation*}
Inference in mixed effects models often aims at finding the population
effects and thus one is interested in integrating out the random effects.
Since the marginal likelihood contains intractable integrals, this
model lends itself to the pseudo-marginal approach. We obtain an unbiased estimator of the marginal likelihood by estimating the integrals using an importance sampling estimator
\begin{equation*}
	h(u \mid y_t, \theta) = \varphi(u; \widehat{x}_t, \tau_q^2), \quad \widehat{x}_t = \arg\max_{x_t} g(y_t\mid x_t, \theta)f(x_t \mid \theta)
\end{equation*}
with proposal variance $\tau_q>0$. We provide more details to importance sampling for mixed effects models in Section S3$\cdot$3 where we also show that Assumption 4 is satisfied in the present example.
For the covariate parameters we assume a diffuse Gaussian prior and
the variance of the random effects are assigned an inverse gamma prior.
We run a pseudo-marginal algorithm with a Gaussian random walk proposal for 500000 iterations. The covariance of
the proposal is set equal to the posterior covariance of the parameters
estimated in a preliminary run and scaled by $\ell^{2}/d=(2$$\cdot$$2)^{2}/9$.
We compare the average integrated autocorrelation time and the acceptance
rate with that of the limiting chain using the same $\ell=2$$\cdot$2 and $\sigma=\hat{\sigma}$, the average being defined as $\widehat{\textsc{iat}}(P_T^\omega) = \sum_{i=1}^d \textsc{iat}(f_i, P_T^\omega)$ for $f_i\left(\theta,z\right)=\theta_{i}$ the $i^{\textrm{th}}$ parameter component.
Here, $\hat{\sigma}$ is the standard deviation of the log-likelihood estimator obtained using 10000 samples of the marginal likelihood evaluated at an estimate $\hat{\theta}$  of the posterior mean. The results are summarized in Table
\ref{tab:comparison3}. For a given number of particles $N$ we report
the associated estimate of the noise in the log-likelihood estimator, the average integrated autocorrelation time averaged and the average acceptance rate.

\begin{table}
\centering

\begin{tabular}{rrrrrr}
Particles $N$ & $\hat{\sigma}$ & $\widehat{\textsc{iat}}$ & $\hat{\pr}_{\mathrm{acc}}$ & $\widehat{\textsc{iat}}\left(\tilde{P}_{\ell=2{\cdot}2,\sigma=\hat{\sigma}}\right)$ & $\hat{\pr}_{\mathrm{acc}}\left(\tilde{P}_{\ell=2{\cdot}2,\sigma=\hat{\sigma}}\right)$\tabularnewline
$12$ & 2$\cdot$00 & 140$\cdot$22  &  8$\cdot$93\% &  162$\cdot$57 & 7$\cdot$67\%\tabularnewline
$15$ & 1$\cdot$76 & 112$\cdot$06  & 10$\cdot$70\% &  121$\cdot$70 & 9$\cdot$93\%\tabularnewline
$18$ & 1$\cdot$63&  98$\cdot$69   & 12$\cdot$30\% &  94$\cdot$14  & 11$\cdot$73\%\tabularnewline
$21$ & 1$\cdot$46 &  72$\cdot$42  & 13$\cdot$93\% &  72$\cdot$31  & 14$\cdot$00\%\tabularnewline
$24$ & 1$\cdot$34 &  66$\cdot$29  & 15$\cdot$10\% &  64$\cdot$45  & 15$\cdot$55\%\tabularnewline
$27$ & 1$\cdot$29 &  61$\cdot$95  & 16$\cdot$08\% &  58$\cdot$08  & 16$\cdot$39\%\tabularnewline
$30$ & 1$\cdot$22 &  58$\cdot$70  & 16$\cdot$85\% &  54$\cdot$12  & 17$\cdot$52\%\tabularnewline
$33$ & 1$\cdot$16 &  52$\cdot$39  & 17$\cdot$77\% &  50$\cdot$26  & 18$\cdot$16\%\tabularnewline
\end{tabular}
\caption{For $N$ particles: standard deviation $\hat{\sigma}$ of the log-likelihood
estimator at the mean, average integrated autocorrelation time $\widehat{\textsc{iat}}$
and average acceptance probability $\hat{\pr}_{\mathrm{acc}}$ for pseudo-marginal
kernel and limiting kernel $\tilde{P}_{\ell,\hat{\sigma}}$ for $\ell=2$$\cdot$2.
}
\label{tab:comparison3}
\end{table}

\begin{figure}
\centering
\includegraphics[width=\linewidth]{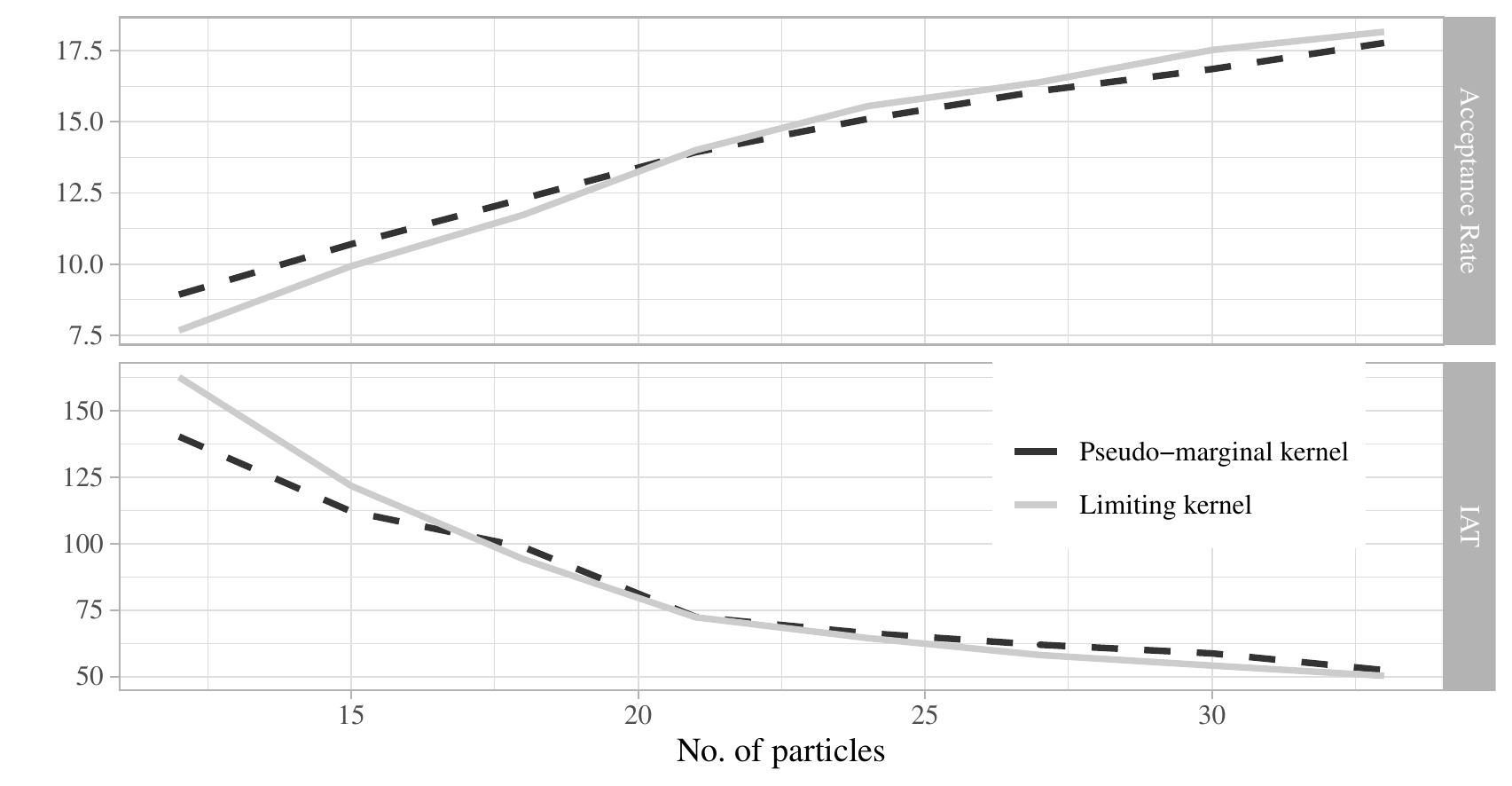}
\caption{Average integrated autocorrelation time (top) and average acceptance rate (bottom) for the pseudo-marginal algorithm as a function of $N$ and the limiting transition kernel $\tilde{P}_{\ell,\hat{\sigma}}$ for $\ell=2{\cdot}2$.}
\label{fig:comparison}
\end{figure}

The average integrated autocorrelation time and the acceptance rate are very close to those of the limiting algorithm. This is visualized in Figure \ref{fig:comparison} where we plot the same quantities against the number of particles $N$. The computing time of the pseudo-marginal algorithm targeting the posterior, $\widehat{\textsc{ct}}(P_T^\omega) = \widehat{\textsc{iat}}(P_T^\omega)/\hat{\sigma}^{2}$, and the computing time of the limiting algorithm, $\widehat{\textsc{ct}}(\tilde{P}_{\ell,\sigma})$, are both optimized for $\hat{\sigma}=1{\cdot}46$, as expected from Table \ref{tab:multi_dim}. In this example, the limiting kernel captures very well the behaviour of the pseudo-marginal algorithm for large data sets and  Table \ref{tab:multi_dim} thus provides useful guidelines on how to tune this scheme.

\section*{Acknowledgement}
Sebastian M. Schmon's research is supported by the Engineering and Physical
Sciences Research Council (EPSRC) grant EP/K503113/1 and Arnaud Doucet's
research is partially supported by the EPSRC grants EP/R018561/1 and EP/R034710/1.

\section*{Supplementary material}
Supplementary material below includes the proofs to all propositions and theorems as well as a set of generalized linear mixed models for which all assumptions hold. It also includes a short review of weak convergence of random measures and some further simulation studies, including a 3-dimensional Lotka-Volterra model.

\printbibliography[heading=subbibliography] % print section bibliography
\end{refsection}

\begin{refsection}

\appendix

\section*{Supplementary Material}

This supplementary material contains the proofs to all theorems and propositions, some background material and additional simulation studies. Section~\ref{appendix:randommeasures} includes a brief survey of weak convergence results for random probability measures on Polish spaces which play an important role in this article. We have not been able to find some of the precise statements we require in the literature so we present their proofs here without any claim of originality. Sections~\ref{sec:prof_randomchain} and \ref{app:clt} provide the proofs for sections 4 and 5, respectively. Finally, section~\ref{sec:furthersim} includes some additional numerical examples: a toy example and a Lotka-Volterra model where the likelihood is estimated using a particle filter as opposed to importance sampling.

\renewcommand{\thesection}{S1}
\section{Random Measures and Weak Convergence on Polish Spaces\label{appendix:randommeasures} }

\subsection{Weak Convergence}
Let $S$ be a Polish space, endowed with the Borel $\sigma$-algebra
$\mathcal{B}\left(S\right)$. We denote $d$ the metric inducing the
topology on $S$ and $\mathcal{P}(S)$ the space of Borel probability
measures on $S$. In the following, we will only consider
(random) probability measures in $\mathcal{P}(S)$ unless stated otherwise.

\begin{definition}[\emph{Weak convergence}]
A sequence of probability measures $(\mu_{n})_{n\geqslant1} $
converges weakly to a probability measure $\mu$, denoted $\mu_{n}\rightsquigarrow\mu$,
if for all $f\in C_{b}(S)$
%\addtocounter{equation}{9}
\begin{equation}
\mu_{n}(f)\rightarrow\mu(f)\quad\text{as }n\rightarrow\infty,\label{eq:weak}
\end{equation}
where $C_{b}(S)$ is the set of bounded continuous real-valued functions
of domain $S$.
\end{definition}

The set of test functions generating this topology can be restricted
to bounded continuous functions $f\colon S\rightarrow[0,1]$ or bounded Lipschitz functions, see for example \textcite[Lemma A.1 and Theorem A.2]{crauel2003}.
The topology of weak convergence can be metrized using the bounded
Lipschitz metric which is given for $\mu,\nu\in\mathcal{P}(S)$ by
\begin{equation}
d_{\mathrm{BL}}(\mu,\nu)=\sup\left\{ \left|\mu(f)-\nu\left(f\right)\right|;f\in\mathrm{BL}(S),\|f\|_{\mathrm{BL}}\leq1\right\} ,\label{eq:metrizeweakconvergence}
\end{equation}
see for example \textcite[Proposition 11.3.2]{dudley2002}. Here, the set $\mathrm{BL}(S)$ denotes the set of bounded Lipschitz functions and we follow \textcite{pollard2002user} by defining the norm
\begin{align}\label{def:BL}
\|f\|_{\mathrm{BL}} & =\max\left\{\|f\|_\mathrm{L}, 2\|f\|_\infty\right\}, \\
\intertext{where}
\|f\|_\mathrm{L} = \sup_{x,y:x\neq y}\frac{|f(x)-f(y)|}{d(x,y)} & \quad \text{and} \quad \|f\|_\infty = \sup_{x}|f(x)|.
\end{align}
This definition gives us the inequality
\begin{equation}\label{ineq:lipschitz}
	\left|f(x) - f(y)\right| \leq \|f\|_\mathrm{BL}\left[\min\left\{1, d(x,y) \right\} \right]
\end{equation}
for every $x, y$.

\subsection{Weak Convergence of Random Measures}
\label{sec:weak_random_measures}

We recall here some facts about random probability measures. Let $(\Omega,\mathcal{F},\mathbb{P})$
denote a probability space. We equip the product space $\Omega\times S$
with the product $\sigma$-algebra, $\mathcal{F}\otimes\mathcal{B}(S)$.
\begin{definition}[\emph{Random probability measure}]
\emph{} A \emph{random probability measure} is a map $\mu\colon\Omega\times\mathcal{B}\left(S\right)\rightarrow[0,1]$
such that for every $B\in\mathcal{B}\left(S\right)$ the map $\omega\mapsto\mu(\omega,B)=\mu^{\omega}(B)$
is measurable while $\mu(\omega,\cdot)\in\mathcal{P}(S)$ for almost
every $\omega\in\Omega$.
\end{definition}

For all bounded and measurable functions $g\colon\Omega\times S\rightarrow\mathbb{R}$,
the assignment $\omega\mapsto\int_{S}g(\omega,x)\mu^{\omega}(\mathrm{d}x)$
is measurable \autocite[see, for example,][Proposition 3.3]{crauel2003}
and thus, for random measures, the map $\omega\mapsto\mu^{\omega}(f)$
is a random variable. As a consequence we have that $\mu^{\omega}\colon\Omega\rightarrow\mathcal{P}(S)$
is a Borel measurable map. Conversely,
it can be shown that any random element of $\left[\mathcal{P}(S),\mathcal{B}\{\mathcal{P}(S)\}\right]$
fulfils the condition set out in Definition \ref{def:random_measure},
see \autocite[Remark 3.20 (i)]{crauel2003} or \autocite[Lemma 1.37]{Kallenberg2006}
for details.%
\begin{definition}[\emph{Weak convergence of random measures}]
A sequence of random probability measures $(\mu_{n}^{\omega})_{n\geqslant1}$ converges weakly almost surely to a probability measure $\mu$,
denoted $\mu_{n}^{\omega}\rightsquigarrow_{a.s.}\mu$, if
\begin{equation}
\mathbb{P}\left(\omega\in\Omega:\,\,\mu_{n}^{\omega}\rightsquigarrow\mu\right)=1.\label{eq:wealconv}
\end{equation}
Further, we say that $(\mu_{n}^{\omega})_{n\geqslant1}$ converges
weakly in probability, denoted $\mu_{n}^{\omega}\rightsquigarrow_{\mathbb{P}}\mu$,
if every subsequence contains a further subsequence which converges
weakly almost surely.
\end{definition}

One can easily verify that the above definition of almost sure weak
convergence, respectively weak convergence in probability, is equivalent
to $\rho(\mu_{n}^{\omega},\mu)\rightarrow0$ almost surely, respectively
in probability, for some metric $\rho$ on $\mathcal{\mathcal{P}}(S)$
metrizing weak convergence, e.g., the bounded Lipschitz metric (\ref{eq:metrizeweakconvergence}),
see for example Theorem \ref{thm:EasycheckWeakCV}.
\begin{remark}[\emph{Measurability of probability metric}]
\emph{ As already mentioned above, for any random measure the map
$\omega\mapsto\mu^{\omega}$ is measurable with respect to the Borel
$\sigma$-algebra $\mathcal{B}\left\{\mathcal{P}(S)\right\}$. Moreover,
any metric $\rho$ inducing the weak topology on $\mathcal{P}(S)$
is trivially continuous in its first argument and hence the map $\mu^{\omega}\mapsto\rho(\mu^{\omega},\nu)$
for some fixed measure $\nu$ is measurable with respect to the Borel
$\sigma$-algebra $\mathcal{B}(\mathbb{R})$. This implies (Borel) measurability of the map $\omega\mapsto\rho(\mu^{\omega},\nu)$ for a non-random measure $\nu.$ }%
\end{remark}

In light of the definition of weak convergence (\ref{eq:weak}) it
is natural to ask whether almost sure weak convergence holds if
\begin{equation}
\mu_{n}^{\omega}(f)\overset{\mathrm{a.s.}}{\longrightarrow}\mu(f)\qquad\text{for all}\qquad f\in C_{b}(S),\label{eq:pointwiseweakAS}
\end{equation}
and similarly whether weak convergence in probability holds if
\begin{equation}
\mu_{n}^{\omega}(f)\overset{\mathbb{P}}{\longrightarrow}\mu(f)\qquad\text{for all}\qquad f\in C_{b}(S).\label{eq:pointwiseweakProba}
\end{equation}
In many practical applications, it appears easier to check (\ref{eq:pointwiseweakAS})
rather than (\ref{eq:wealconv}), similarly checking (\ref{eq:pointwiseweakProba})
appears easier than having to check that every subsequence of $(\mu_{n}^{\omega})_{n\geqslant1}$
contains a subsequence which converges weakly almost surely. Relating
those statements is inconvenienced by the fact that weak convergence
is usually checked using an uncountable convergence determining class
of functions, e.g., the space of bounded continuous functions. However,
we show here that these equivalences hold true for Polish spaces; see Theorem
\ref{thm:EasycheckWeakCV} below.

Almost sure weak convergence can be shown using the existence of a
countable convergence determining subclass $\mathcal{C}\subset\mathrm{BL}(S)\subset C_{b}(S)$.
Considering subsequences and using a diagonal argument we can show
the equivalence of the statement also holds if almost sure convergence
is replaced by convergence in probability. For the purposes of this
paper we confine our attention to weak convergence in probability.
To prove the statements above we first need an auxiliary result, which
also appeared in \textcite[Lemma 4]{sweeting1989}.
\begin{proposition}
\label{prop:subsequence} Suppose $A$ is a countable set and consider
random variables $X_{n}(a)\colon\Omega\rightarrow\mathbb{R}$ indexed
by $a\in A$ and $n\in \mathbb{N}$. Moreover, assume that for every $a\in A$ the sequence
$\{X_{n}(a)\}_{n\geqslant1}$ converges to $X(a)$ in probability, i.e.,
\[
X_{n}(a)\overset{\mathbb{P}}{\rightarrow}X(a)\quad\forall a\in A.
\]
Then there exists a subsequence $N'\subset\mathbb{N}$ such that along
$N'$
\[
\mathbb{P}\left\{\omega:X_{n}(a)\rightarrow X(a)\quad\forall a\in A\right\}=1.
\]
\end{proposition}

\begin{proof}
Choose $a_{1}\in A$. Since we have $X_{n}(a_{1})\overset{\mathbb{P}}{\rightarrow}X(a_{1})$
we can extract a subsequence $n_{1,1},n_{1,2},\ldots$ such that
\[
\left\{X_{n_{1,1}}(a_{1}),X_{n_{1,2}}(a_{1}),X_{n_{1,3}}(a_{1}),\ldots\right\}
\]
converges almost surely. Pick now $a_{2}\in A$, we can now extract
a further subsequence
\[
\left\{X_{n_{2,1}}(a_{2}),X_{n_{2,2}}(a_{2}),X_{n_{2,3}}(a_{2}),\ldots\right\}
\]
along which we have almost sure convergence. We can iterate this procedure
to get another subsequence
\[
\left\{X_{n_{3,1}}(a_{3}),X_{n_{3,2}}(a_{3}),X_{n_{3,3}}(a_{3}),\ldots\right\}.
\]
Along the subsequence $N'=\left(n_{1,1},n_{2,2},n_{3,3},...\right)$,
we have almost sure convergence of $X_{n'}(a)\rightarrow X(a)$ for
all $a\in A$.
\end{proof}
The existence of a countable convergence determining class for Polish
spaces is guaranteed by the following Proposition. The proof is adapted
from \textcite[Theorem 2.2]{bertipratellirigo2006}.
\begin{proposition}
\label{prop:countsubclass} Consider $\mathcal{P}(S)$ equipped with
the Borel $\sigma$-algebra generated by the topology of weak convergence.
There exists a countable convergence determining subclass $\mathcal{C}\subset\mathrm{BL}(S)$.
\end{proposition}

\begin{proof}
Take a countable set $\{s_{1},s_{2},\ldots\}$ dense in $S$ and let
$H=[0,1]^{\mathbb{N}}$ be the Hilbert cube. For $x\in S$, define
the map $h\colon S\rightarrow H$ by

\[
h(x)=\left\{d(x,s_{1})\wedge1,d(x,s_{2})\wedge1,\ldots\right\}.
\]

We can equip $H$ with the topology of coordinate wise convergence.
Writing $u=(u_{1,}u_{2,}\ldots)$ and $v=(v_{1,}v_{2,}\ldots)$ for
elements $u,v\in H$, this topology is induced by the metric
\[
\alpha(u,v)=\sum_{i=1}^{\infty}\frac{|u_{i}-v_{i}|}{2^{i}}.
\]

The Hilbert cube $H$ is compact by Tychonoff's Theorem \autocite[see for example][Theorem 2.2.8.]{dudley2002},
$h$ is a homeomorphism from $S$ to $h(S)$ \autocite[Theorem A.1.1.]{borkar1991}
and its closure $\overline{h(S)}\subset H$ is compact. For $\mu\in\mathcal{P}(S)$
denote $\nu=\mu\circ h^{-1}$ the image measure on $h(S)$.

Note that any Lipschitz continuous function on $h(S)$ can be extended
to $\overline{h(S)}$ without increasing its norm \autocite[Proposition 11.2.3.]{dudley2002}.
By the Arzel\`a--Ascoli theorem, the sets $B_{n}=[f\in\mathrm{BL}\{\overline{h(S)}\}:\:\|f\|_\mathrm{BL}\leq n]$
are compact and thus separable under the $\|\cdot\|_{\infty}$-norm.
Therefore $\mathrm{BL}\{\overline{h(S)}\}=\bigcup_{n=1}^{\infty}B_{n}$
is separable under the $\|\cdot\|_{\infty}$-norm and so is $\mathrm{BL}\{h(S)\}$.
Hence, we can pick a countable set $\mathcal{D}$ which is dense in
$\mathrm{BL}\{h(S)\}$. Defining $\mathcal{C}=\{g\circ h:\:g\in\mathcal{D}\}$
we have $\mathcal{C}\subset\mathrm{BL}(S)$ since for all $x,y\in S$
and $i\in\mathbb{N}$
\[
|d(x,s_{i})\wedge1-d(y,s_{i})\wedge1|\leq d(x,y)
\]
and thus

\[
|g\circ h(x)-g\circ h(y)|\leq L_{g}\alpha\{h(x),h(y)\}=L_{g}\sum_{i=1}^{\infty}\frac{|d(x,s_{i})\wedge1-d(y,s_{i})\wedge1|}{2^{i}}\leq L_{g}d(x,y),
\]

where $L_{g}$ denotes the Lipschitz constant of the function $g.$

Now assume that $\mu_{n}(f)\rightarrow\mu(f)$ for all $f\in\mathcal{C}$.
Then by a change of variable
\[
\int_{S}f\,d\mu_{n}=\int_{S}g\circ h\:d\mu_{n}=\int_{h(S)}g\:d\nu_{n}\rightarrow\int_{h(S)}g\:d\nu
\]
for all $g\in\mathcal{D}.$ Since $\mathcal{D}$ is dense in $\mathrm{BL}\{h(S)\}$
with respect to the $\|\cdot\|_{\infty}$-norm we have convergence
for all bounded Lipschitz functions and thus $\nu_{n}\rightsquigarrow\nu.$
By continuity of $h^{-1}$ we also have convergence $\mu_{n}\rightsquigarrow\mu.$
\end{proof}
Equipped with these results we can now prove some equivalences which
facilitate the verification of weak convergence of random probability
measures in the sense introduced above. We will prove the following
statements only for convergence in probability. The modifications
for almost sure convergence are obvious.
\begin{theorem}
\label{thm:EasycheckWeakCV}Let $\left(\mu_{n}^{\omega}\right)_{n\geqslant1}$ be a sequence of random probability measures and $\mu$ a probability
measure. Then the following statements are equivalent
\begin{itemize}
\item[(i)]  $d_{\mathrm{BL}}(\mu_{n}^{\omega},\mu)\overset{\mathbb{P}}{\rightarrow}0,$
\item[(ii)]  $\mu_{n}^{\omega}\rightsquigarrow_{\mathbb{P}}\mu$
\item[(iii)] $\mu_{n}^{\omega}(f)\overset{\mathbb{P}}{\longrightarrow}\mu(f)\qquad\text{for all}\qquad f\in C_{b}(S)$
\item[(iv)] $\mu_{n}^{\omega}(f)\overset{\mathbb{P}}{\longrightarrow}\mu(f)\qquad\text{for all}\qquad f\in\mathrm{BL}(S)$.
\end{itemize}
The same results hold if convergence in probability is replaced by
almost sure convergence throughout.
\end{theorem}

\begin{proof}
The equivalence $(i)\Leftrightarrow(ii)$ is immediate since $d_{\mathrm{BL}}$
metrizes weak convergence. The implications $(ii)\Rightarrow(iii)\Rightarrow(iv)$ are trivial.
To show $(iv)\Rightarrow(ii),$ note that by Proposition \ref{prop:countsubclass}
there exists a countable convergence determining subclass $\mathcal{C}\subset\mathrm{BL}(S)$.
By virtue of Proposition \ref{prop:subsequence} there exists a subsequence
$(n_{1},n_{2},\ldots)$ such that for all $g\in\mathcal{C}$
\[
\mu_{n_{k}}^{\omega}(g)\overset{\mathrm{a.s.}}{\longrightarrow}\mu(g)\qquad\text{as }k\rightarrow\infty.
\]
Now, given $(n_{k})_{k\in\mathbb{N}}$ define
\[
A(g)=\left\{ \omega\in\Omega:\mu_{n_{k}}^{\omega}(g)\longrightarrow\mu(g)\quad\text{as}\quad k\rightarrow\infty\right\} .
\]
We have $\mathbb{P}\{A(g)\}=1$ for all $g\in\mathcal{C}$ and for $\bigcap_{g\in\mathcal{C}}A(g)=A\in\mathcal{B}(S)$
we find $\mathbb{P}(A)=1$. Since we can apply this reasoning to any
subsequence we always find a further subsequence such that $(\mu_{n_{k_{j}}}^{\omega})$
converges almost surely. See also \textcite[Theorem 9]{sweeting1989}
and \textcite[Theorem 2.2]{bertipratellirigo2006}.
\end{proof}

\begin{remark}
\label{rem:conditional_weak}If the random measure is induced by a
regular conditional distribution, i.e., let $\left(\mu_{n}^{\omega}\right)_{n\geqslant1}$ denote a sequence of transition kernels such that
\[
\mu_{n}^{\omega}(\cdot)=\mathbb{P}(X_{n}\in\cdot\mid\mathscr{\mathcal{F}}_{n})(\omega)\qquad\mathbb{P}-a.s.
\]
for some filtration $\left(\mathscr{\mathcal{F}}_{n}\right)_{n\geqslant 1}$, we
have
\[
\int f(x)\mu_{n}^{\omega}(\mathrm{d}x)=E\left\{f(X_{n})\mid\mathcal{F}_{n}\right\}(\omega)\qquad\mathbb{P}-a.s.
\]
and thus equivalently to $\mu_{n}\rightsquigarrow_{\mathbb{P}}\mu$
then we can write
\begin{equation}
E\left\{f(X_{n})\mid\mathcal{F}_{n}\right\}\overset{\mathbb{P}}{\longrightarrow}E\left\{f(X)\right\},\label{eq:cond_conv}
\end{equation}
where $X\sim\mu$. For brevity we will also use the notation $X_{n}\mid\mathcal{F}_{n}\rightsquigarrow_{\mathbb{P}}\mu$
instead of (\ref{eq:cond_conv}).%
\end{remark}

\subsection{Product Spaces}

We address here the setting where the spaces are of the form $S^{k}=S\times S\times\cdots\times S$
or $S^{\mathbb{N}}=S\times S\times\ldots$. We will equip these product
spaces with the product topology and the respective Borel $\sigma$-algebra.
The following lemma is helpful to characterize weak convergence in
probability in this context.
\begin{lemma}
\label{lem:finfitedim} For fixed $k$, let $(\mu_{n}^{\omega})_{n\geqslant 1}$
denote random measures on $S^{k}$ and $\mu$ a non-random measure
on $S^{k}$. Then the following are equivalent
\end{lemma}

\begin{itemize}
\item[\emph{(i)}]
\[
\mu_{n}^{\omega}\rightsquigarrow_{\mathbb{P}}\mu,
\]
\item[\emph{(ii)}]
\[
\mu_{n}^{\omega}(f)\overset{\mathbb{P}}{\rightarrow}\mu(f)
\]
for all $f\in C_{b}(S^{k})$.
\item[\emph{(iii)}]
\[
\int_{S^{k}}\prod_{i=1}^{k}f_{i}(x_{i})\mu_{n}^{\omega}(\mathrm{d}x_{1}\ldots\mathrm{d}x_{k})\overset{\mathbb{P}}{\rightarrow}\int_{S^{k}}\prod_{i=1}^{k}f_{i}(x_{i})\mu(\mathrm{d}x_{1}\ldots\mathrm{d}x_{k})
\]
for all $f_{1},\ldots f_{k}\in C_{b}(S)$.
\item[\emph{(iv)}]
\[
\int_{S^{k}}\prod_{i=1}^{k}f_{i}(x_{i})\mu_{n}^{\omega}(\mathrm{d}x_{1}\ldots\mathrm{d}x_{k})\overset{\mathbb{P}}{\rightarrow}\int_{S^{k}}\prod_{i=1}^{k}f_{i}(x_{i})\mu(\mathrm{d}x_{1}\ldots\mathrm{d}x_{k})
\]
for all $f_{1},\ldots f_{p}\in\mathrm{BL}(S)$.
\end{itemize}

\begin{proof}
The implications $(i)\Rightarrow(ii)\Rightarrow(iii)\Rightarrow(iv)$
are trivial. Thus, we only need to show $(iv)\Rightarrow(i)$. We
now by Proposition \ref{prop:countsubclass} that there exists a countable
convergence determining class $\mathcal{C}\subset\mathrm{BL}(S)$,
so we can assume $f_{1},f_{2},\ldots\in\mathcal{C}$. Without loss
of generality we can assume $\|f_{i}\|_{\infty}\leq1$ for all $i$
and $1\in\mathcal{C}$. Then we have that for every $i\in\{1,\ldots,k\}$
the marginal of the $i$th coordinate, denoted $\mu_{n,i}^{\omega}$,
converges to $\mu_{i}$ weakly in probability, i.e. for all $i$ and
all $f_{i}\in\mathcal{C}$ we have
\[
\int_{S}f_{i}(x)\mu_{n,i}^{\omega}(\mathrm{d}x_{i})\overset{\mathbb{P}}{\rightarrow}\int_{S}f_{i}(x)\mu_{i}(\mathrm{d}x_{i}).
\]
Now by Proposition \ref{prop:subsequence} for every $i\in\{1,\ldots,k\}$
every subsequence $N\subset\mathbb{N}$ contains a further subsequence
$N'\subset N$ such that we have convergence almost sure convergence
for all $g\in\mathcal{C}$, i.e. denoting
\[
A_{i}:=\left\{ \omega\in\Omega:\int_{S}g(x_{i})\mu_{n',i}^{\omega}(\mathrm{d}x_{i})\longrightarrow\int_{S}g(x_{i})\mu_{i}(\mathrm{d}x_{i})\quad\text{for all }g\in\mathcal{C}\right\}
\]
we have $P(A_{i})=1$. We can extract a further subsequence $N''\subset N'$
such that along $N''$ we have convergence almost surely for all $i$
and all $g$ and thus for $\omega\in A:=\cap_{i=1}^{k}A_{i}$ the
sequence $\left\{ \mu_{n}^{\omega};n\in N''\right\} $ is tight, since
$\left\{ \mu_{n,i}^{\omega};n\in N''\right\} $ is tight for every
$i$ \autocite[see][Chapter 3 Proposition 2.4.]{ethier2005}. We can
conclude that for every such $\omega$ every subsequence of $(\mu_{n}^{\omega})_{n\geqslant1}$
has a further subsequence that converges. It remains to show that
the functions of the form $\prod_{i=1}^{k}f_{i}$ are measure determining.
However, by \textcite[Chapter 2 Proposition 4.6.]{ethier2005} if $\mathcal{C}$
is measure determining on $S$ then so is the product for $S^{k}$.
\end{proof}

If $S=\mathbb{R}^{k}$ for some $k\in\mathbb{N}$ we can check weak
convergence in probability by considering moment generating functions.
The following result is shown by \textcite[Corollary 3]{sweeting1989};
see also \textcite[Lemma 1]{castillo2015}.
\begin{proposition}
\label{prop:weakconv_moment}Let $\left(\mu_{n}^{\omega}\right)_{n\geqslant1}$
be a sequence of random probability measures and assume there exists
$u_{0}>0$ such that for all $n\in\mathbb{N}$ the moment generating
functions
\[
m_{n}(u,\omega)=\int\exp\left(u^{\mathtt{\mathsf{T}}}x\right)\mu_{n}^{\omega}(\mathrm{d}x)
\]
exist for $|u|<u_{0}$ then $\mu_{n}^{\omega}\rightsquigarrow_{\mathbb{P}}\mu$
if and only if for every $u\in\mathbb{R}^{k}$
\[
m_{n}(u,\cdot)\overset{\mathbb{P}}{\longrightarrow}m(u,\cdot)=\int\exp\left(u^{\mathtt{\mathsf{T}}}x\right)\mu^{\omega}(\mathrm{d}x).
\]
\end{proposition}

\begin{proof}
This can be seen by considering the class of functions of the form
$f_{u}(x)=\exp(u^{\mathtt{\T}}x)$ for $u\in\mathbb{Q},\quad|u|<u_{0}$
and showing that they form a countable convergence determining class,
see \textcite[Corollary 3]{sweeting1989}. Consider the case $k=1$ and
a sequence of measures $(\mu_{n})_{n\geqslant1} $ and $\mu$
such that
\[
m_{n}(u)=\int e^{ux}\mu_{n}(\mathrm{d}x)\rightarrow m(u)=\int e^{ux}\mu(\mathrm{d}x).
\]

Denote a compact set $K=[-c,c]$. Then by the Markov inequality
\[
\mu_{n}\big(K^{\complement}\big)=\int_{|x|\geq c}\mu_{n}(dx)\leq\frac{m_{n}(u_{0})}{e^{u_{0}c}}
\]
and $m_{n}(u_{0})\rightarrow m(u_{0}).$ Hence, $\mu_{n}(K^{\complement})$
is bounded and we can find $c$ such that $\sup_{n}\mu_{n}(K^{\complement})<\epsilon$
and $(\mu_{n})_{n\geqslant1} $ is tight. By continuity the
$f_{u}$ are measure determining so we can conclude that the limit
is unique. For $k>1$ we can use the same argument to show that the
marginals are tight, see the proof of Lemma \ref{lem:finfitedim}.
\end{proof}
Lemma \ref{lem:finfitedim} can be readily extended to countably
infinite product spaces by considering convergence of the finite dimensional
distribution. Let us therefore denote $\mu\circ\pi_{k}^{-1}\colon S^{\mathbb{N}}\rightarrow S^{k};k\in\mathbb{N}$
the canonical projections. For non-random measures, it is well-known
that convergence of the projections already implies convergence on
the whole of $S^{\mathbb{N}}$ \autocite[Example 2.6]{billingsley1999}.
Since there are countably many such projections, we can apply the
reasoning of Proposition \ref{prop:subsequence} to conclude that for checking
$\mu_{n}^{\omega}\rightsquigarrow_{\mathbb{P}}\mu$ on $S^{\mathbb{N}}$
we just need to show
\[
\int_{S^{k}}\prod_{i=1}^{k}f_{i}(x_{i})\mu_{n}^{\omega}(\mathrm{d}x_{1},\ldots,\mathrm{d}x_{k})\overset{\mathbb{P}}{\rightarrow}\int_{S^{k}}\prod_{i=1}^{k}f_{i}(x_{i})\mu(\mathrm{d}x_{1},\ldots,\mathrm{d}x_{k})
\]
for all $f_{1},\ldots f_{k}\in\mathrm{BL}(S)$ and $k\in\mathbb{N}$.
The following Lemma is essentially a version of \textcite[Chapter 3 Proposition 4.6 b]{ethier2005}
extended to random measures.
\begin{lemma}
\label{lem:infinitedist}Let $\left(\mu_{n}^{\omega}\right)_{n\geqslant1}$
be a sequence of random probability measures and $\mu$ a non-random
probability measure on $S^{\mathbb{N}}$. Then $\mu_{n}^{\omega}\rightsquigarrow_{\mathbb{P}}\mu$
is equivalent to
\[
\int_{S^{k}}\prod_{i=1}^{k}f_{i}(x_{i})\mu_{n}^{\omega}(\mathrm{d}x_{1},\ldots,\mathrm{d}x_{k})\overset{\mathbb{P}}{\rightarrow}\int_{S^{k}}\prod_{i=1}^{k}f_{i}(x_{i})\mu(\mathrm{d}x_{1},\ldots,\mathrm{d}x_{k})
\]
for all $f_{1},\ldots f_{k}\in\mathrm{BL}(S)$ and $k\in\mathbb{N}$.
\end{lemma}

\begin{proof}
Suppose for any $k$ that the above convergence holds for all test
functions $f_{1},\ldots f_{k}\in\mathrm{BL}(S)$. We have shown in
Lemma \ref{lem:finfitedim} that this is equivalent of convergence
of the canonical projections $\mu_{n}^{\omega}\circ\pi_{k}^{-1}$
on $S^{k}$ (in probability) for any given $k$. Hence, using Proposition
\ref{prop:subsequence} for every subsequence $N\subset\mathbb{N}$
there is a subsequence $N'\subset N$ such that along $N'$
\[
\mathbb{P}\left(\omega\in\Omega:\mu_{n}^{\omega}\circ\pi_{k}^{-1}\rightsquigarrow\mu\circ\pi_{k}^{-1}\text{ as }n\rightarrow\infty\quad\text{for all }k\in\mathbb{N}\right)=1.
\]
An application of \textcite[Chapter 3 Proposition 4.6 b]{ethier2005}
concludes the proof.
\end{proof}

\renewcommand{\thesection}{S2}
\section{Proofs of Section 4}\label{sec:prof_randomchain}

\setcounter{lemma}{0}

\subsection{Proofs for Section 4.1}
\begin{lemma}%[\emph{Random Markov chain}]
\textup{ Given a random probability
measure} $\mu^{\omega}$\textup{ and random Markov kernel $K^{\omega}$,
there exists an almost surely unique random probability measure $\mu^{\mathbb{N},\omega}$
on $S^{\mathbb{N}}$ such that
\[
\mu^{\mathbb{N},\omega}(A_{1}\times\ldots\times A_{k}\times E_{k+1})=\int_{A_{1}}\mu^{\omega}(\mathrm{d}x_{1})\int_{A_{2}}K^{\omega}(x_{1},\mathrm{d}x_{2})\ldots\int_{A_{k}}K^{\omega}(x_{k-1},\mathrm{d}x_{k})
\]
for any }$A_{i}\in\mathcal{B}(S)$ $(i=1,\ldots,k)$, $k\in \mathbb{N}$ and $E_{k+1}=\boldsymbol{\times}_{i=k+1}^{\infty}S$.
\end{lemma}

\begin{proof}[Proof of Lemma \ref{lem:Random-Markov-chain}]
For $\mathbb{P}-$almost all $\omega$, the existence and uniqueness
of the distribution $\mu^{\mathbb{N},\omega}$ on $\{S^{\mathbb{N}},\mathcal{B}(S)^{\mathbb{N}}\}$
can be obtained using the Ionescu-Tulcea extension theorem; see, e.g.,
\textcite[Theorem 6.17]{Kallenberg2006} or \textcite[Theorem 14.32]{Klenke2013}.
Measurability follows analogously by noting that $\omega\mapsto\mu^{\mathbb{N}}(\omega,A)$
is measurable for any $A\in\mathcal{E}=\{A_{1}\times\ldots\times A_{k}\times E_{k+1};A_{i}\in\mathcal{B}(S),i=1,\ldots,k,k\in\mathbb{N}\}$ and that $\mathcal{E}$ forms a $\pi-$system that generates $\mathcal{B}(S)^{\mathbb{N}}.$
By \textcite[Remark 3.2]{crauel2003} this is enough to obtain measurability
for every $A\in\mathcal{B}(S)^{\mathbb{N}}$.
\end{proof}

\begin{theorem}%[Weak convergence of random Markov chains]
If the following assumptions hold,
\begin{itemize}
\item[(T.1)]\label{t1} the random probability measures $\left(\mu_{n}^{\omega}\right)_{n\geqslant1}$ converge weakly in probability to a probability measure $\mu$ as $n\rightarrow \infty$,
\item[(T.2)] the random Markov transition kernels $\left(K_n^\omega\right)_{n\geqslant1}$ satisfy
\begin{equation*}
\int\left|K_{n}^{\omega}f(x)-Kf(x)\right|\mu_{n}^{\omega}(\mathrm{d}x)\rightarrow0\label{eq:ass2}
\end{equation*}
in probability as $n\rightarrow \infty$ for all $f\in\mathrm{BL}(S)$ where $K$ is a Markov transition kernel
,
\item[(T.3)] the transition kernel $K$ is such that \textup{$x\mapsto Kf(x)$
is continuous for any $f\in C_{b}(S)$, }
\end{itemize}
then, as $n\rightarrow \infty$, the measures $(\mu_{n}^{\mathbb{N},\omega})_{n\geqslant1}$ on $S^\mathbb{N}$ converge weakly in probability to
the measure $\mu^\mathbb{N}$ induced by the Markov chain with initial distribution $\mu$ and transition kernel $K$.
\end{theorem}

\begin{proof}[Proof of Theorem \ref{theo:appendix}]
By Section \ref{sec:weak_random_measures} Lemma \ref{lem:infinitedist}, we need to show that for any $k\geq0$
and any $f_{0},\ldots,f_{k}\in\mathcal{\mathrm{BL}}(S)$
\begin{equation}
E^{\omega}\left\{f_{0}(X_{n,0}^{\omega})\cdots f_{k}(X_{n,k}^{\omega})\right\}\overset{\mathbb{P}}{\longrightarrow}E\left\{f_{0}(X_{0})\cdots f_{k}(X_{k})\right\}\label{eq:fdds}
\end{equation}
where $E^{\omega}$, resp. $E$, denotes the expectation
w.r.t. the law of $\mathbf{X}_{n}^{\omega}$, respectively w.r.t.
the law of $\mathbf{X}$.
We prove this by induction. For $k=0$,
this follows directly from ($T.1$). Now assume that (\ref{eq:fdds}) is true for $k\geq0$, i.e.
\[
\left|E^{\omega}\left\{f_{0}(X_{n,0}^{\omega})f_{1}(X_{n,1}^{\omega})\cdots f_{k}(X_{n,k}^{\omega})\right\} - E\left\{f_{0}(X_{0})f_{1}(X_{1})\cdots f_{k}(X_{k})\right\}\right|\overset{\mathbb{P}}{\longrightarrow}0.
\]
By Lemma \ref{lem:finfitedim} this is equivalent to weak convergence
in probability of the vector of the first $k$ states, i.e., for all
$f\in C_{b}(S^{k})$

\begin{equation}
E^{\omega}\left\{f(X_{0}^{n},\ldots,X_{k}^{n})\right\}\overset{\mathbb{P}}{\longrightarrow}E\left\{f(X_{0},\ldots,X_{k})\right\}.\label{eq:joint_convergence}
\end{equation}

For $k+1$, we have
\begin{align}
 & \left|E^{\omega}\left\{f_{0}(X_{n,0}^{\omega})\cdots f_{k}(X_{n,k}^{\omega})f_{k+1}(X_{n,k+1}^{\omega})\right\}-E\left\{f_{0}(X_{0})\cdots f(X_{k})f_{k+1}(X_{k+1})\right\}\right|\nonumber \\
 & =\left|{E^{\omega}}\left\{f_{0}(X_{n,0}^{\omega})\cdots f_{k}(X_{n,k}^{\omega})K_{n}^{\omega}f_{k+1}(X_{n,k}^{\omega})\right\}-E\left\{f_{0}(X_{0})\cdots f(X_{k})Kf_{k+1}(X_{k})\right\}\right|\nonumber \\
 & \leq\left|E^{\omega}\left\{f_{0}(X_{n,0}^{\omega})\cdots f_{k}(X_{n,k}^{\omega})K_{n}^{\omega}f_{k+1}(X_{n,k}^{\omega})-f_{0}(X_{n,0}^{\omega})\cdots f_{k}(X_{n,k}^{\omega})Kf_{k+1}(X_{n,k}^{\omega})\right\}\right|\nonumber \\
 & \qquad+\left|E^{\omega}\left\{f_{0}(X_{n,0}^{\omega})\cdots f_{k}(X_{n,k}^{\omega})Kf_{k+1}(X_{n,k}^{\omega})\right\}-E\left\{f_{0}(X_{0})\cdots f_{k}(X_{k})Kf_{k+1}(X_{k})\right\}\right|\nonumber \\
 & \leq E^{\omega}\left\{\left|K_{n}^{\omega}f_{k+1}(X_{n,k}^{\omega})-Kf_{k+1}(X_{n,k}^{\omega})\right|\right\}\label{eq:fin1}\\
 & \qquad+\left|E^{\omega}\left\{f_{0}(X_{n,0}^{\omega})\cdots f_{k}(X_{n,k}^{\omega})Kf_{k+1}(X_{n,k}^{\omega})\right\}-E\left\{f_{0}(X_{0})\cdots f_{k}(X_{k})Kf_{k+1}(X_{k})\right\}\right|.\label{eq:fin2}
\end{align}

The term (\ref{eq:fin1}) converges due to (\ref{eq:ass2}). For the
term (\ref{eq:fin2}), the function $Kf_{k+1}$ is bounded and it
is assumed continuous so the function $f_{0}\cdots f_{k}Kf_{k+1}\in C_{b}(S^{k})$.
Hence this term vanishes by (\ref{eq:joint_convergence}).
\end{proof}

\subsection{Some Auxiliary Results \label{sec:appendix}}
\begin{lemma}
\label{cor1} Under Assumption \ref{ass1}, we have
\[
\varphi(\mathrm{d}\theta;\hat{\theta}_{T}^{\omega},\Sigma/T)\rightsquigarrow_{\mathbb{P}^{Y}}\delta_{\bar{\theta}}(\mathrm{d}\theta)
\]
and
\[
\pi_{T}^{\omega}(\mathrm{d}\theta)\rightsquigarrow_{\mathbb{P}^{Y}}\delta_{\bar{\theta}}(\mathrm{d}\theta).
\]
\end{lemma}

\begin{proof}
Using the moment generating function of the normal distribution, we
have as $T\rightarrow\infty$
\[
\int e^{u^{\T}\theta}\varphi(\theta;\hat{\theta}_{T}^{\omega},\Sigma/T)\mathrm{d}\theta=\exp\left(u^\T\hat{\theta}_{T}^{\omega}+u^\T\Sigma u/{2T}\right)\overset{\mathbb{P}^{Y}}{\longrightarrow}\exp(u^\T\bar{\theta})=\int\exp(u^\T\theta)\delta_{\bar{\theta}}(\mathrm{d}\theta),
\]
where $\delta_{\bar{\theta}}$ denotes the Dirac measure at $\bar{\theta}$
and thus $\varphi(\mathrm{d}\theta;\hat{\theta}_{T}^{\omega},\Sigma/T)\rightsquigarrow_{\mathbb{P}^{Y}}\delta_{\bar{\theta}}(\mathrm{d}\theta)$
by Proposition \ref{prop:weakconv_moment}. This implies that for
$f\in C_{b}(\mathbb{R}^{d})$
\begin{align*}
 & \left|\int f(\theta)\pi_{T}^{\omega}(\theta)\mathrm{d}\theta-\int f(\theta)\delta_{\bar{\theta}}(\mathrm{d}\theta)\right|\\
 & \leq\left|\int f(\theta)\pi_{T}^{\omega}(\theta)\mathrm{d}\theta-\int f(\theta)\varphi(\theta;\hat{\theta}_{T}^{\omega},\Sigma/T)\mathrm{d}\theta\right|+\left|\int f(\theta)\varphi(\theta;\hat{\theta}_{T}^{\omega},\Sigma/T)\mathrm{d}\theta-\int f(\theta)\delta_{\bar{\theta}}(\mathrm{d}\theta)\right|\\
 & \leq\|f\|_{\infty}\int\left|\pi_{T}^{\omega}(\theta)-\varphi(\theta;\hat{\theta}_{T}^{\omega},\Sigma/T)\right|\mathrm{d}\theta+\left|\int f(\theta)\varphi(\theta;\hat{\theta}_{T}^{\omega},\Sigma/T)\mathrm{d}\theta-\int f(\theta)\delta_{\bar{\theta}}(\mathrm{d}\theta)\right|,
\end{align*}
where the first term on the r.h.s. converges to zero in probability
under Assumption \ref{ass1} while the second term converges to zero
as $\varphi(\mathrm{d}\theta;\hat{\theta}_{T}^{\omega},\Sigma/T)\rightsquigarrow_{\mathbb{P}^{Y}}\delta_{\bar{\theta}}(\mathrm{d}\theta)$.
Hence, it follows that $\pi_{T}^{\omega}(\mathrm{d}\theta)\rightsquigarrow_{\mathbb{P}^{Y}}\delta_{\bar{\theta}}(\mathrm{d}\theta).$
\end{proof}
To analyse the asymptotic properties of the pseudo-marginal algorithm,
we rescale the parameter component. A simple change of variables and
the fact that convergence in total variation in probability implies
weak convergence in probability shows that the following result holds.
\begin{lemma}
\label{cor} Under Assumption \ref{ass1}, we have
\[
\int\left|\tilde{\pi}_{T}^{\omega}(\tilde{\theta})-\varphi(\tilde{\theta};0,\Sigma)\right|\mathrm{d}\theta\overset{\mathbb{P}^{Y}}{\longrightarrow}0,\quad\text{as}\quad T\rightarrow\infty,
\]
and thus $\tilde{\pi}_{T}^{\omega}(\mathrm{d}\tilde{\theta})\rightsquigarrow_{\mathbb{P}^{Y}}\mathcal{\varphi}(\mathrm{d}\tilde{\theta};0,\Sigma)$.
\end{lemma}

\begin{lemma}[Convergence of marginal distributions]
 \label{prop:prep} Under Assumptions \ref{ass1} and \ref{ass2},
the marginal distribution of the proposal at stationarity
\[
\pi_{T}^{\omega}q_{T}(\mathrm{d}\vartheta)=\int\pi_{T}^{\omega}(\mathrm{d}\theta)q_{T}(\theta,\mathrm{d}\vartheta)
\]
satisfies
\[
\pi_{T}^{\omega}q_{T}(\mathrm{d}\vartheta)\rightsquigarrow_{\mathbb{P}^{Y}}\delta_{\bar{\theta}}(\mathrm{d}\vartheta).
\]
\end{lemma}

\begin{proof}
Let $f\in\mathrm{BL}(\mathbb{\mathbb{R}}),$ then we have
\begin{align*}
\left|\int f(\vartheta)\pi_{T}^{\omega}q_{T}(\mathrm{d}\vartheta)-f(\bar{\theta})\right| & =\left|\int f(\theta+\xi/\surd{T})\int\pi_{T}^{\omega}(\mathrm{d}\theta)\nu(\mathrm{d}\xi)-f(\bar{\theta})\right|\\
 & \leq\left|\iint\left(f(\theta+\xi/\surd{T})-f(\theta)\right)\nu(\mathrm{d}\xi)\pi_{T}^{\omega}(\mathrm{d}\theta)\right|+\left|\iint f(\theta)\pi_{T}^{\omega}(\mathrm{d}\theta)\nu(\mathrm{d}\xi)-f(\bar{\theta})\right|\\
 & \leq\iint\left|f(\theta+\xi/\surd{T})-f(\theta)\right|\nu(\mathrm{d}\xi)\pi_{T}^{\omega}(\mathrm{d}\theta)+\left|\int f(\theta)\pi_{T}^{\omega}(\mathrm{d}\theta)-f(\bar{\theta})\right|.
\end{align*}
The second term on the r.h.s. vanishes due to Lemma \ref{cor1}.
For the first term we use the fact that $f$ is bounded Lipschitz, hence
\begin{align*}
\iint\left|f(\theta+\xi/\surd{T})-f(\theta)\right|\nu(\mathrm{d}\xi)\pi_{T}^{\omega}(\mathrm{d}\theta) & \leq \|f\|_\mathrm{BL} \iint \min\left\{1,\frac{\left\Vert \xi\right\Vert}{\surd{T}}\right\} \nu(\mathrm{d}\xi)\pi_{T}^{\omega}(\mathrm{d}\theta)\\
 & = \|f\|_\mathrm{BL} \iint \min\left\{1,\frac{\left\Vert \xi\right\Vert}{\surd{T}}\right\} \nu(\mathrm{d}\xi) \rightarrow 0.
\end{align*}
\end{proof}

The proof of the following Lemmas are straightforward and thus omitted.
\begin{lemma}
\label{prop_lip} The map $x\mapsto\min\left(1,ae^{x}\right)$
with $a>0$ is $1-$Lipschitz, i.e., for all $x,y\in\mathbb{R}$
\[
\left|\min\left( 1,ae^{x}\right) -\min\left(1,ae^{y}\right) \right|\leq|x-y|.
\]
\end{lemma}

\begin{lemma}
\label{lem:Gaussian_continuity} Under Assumption \ref{ass3}
\begin{itemize}
\item[(i)]  the function
\[
\theta\mapsto d_{\mathrm{BL}}\left[\varphi\left\{\,\cdot\,;\sigma^{2}(\theta)/2,\sigma^{2}(\theta)\right)\},\varphi\left\{\,\cdot\,;\sigma^{2}(\bar{\theta})/2,\sigma^{2}(\bar{\theta})\right\}\mid \mathcal{Y}_T\right]
\]
is bounded for all $\theta$ and continuous at $\bar{\theta}$;
\item[(ii)]  for all $f\in\mathrm{BL}(\mathbb{R})$ the functions
\[
\theta\mapsto\left|\int f(z)\varphi\left\{\mathrm{d}z;\sigma^{2}(\theta)/2,\sigma^{2}(\theta)\right\}-\int f(z)\varphi\left\{\mathrm{d}z;\sigma^{2}(\bar{\theta})/2,\sigma^{2}(\bar{\theta})\right\}\right|
\]
are bounded for all $\theta$ and continuous at $\bar{\theta}$.
\end{itemize}
\end{lemma}

\subsection{Proof of Theorem \ref{theorem}}

In order to prove Theorem$\,$\ref{theorem}, we need to prove Propositions$\:$\ref{prop:4.1},
\ref{prop:4.2} and \ref{prop:4.3} of Section 4$\cdot$3.
%which we repeat here for the
%readers convenience.
%\begin{proposition}[Proposition \ref{prop:4.1}]
%Under Assumptions \ref{ass1}, \ref{ass2} and \ref{ass3}, we have
%\[
%\tilde{\pi}_{T}^{\omega}(\mathrm{d}\tilde{\theta},\mathrm{d}z)\rightsquigarrow_{\mathbb{P}^{Y}}\tilde{\pi}(\mathrm{d}\tilde{\theta},\mathrm{d}z),
%\]
%as $T\rightarrow\infty$ where $\tilde{\pi}_{T}^{\omega}(\mathrm{d}\tilde{\theta},\mathrm{d}z):=\tilde{\pi}_{T}^{\omega}(\mathrm{d}\tilde{\theta})\mathrm{exp}\left(z\right)\tilde{g}_{T}^{\omega}(\mathrm{d}z\mid\tilde{\theta})$.
%\end{prop1}

\setcounter{proposition}{0}

\begin{proposition}
Under Assumptions \ref{ass1} and \ref{ass3},
we have
\[
\tilde{\pi}_{T}^{\omega}(\mathrm{d}\tilde{\theta},\mathrm{d}z)\rightarrow \tilde{\pi}(\mathrm{d}\tilde{\theta},\mathrm{d}z),
\]
weakly in $\mathbb{P}^Y$-probability as $T\rightarrow\infty$ where $\tilde{\pi}_{T}^{\omega}(\mathrm{d}\tilde{\theta},\mathrm{d}z)=\tilde{\pi}_{T}^{\omega}(\mathrm{d}\tilde{\theta})\mathrm{exp}\left(z\right)\tilde{g}_{T}^{\omega}(\mathrm{d}z\mid\tilde{\theta})$.
\end{proposition}

\begin{proof}[Proof of Proposition \ref{prop:4.1}]
As established in Lemma$\,$\ref{lem:finfitedim}, it is enough to
check convergence for products of bounded Lipschitz functions. Now,
without loss of generality, assume that $\|f_{1}\|_{\infty}$,$\|f_{2}\|_{\infty}\leq1/2$.
Then we have
\begin{align}
 & \left|\iint f_{1}(\tilde{\theta})f_{2}(z)\tilde{\pi}_{T}^{\omega}(\mathrm{d}\tilde{\theta})e^{z}\tilde{g}_{T}^{\omega}(\mathrm{d}z\mid\tilde{\theta})-\iint f_{1}(\tilde{\theta})f_{2}(z)\varphi(\mathrm{d}\tilde{\theta};0,\Sigma)\varphi\left\{\mathrm{d}z;\sigma^{2}(\bar{\theta})/2,\sigma^{2}(\bar{\theta})\right\}\right|\nonumber \\
 & \leq\iint e^{z}\tilde{g}_{T}^{\omega}(z\mid\tilde{\theta})\mathrm{d}z\left|\tilde{\pi}_{T}^{\omega}(\tilde{\theta})-\varphi(\tilde{\theta};0,\Sigma)\right|\mathrm{d}\tilde{\theta}\nonumber \\
 & \quad+\int\varphi(\tilde{\theta};0,\Sigma)\left|\int f_{2}(z)e^{z}\tilde{g}_{T}^{\omega}(\mathrm{d}z\mid\tilde{\theta})-\int f_{2}(z)\varphi\left\{\mathrm{d}z;\sigma^{2}(\bar{\theta})/2,\sigma^{2}(\bar{\theta})\right\}\right|\mathrm{d}\tilde{\theta}\nonumber \\
 & \leq\int\left|\tilde{\pi}_{T}^{\omega}(\mathrm{d}\tilde{\theta})-\varphi(\tilde{\theta};0,\Sigma)\right|\mathrm{d}\tilde{\theta}\label{eq:prop1.1}\\
 & \quad+\int\varphi(\theta;\hat{\theta}_{T}^{\omega},\Sigma/T)\left|\int f_{2}(z)e^{z}g_{T}^{\omega}(\mathrm{d}z\mid\theta)-\int f_{2}(z)\varphi\left\{\mathrm{d}z;\sigma^{2}(\theta)2,\sigma^{2}(\theta)\right\}\right|\mathrm{d\theta}\label{eq:prop1.2}\\
 & \quad+\int\varphi(\theta;\hat{\theta}_{T}^{\omega},\Sigma/T)\left|\int f_{2}(z)\varphi\left\{\mathrm{d}z;\sigma^{2}(\theta)/2,\sigma^{2}(\theta)\right\}-\int f_{2}(z)\varphi\left\{\mathrm{d}z;\sigma^{2}(\bar{\theta})/2,\sigma^{2}(\bar{\theta})\right\}\right|\mathrm{d\theta}\label{eq:prop1.3}
\end{align}
The term (\ref{eq:prop1.1}) converges to zero in $\mathbb{P}^{Y}$-probability
by Lemma$\,$\ref{cor}. For (\ref{eq:prop1.2}), write $B(\bar{\theta})\subset\Theta$
for the $\varepsilon$-ball on which the uniform CLT in Assumption
\ref{ass3} holds, that is
\[
\sup_{\theta\in B(\bar{\theta})}h_{T}(\theta)=\sup_{\theta\in B(\bar{\theta})}\left|\int f_{2}(z)e^{z}g_{T}^{\omega}(\mathrm{d}z\mid\theta)\mathrm{d}z-\int f_{2}(z)\varphi\left\{\mathrm{d}z;\sigma^{2}(\theta)/2,\sigma^{2}(\theta)\right\}\mathrm{d}z\right|\overset{\mathbb{P}^{Y}}{\longrightarrow}0.
\]
We can bound (\ref{eq:prop1.2}) as follows
\begin{align*}
 & \int_{B(\bar{\theta})}\varphi(\theta;\hat{\theta}_{T}^{\omega},\Sigma/T)\left|\int f_{2}(z)e^{z}g_{T}^{\omega}(\mathrm{d}z\mid\theta)-\int f_{2}(z)\varphi\left\{\mathrm{d}z;\sigma^{2}(\theta)/2,\sigma^{2}(\theta)\right\}\right|\mathrm{d\theta}\\
 & +\int_{B(\bar{\theta})^{\complement}}\varphi(\theta;\hat{\theta}_{T}^{\omega},\Sigma/T)\left|\int f_{2}(z)e^{z}g_{T}^{\omega}(\mathrm{d}z\mid\theta)-\int f_{2}(z)\varphi\left\{\mathrm{d}z;\sigma^{2}(\theta)/2,\sigma^{2}(\theta)\right\}\right|\mathrm{d\theta}\\
 & \leq\sup_{\theta\in B(\bar{\theta})}h_{T}(\theta)+\int_{B(\bar{\theta})^{\complement}}\varphi(\theta;\hat{\theta}_{T}^{\omega},\Sigma/T)\mathrm{d\theta},
\end{align*}
since $\|f_{2}\|_{\infty}\leq1/2$. We have already mentioned that the
first term vanishes in probability whereas for the second term we
have
\[
\int_{B(\bar{\theta})^{\complement}}\varphi(\theta;\hat{\theta}_{T}^{\omega},\Sigma/T)\mathrm{d\theta}\overset{\mathbb{P}^{Y}}{\longrightarrow}\delta_{\bar{\theta}}\big\{B(\bar{\theta})^{\complement}\big\}=0,
\]
by Lemma$\,$\ref{cor1}. Thus (\ref{eq:prop1.2}) vanishes in $\mathbb{P}^{Y}$-probability.
Finally we consider (\ref{eq:prop1.3}). By Lemma \ref{lem:Gaussian_continuity}
\[
h(\theta)=\left|\int f_{2}(z)\varphi\left\{\mathrm{d}z;\sigma^{2}(\theta)/2,\sigma^{2}(\theta)\right\}-\int f_{2}(z)\varphi\left\{\mathrm{d}z;\sigma^{2}(\bar{\theta})/2,\sigma^{2}(\bar{\theta})\right\}\right|
\]
is bounded and continuous at $\bar{\theta}$. Since $\varphi(\mathrm{d}\theta;\hat{\theta}_{T}^{\omega},\Sigma/T)$
converges weakly in probability to a point mass in $\bar{\theta}$
(by Lemma$\,$\ref{cor1}) we can conclude that
\[
\int f(\theta)\varphi(\mathrm{d}\theta;\hat{\theta}_{T}^{\omega},\Sigma/T)\overset{\mathbb{P}^{Y}}{\rightarrow}\int f(\theta)\delta_{\bar{\theta}}(\mathrm{d}\theta)
\]
for every bounded function $f$ which is continuous at $\overline{\theta}$.
In particular,
\[
\int h(\theta)\varphi(\mathrm{d}\theta;\hat{\theta}_{T}^{\omega},\Sigma/T)\overset{\mathbb{P}^{Y}}{\rightarrow}0.
\]
\end{proof}

%\newpage
\begin{proposition}
Under Assumptions \ref{ass1}, \ref{ass2} and \ref{ass3}, as $T\rightarrow\infty$
we have for any $f\in\mathrm{BL}(\mathbb{R}^{d+1})$
\[
\mathbb{\int}|\tilde{P}_{T}^{\omega}f(\theta,z)-\tilde{P}f(\theta,z)|\tilde{\pi}_{T}^{\omega}(\mathrm{d}\theta,\mathrm{d}z)\rightarrow 0,\quad \text{in } \mathbb{P}^Y\text{-probability}.
\]
\end{proposition}

\begin{proof}[Proof of Proposition \ref{prop:4.2}]
Let $f\in\mathrm{BL}(\mathbb{R}^{d+1})$. Denote
\[
\Pi_{T}^{\omega}f(\tilde{\theta},z)=\iint f(\tilde{\theta}^{\prime},z')\tilde{\alpha}_{T}^{\omega}\big\{(\tilde{\theta},z),(\tilde{\theta}^{\prime},z')\big\}\tilde{q}(\tilde{\theta},\mathrm{d}\tilde{\theta}^{\prime})\tilde{g}_{T}^{\omega}(\mathrm{d}z'\mid\tilde{\theta}^{\prime})
\]
and
\[
\Pi f(\tilde{\theta},z)=\iint f(\tilde{\theta}^{\prime},z')\tilde{\alpha}\big\{(\tilde{\theta},z),(\tilde{\theta}^{\prime},z')\big\}\tilde{q}(\tilde{\theta},\mathrm{d}\tilde{\theta}^{\prime})g(\mathrm{d}z'\mid\overline{\theta}),
\]
where $g(\,\cdot\,\mid\vartheta)=\varphi\{\,\cdot\,;-\sigma^{2}(\vartheta)/2,\sigma^{2}(\vartheta)\}$.
Then we have
\begin{align*}
\tilde{P}_{T}^{\omega}f(\tilde{\theta},z) & =\Pi_{T}^{\omega}f(\tilde{\theta},z)+f(\tilde{\theta},z)\left\{1-\Pi_{T}^{\omega}1(\tilde{\theta},z)\right\}
\end{align*}
and
\begin{align}
\tilde{P}f(\tilde{\theta},z) & =\Pi f(\tilde{\theta},z)+f(\tilde{\theta},z)\left\{1-\Pi1(\tilde{\theta},z)\right\}. \label{eq:op_decomp}
\end{align}
Because
\begin{align*}
 & {E^{\omega}}\left\{\left|\tilde{P}_{T}^{\omega}f(\tilde{\vartheta}_{0}^{T},Z_{0}^{T})-\tilde{P}f(\tilde{\vartheta}_{0}^{T},Z_{0}^{T})\right|\right\}\\
 & =E^{\omega}\Bigg[\Big|\Pi_{T}^{\omega}f(\tilde{\vartheta}_{0}^{T},Z_{0}^{T})+f(\tilde{\vartheta}_{0}^{T},Z_{0}^{T})\left\{1-\Pi_{T}^{\omega}1(\tilde{\vartheta}_{0}^{T},Z_{0}^{T})\right\} \\
& \qquad\,\, -\Pi f(\tilde{\vartheta}_{0}^{T},Z_{0}^{T})-f(\tilde{\vartheta}_{0}^{T},Z_{0}^{T})\left\{1-\Pi1(\tilde{\vartheta}_{0}^{T},Z_{0}^{T})\right\}\Big|\Bigg]\\
 & \leq E^{\omega}\left\{\left|\Pi_{T}^{\omega}f(\tilde{\vartheta}_{0}^{T},Z_{0}^{T})-\Pi f(\tilde{\vartheta}_{0}^{T},Z_{0}^{T})\right|\right\}+E^{\omega}\left\{\left|\Pi_{T}^{\omega}1(\tilde{\vartheta}_{0}^{T},Z_{0}^{T})-\Pi 1(\tilde{\vartheta}_{0}^{T},Z_{0}^{T})\right|\right\}
\end{align*}
and $1\in\mathrm{\mathrm{BL}}(\mathbb{R}^{d+1})$ it is sufficient
to show that for any choice of $f\in\mathrm{BL}(\mathbb{R}^{d+1})$
we have $$E^{\omega}\left\{\left|\Pi_{T}^{\omega}f(\tilde{\theta},z)-\Pi f(\tilde{\theta},z)\right|\right\}\overset{\mathbb{P}^{Y}}{\rightarrow}0.$$
Thus
\begin{align}
 & E^{\omega}\left\{\left|\Pi_{T}^{\omega}f(\tilde{\theta},z)-\Pi f(\tilde{\theta},z)\right|\right\}\nonumber \\
 & =\iint\tilde{\pi}_{T}^{\omega}(\mathrm{d}\tilde{\theta},\mathrm{d}z)\Bigg|\iint\tilde{q}(\tilde{\theta},\mathrm{d}\tilde{\theta}^{\prime})\tilde{\alpha}_{T}^{\omega}\big\{(\tilde{\theta},z),(\tilde{\theta}^{\prime},z')\big\}f(\tilde{\theta}^{\prime},z')\tilde{g}_{T}^{\omega}(\mathrm{d}z'\mid\tilde{\theta}^{\prime})\nonumber \\
 & \qquad-\iint\tilde{q}(\tilde{\theta},\mathrm{d}\tilde{\theta}^{\prime})\tilde{\alpha}\{(\tilde{\theta},z),(\tilde{\theta}^{\prime},z')\}f(\tilde{\theta}^{\prime},z')g(\mathrm{d}z'\mid\overline{\theta})\Bigg|\nonumber \\
 & =\iint e^{z}\tilde{g}_{T}^{\omega}(\mathrm{d}z\mid\tilde{\theta})\Bigg|\iint\min\left\{\tilde{\pi}_{T}^{\omega}(\tilde{\theta})\tilde{q}(\tilde{\theta},\tilde{\theta}^{\prime}),\tilde{\pi}_{T}^{\omega}(\tilde{\theta}^{\prime})\tilde{q}(\tilde{\theta}^{\prime},\tilde{\theta})e^{z'-z}\right\}f(\tilde{\theta}^{\prime},z')\tilde{g}_{T}^{\omega}(\mathrm{d}z'\mid\tilde{\theta}^{\prime})\mathrm{d}\tilde{\theta}^{\prime}\nonumber \\
 & \qquad-\iint\tilde{\pi}_{T}^{\omega}(\tilde{\theta})\tilde{q}(\tilde{\theta},\tilde{\theta}^{\prime})\tilde{\alpha}\{(\tilde{\theta},z),(\tilde{\theta}^{\prime},z')\}f(\tilde{\theta}^{\prime},z')g(\mathrm{d}z'\mid\overline{\theta})\mathrm{d}\tilde{\theta}^{\prime}\Bigg|\mathrm{d}\tilde{\theta}\nonumber \\
 & \leq\iint e^{z}\tilde{g}_{T}^{\omega}(\mathrm{d}z\mid\tilde{\theta})\Bigg|\iint\min\left\{\tilde{\pi}_{T}^{\omega}(\tilde{\theta})\tilde{q}(\tilde{\theta},\tilde{\theta}^{\prime}),\tilde{\pi}_{T}^{\omega}(\tilde{\theta}^{\prime})\tilde{q}(\tilde{\theta}^{\prime},\tilde{\theta})e^{z'-z}\right\}f(\tilde{\theta}^{\prime},z')\tilde{g}_{T}^{\omega}(\mathrm{d}z'\mid\tilde{\theta}^{\prime})\mathrm{d}\tilde{\theta}^{\prime}\nonumber \\
 & \qquad-\iint\min\left\{\varphi(\tilde{\theta};0,\Sigma)\tilde{q}(\tilde{\theta},\tilde{\theta}^{\prime}),\varphi(\tilde{\theta}^{\prime};0,\Sigma)\tilde{q}(\tilde{\theta}^{\prime},\tilde{\theta})e^{z'-z}\right\}f(\tilde{\theta}^{\prime},z')\tilde{g}_{T}^{\omega}(\mathrm{d}z'\mid\tilde{\theta}^{\prime})\mathrm{d}\tilde{\theta}^{\prime}\Bigg|\mathrm{d}\tilde{\theta}\nonumber \\
 & +\iint e^{z}\tilde{g}_{T}^{\omega}(\mathrm{d}z\mid\tilde{\theta})\Bigg|\iint\min\left\{\varphi(\tilde{\theta};0,\Sigma)\tilde{q}(\tilde{\theta},\tilde{\theta}^{\prime}),\varphi(\tilde{\theta}^{\prime};0,\Sigma)\tilde{q}(\tilde{\theta}^{\prime},\tilde{\theta})e^{z'-z}\right\}f(\tilde{\theta}^{\prime},z')\tilde{g}_{T}^{\omega}(\mathrm{d}z'\mid\tilde{\theta}^{\prime})\mathrm{d}\tilde{\theta}^{\prime}\nonumber \\
 & \qquad-\iint\tilde{\pi}_{T}^{\omega}(\tilde{\theta})\tilde{q}(\tilde{\theta},\tilde{\theta}^{\prime})\tilde{\alpha}\{(\tilde{\theta},z),(\tilde{\theta}^{\prime},z')\}f(\tilde{\theta}^{\prime},z')g(\mathrm{d}z'\mid\overline{\theta})\mathrm{d}\tilde{\theta}^{\prime}\Bigg|\mathrm{d}\tilde{\theta}.\label{eq:final_part2-1}
\end{align}
By taking $\varphi(\tilde{\theta};0,\Sigma)$ out in last two
lines of (\ref{eq:final_part2-1}), this can be rewritten as
\begin{align}
 & \ \iint e^{z}\tilde{g}_{T}^{\omega}(\mathrm{d}z\mid\tilde{\theta})\Bigg|\iint\min\left\{\tilde{\pi}_{T}^{\omega}(\tilde{\theta})\tilde{q}(\tilde{\theta},\tilde{\theta}^{\prime}),\tilde{\pi}_{T}^{\omega}(\tilde{\theta}^{\prime})\tilde{q}(\tilde{\theta}^{\prime},\tilde{\theta})e^{z'-z}\right\}f(\tilde{\theta}^{\prime},z')\tilde{g}_{T}^{\omega}(\mathrm{d}z'\mid\tilde{\theta}^{\prime})\mathrm{d}\tilde{\theta}^{\prime}\nonumber \\
 & \qquad-\iint\min\left\{\varphi(\tilde{\theta};0,\Sigma)\tilde{q}(\tilde{\theta},\tilde{\theta}^{\prime}),\varphi(\tilde{\theta}^{\prime};0,\Sigma)\tilde{q}(\tilde{\theta}^{\prime},\tilde{\theta})e^{z'-z}\right\}f(\tilde{\theta}^{\prime},z')\tilde{g}_{T}^{\omega}(\mathrm{d}z'\mid\tilde{\theta}^{\prime})\mathrm{d}\tilde{\theta}^{\prime}\Bigg|\mathrm{d}\tilde{\theta}\label{eq:final_part1.1}\\
 & +\iint e^{z}\tilde{g}_{T}^{\omega}(\mathrm{d}z\mid\tilde{\theta})\Bigg|\iint\varphi(\tilde{\theta};0,\Sigma)\tilde{q}(\tilde{\theta},\tilde{\theta}^{\prime})\tilde{\alpha}\big\{(\tilde{\theta},z),(\tilde{\theta}^{\prime},z')\big\}f(\tilde{\theta}^{\prime},z')\tilde{g}_{T}^{\omega}(\mathrm{d}z'\mid\tilde{\theta}^{\prime})\mathrm{d}\tilde{\theta}^{\prime}\nonumber \\
 & \qquad-\iint\tilde{\pi}_{T}^{\omega}(\tilde{\theta})\tilde{q}(\tilde{\theta},\tilde{\theta}^{\prime})\tilde{\alpha}\big\{(\tilde{\theta},z),(\tilde{\theta}^{\prime},z')\big\}f(\tilde{\theta}^{\prime},z')g(\mathrm{d}z'\mid\overline{\theta})\mathrm{d}\tilde{\theta}^{\prime}\Bigg|\mathrm{d}\tilde{\theta}.\label{eq:final_part1.2}
\end{align}
For (\ref{eq:final_part1.1}), we use the inequality $|\min(a,b)-\min(c,d)|\leq|a-c|+|b-d|$:
\begin{align*}
 & \iint e^{z}\tilde{g}_{T}^{\omega}(\mathrm{d}z\mid\tilde{\theta})\Bigg|\iint\min\left\{\tilde{\pi}_{T}^{\omega}(\tilde{\theta})\tilde{q}(\tilde{\theta},\tilde{\theta}^{\prime}),\tilde{\pi}_{T}^{\omega}(\tilde{\theta}^{\prime})\tilde{q}(\tilde{\theta}^{\prime},\tilde{\theta})e^{z'-z}\right\}f(\tilde{\theta}^{\prime},z')\tilde{g}_{T}^{\omega}(\mathrm{d}z'\mid\tilde{\theta}^{\prime})\\
 & \qquad-\min\left\{\varphi(\tilde{\theta};0,\Sigma)\tilde{q}(\tilde{\theta},\tilde{\theta}^{\prime}),\varphi(\tilde{\theta}^{\prime};0,\Sigma)\tilde{q}(\tilde{\theta}^{\prime},\tilde{\theta})e^{z'-z}\right\}f(\tilde{\theta}^{\prime},z')\tilde{g}_{T}^{\omega}(\mathrm{d}z'\mid\tilde{\theta}^{\prime})\mathrm{d}\tilde{\theta}^{\prime}\Bigg|\mathrm{d}\tilde{\theta}\\
 &  \leq \|f\|_{\infty} \iiiint e^{z}\tilde{g}_{T}^{\omega}(\mathrm{d}z\mid\tilde{\theta})\tilde{q}(\tilde{\theta},\mathrm{d}\tilde{\theta}^{\prime})\tilde{g}_{T}^{\omega}(\mathrm{d}z'\mid\tilde{\theta}^{\prime})\left|\tilde{\pi}_{T}^{\omega}(\tilde{\theta})-\varphi(\tilde{\theta};0,\Sigma)\right|\mathrm{d}\tilde{\theta}\\
 & \qquad+ \|f\|_{\infty} \iiiint\tilde{g}_{T}^{\omega}(\mathrm{d}z\mid\tilde{\theta})e^{z'}\tilde{g}_{T}^{\omega}(\mathrm{d}z'\mid\tilde{\theta}^{\prime})\tilde{q}(\tilde{\theta}^{\prime},\tilde{\theta})\left|\tilde{\pi}_{T}^{\omega}(\tilde{\theta}^{\prime})-\varphi(\tilde{\theta}^{\prime};0,\Sigma)\right|\mathrm{d}\tilde{\theta}^{\prime}\mathrm{d}\tilde{\theta}\\
 & =2  \|f\|_{\infty} \int\left|\tilde{\pi}_{T}^{\omega}(\tilde{\theta})-\varphi(\tilde{\theta};0,\Sigma)\right|\mathrm{d}\tilde{\theta}\overset{\mathbb{P}^{Y}}{\longrightarrow}0,
\end{align*}
by Lemma \ref{cor}. For the part (\ref{eq:final_part1.2}) note that
\begin{align}
 & \iint e^{z}\tilde{g}_{T}^{\omega}(\mathrm{d}z\mid\tilde{\theta})\Bigg|\iint\varphi(\tilde{\theta};0,\Sigma)\tilde{q}(\tilde{\theta},\mathrm{d}\tilde{\theta}^{\prime})\tilde{\alpha}\big\{(\tilde{\theta},z),(\tilde{\theta}^{\prime},z')\big\}f(\tilde{\theta}^{\prime},z')\tilde{g}_{T}^{\omega}(\mathrm{d}z'\mid\tilde{\theta}^{\prime})\nonumber \\
 & \qquad-\iint\tilde{\pi}_{T}^{\omega}(\tilde{\theta})\tilde{q}(\tilde{\theta},\mathrm{d}\tilde{\theta}^{\prime})\tilde{\alpha}\big\{(\tilde{\theta},z),(\tilde{\theta}^{\prime},z')\big\}f(\tilde{\theta}^{\prime},z')g(\mathrm{d}z'\mid\overline{\theta})\Bigg|\mathrm{d}\tilde{\theta}\nonumber \\
 & \leq\iint e^{z}\tilde{g}_{T}^{\omega}(\mathrm{d}z\mid\tilde{\theta})\Bigg|\iint\varphi(\tilde{\theta};0,\Sigma)\tilde{q}(\tilde{\theta},\mathrm{d}\tilde{\theta}^{\prime})\tilde{\alpha}\big\{(\tilde{\theta},z),(\tilde{\theta}^{\prime},z')\big\}f(\tilde{\theta}^{\prime},z')\tilde{g}_{T}^{\omega}(\mathrm{d}z'\mid\tilde{\theta}^{\prime})\nonumber \\
 & \qquad-\iint\tilde{\pi}_{T}^{\omega}(\tilde{\theta})\tilde{q}(\tilde{\theta},\mathrm{d}\tilde{\theta}^{\prime})\tilde{\alpha}\big\{(\tilde{\theta},z),(\tilde{\theta}^{\prime},z')\big\}f(\tilde{\theta}^{\prime},z')\tilde{g}_{T}^{\omega}(\mathrm{d}z'\mid\tilde{\theta}^{\prime})\Bigg|\mathrm{d}\tilde{\theta}\label{eq:1}\\
 & \qquad+\iint\tilde{\pi}_{T}^{\omega}(\mathrm{d}\tilde{\theta})e^{z}\tilde{g}_{T}^{\omega}(\mathrm{d}z\mid\tilde{\theta})\left|\iint\tilde{q}(\tilde{\theta},\mathrm{d}\tilde{\theta}^{\prime})\tilde{g}_{T}^{\omega}(\mathrm{d}z'\mid\tilde{\theta}^{\prime})\tilde{\alpha}\big\{(\tilde{\theta},z),(\tilde{\theta}^{\prime},z')\big\}f(\tilde{\theta}^{\prime},z')\right.\nonumber \\
 & \qquad-\left.\iint\tilde{q}(\tilde{\theta},\mathrm{d}\tilde{\theta}^{\prime})g(\mathrm{d}z'\mid\overline{\theta})\tilde{\alpha}\big\{(\tilde{\theta},z),(\tilde{\theta}^{\prime},z')\big\}f(\tilde{\theta}^{\prime},z')\right|.\label{eq:2}
\end{align}
For the first part (\ref{eq:1}) we have
\begin{align*}
 & \iint e^{z}\tilde{g}_{T}^{\omega}(\mathrm{d}z\mid\tilde{\theta})\Bigg|\iint\varphi(\tilde{\theta};0,\Sigma)\tilde{q}(\tilde{\theta},\tilde{\theta}^{\prime})\tilde{\alpha}\big\{(\tilde{\theta},z),(\tilde{\theta}^{\prime},z')\big\}f(\tilde{\theta}^{\prime},z')\tilde{g}_{T}^{\omega}(\mathrm{d}z'\mid\tilde{\theta}^{\prime})\mathrm{d}\tilde{\theta}^{\prime}\\
 & \qquad-\iint\tilde{\pi}_{T}^{\omega}(\tilde{\theta})\tilde{q}(\tilde{\theta},\tilde{\theta}^{\prime})\tilde{\alpha}\big\{(\tilde{\theta},z),(\tilde{\theta}^{\prime},z')\big\}f(\tilde{\theta}^{\prime},z')\tilde{g}_{T}^{\omega}(\mathrm{d}z'\mid\tilde{\theta}^{\prime})\mathrm{d}\tilde{\theta}^{\prime}\Bigg|\mathrm{d}\tilde{\theta}\\
 &  \leq \|f\|_{\infty} \iiiint e^{z}\tilde{g}_{T}^{\omega}(\mathrm{d}z\mid\tilde{\theta})\tilde{q}(\tilde{\theta},\mathrm{d}\tilde{\theta}^{\prime})\tilde{g}_{T}^{\omega}(\mathrm{d}z'\mid\tilde{\theta}^{\prime})\left|\varphi(\tilde{\theta};0,\Sigma)-\tilde{\pi}_{T}^{\omega}(\tilde{\theta})\right|\mathrm{d}\tilde{\theta}\\
 & =  \|f\|_{\infty} \int\left|\varphi(\tilde{\theta};0,\Sigma)-\tilde{\pi}_{T}^{\omega}(\tilde{\theta})\right|\mathrm{d}\tilde{\theta}\overset{\mathbb{P}^{Y}}{\longrightarrow}0,
\end{align*}
again by Lemma \ref{cor}. The second part (\ref{eq:2})
\begin{align}
 & \iint\tilde{\pi}_{T}^{\omega}(\mathrm{d}\tilde{\theta})e^{z}\tilde{g}_{T}^{\omega}(\mathrm{d}z\mid\tilde{\theta})\left|\iint\tilde{q}(\tilde{\theta},\mathrm{d}\tilde{\theta}^{\prime})\tilde{g}_{T}^{\omega}(\mathrm{d}z'\mid\tilde{\theta}^{\prime})\tilde{\alpha}\big\{(\tilde{\theta},z),(\tilde{\theta}^{\prime},z')\big\}f(\tilde{\theta}^{\prime},z')\right.\nonumber \\
 & \qquad-\left.\iint\tilde{q}(\tilde{\theta},\mathrm{d}\tilde{\theta}^{\prime})g(\mathrm{d}z'\mid\overline{\theta})\tilde{\alpha}\big\{(\tilde{\theta},z),(\tilde{\theta}^{\prime},z')\big\}f(\tilde{\theta}^{\prime},z')\right|\nonumber \\
 & \leq\iiint\tilde{\pi}_{T}^{\omega}(\mathrm{d}\tilde{\theta})e^{z}\tilde{g}_{T}^{\omega}(\mathrm{d}z\mid\tilde{\theta})\tilde{q}(\tilde{\theta},\mathrm{d}\tilde{\theta}^{\prime})\left|\int\tilde{g}_{T}^{\omega}(\mathrm{d}z'\mid\tilde{\theta}^{\prime})\tilde{\alpha}\big\{(\tilde{\theta},z),(\tilde{\theta}^{\prime},z')\big\}f(\tilde{\theta}^{\prime},z')\right.\nonumber \\
 & \qquad-\left.\int\tilde{\alpha}\big\{(\tilde{\theta},z),(\tilde{\theta}^{\prime},z')\big\}g(\mathrm{d}z'\mid\hat{\theta}_{T}^{\omega}+\tilde{\theta}^{\prime}/\surd{T})f(\tilde{\theta}^{\prime},z')\right|\label{eq:end1}\\
 & \qquad+\iiint\tilde{\pi}_{T}^{\omega}(\mathrm{d}\tilde{\theta})e^{z}\tilde{g}_{T}^{\omega}(\mathrm{d}z\mid\tilde{\theta})\tilde{q}(\tilde{\theta},\mathrm{d}\tilde{\theta}^{\prime})\left|\int g(\mathrm{d}z'\mid\hat{\theta}_{T}^{\omega}+\tilde{\theta}^{\prime}/\surd{T})f(\tilde{\theta}^{\prime},z')\tilde{\alpha}\big\{(\tilde{\theta},z),(\tilde{\theta}^{\prime},z')\big\}f(\tilde{\theta}^{\prime},z')\right.\nonumber \\
 & \qquad-\left.\int\tilde{\alpha}\{(\tilde{\theta},z),(\tilde{\theta}^{\prime},z')\}g(\mathrm{d}z'\mid\bar{\theta})f(\tilde{\theta}^{\prime},z')\right|\label{eq:end2}
\end{align}
We first consider (\ref{eq:end1}) using $\theta=\hat{\theta}_{T}^{\omega}+\tilde{\theta}/\surd{T}$,
and similarly for $\theta'$,
\begin{align*}
 & \iiint\tilde{\pi}_{T}^{\omega}(\mathrm{d}\tilde{\theta})e^{z}\tilde{g}_{T}^{\omega}(\mathrm{d}z\mid\tilde{\theta})\tilde{q}(\tilde{\theta},\mathrm{d}\tilde{\theta}^{\prime})\Big|\int\min\left\{1,\frac{\varphi(\tilde{\theta}^{\prime};0,\Sigma)}{\varphi(\tilde{\theta};0,\Sigma)}\frac{\tilde{q}(\tilde{\theta}^{\prime},\tilde{\theta})}{\tilde{q}(\tilde{\theta},\tilde{\theta}^{\prime})}e^{z'-z}\right\}\tilde{g}_{T}^{\omega}(\mathrm{d}z'\mid\tilde{\theta}^{\prime})f(\tilde{\theta}^{\prime},z')\\
 & \qquad-\int\min\left\{1,\frac{\varphi(\tilde{\theta}^{\prime};0,\Sigma)}{\varphi(\tilde{\theta};0,\Sigma)}\frac{\tilde{q}(\tilde{\theta}^{\prime},\tilde{\theta})}{\tilde{q}(\tilde{\theta},\tilde{\theta}^{\prime})}e^{z'-z}\right\}g(\mathrm{d}z'\mid\hat{\theta}_{T}^{\omega}+\tilde{\theta}^{\prime}/\surd{T})f(\tilde{\theta}^{\prime},z')\Big|\\
 & =\iiint\pi_{T}^{\omega}(\mathrm{d}\theta)e^{z}g_{T}^{\omega}(\mathrm{d}z\mid\theta)q_{T}(\theta,\mathrm{d}\theta') \\
& \quad \times \Big|\int\min\left\{1,\frac{\varphi(\theta';\hat{\theta}_{T}^{\omega},\Sigma/T)}{\varphi(\theta;\hat{\theta}_{T}^{\omega},\Sigma/T)}\frac{q_{T}(\theta',\theta)}{q_{T}(\theta,\theta')}e^{z'-z}\right\}g_{T}^{\omega}(\mathrm{d}z'\mid\theta')f\big\{\surd{T}(\theta'-\hat{\theta}_{T}^{\omega}),z'\big\}\\
 & \qquad-\int\min\left\{1,\frac{\varphi(\theta';\hat{\theta}_{T}^{\omega},\Sigma/T)}{\varphi(\theta;\hat{\theta}_{T}^{\omega},\Sigma/T)}\frac{q_{T}(\theta',\theta)}{q_{T}(\theta,\theta')}e^{z'-z}\right\}g(\mathrm{d}z'\mid\theta')f\big\{\surd{T}(\theta'-\hat{\theta}_{T}^{\omega}),z'\big\}\Big|
\end{align*}
In the rest of the proof, without loss of generality, we will consider $f$ such that $\|f\|_\mathrm{L}\leq 1$ 
\begin{align*}
& \left|f\big\{\surd{T}(\theta'-\hat{\theta}_{T}^{\omega}),x\big\}-f\big\{\surd{T}(\theta'-\hat{\theta}_{T}^{\omega}),y\big\}\right|\\
&\leq d\left[\big\{\surd{T}(\theta'-\hat{\theta}_{T}^{\omega}),x\big\},\big\{\surd{T}(\theta'-\hat{\theta}_{T}^{\omega}),y\big\}\right] =|x-y|
 \end{align*}
and thus $x\mapsto f\big\{\surd{T}(\theta'-\hat{\theta}_{T}^{\omega}),x\big\}$
is Lipschitz with coefficient 1 uniformly in $T$. Moreover, due to
Lemma \ref{prop_lip}, the map
\[
z'\mapsto\min\left\{1,e^{-z}\frac{\varphi(\theta';\hat{\theta}_{T}^{\omega},\Sigma/T)}{\varphi(\theta;\hat{\theta}_{T}^{\omega},\Sigma/T)}\frac{q_{T}(\theta',\theta)}{q_{T}(\theta,\theta')}e^{z'}\right\}
\]
is Lipschitz with Lipschitz constant 1 uniformly for all $\theta,\theta',z$
and $T$. Thus, using the triangle inequality, we can write
\begin{align*}
 & \iiint\pi_{T}^{\omega}(\mathrm{d}\theta)e^{z}g_{T}^{\omega}(\mathrm{d}z\mid\theta)q_{T}(\theta,\mathrm{d}\theta')\Big|\int\min\left\{1,\frac{\varphi(\theta';\hat{\theta}_{T}^{\omega},\Sigma/T)}{\varphi(\theta;\hat{\theta}_{T}^{\omega},\Sigma/T)}e^{z'-z}\right\}f\big\{\surd{T}(\theta'-\hat{\theta}_{T}^{\omega}),z'\big\}g_{T}^{\omega}(z'\mid\theta')\\
 & \qquad-\min\left\{1,\frac{\varphi(\theta';\hat{\theta}_{T}^{\omega},\Sigma/T)}{\varphi(\theta;\hat{\theta}_{T}^{\omega},\Sigma/T)}e^{z'-z}\right\}f\big\{\surd{T}(\theta'-\hat{\theta}_{T}^{\omega}),z'\big\}g(z'\mid\theta')\Big|\mathrm{d}z'\\
 & \leq2\iint\pi_{T}^{\omega}(\mathrm{d}\theta)q_{T}(\theta,\mathrm{d}\theta')\cdot\sup_{f\in\mathrm{BL}(\mathbb{R}),\thinspace\|f\|_{\mathrm{BL}}\leq1}\quad\Big|\int f(z')g_{T}^{\omega}(\mathrm{d}z'\mid\theta')-\int f(z')g(\mathrm{d}z'\mid\theta')\Big|\mathrm{d}\theta\\
 & =2\iint\pi_{T}^{\omega}(\mathrm{d}\theta)q_{T}(\theta,\mathrm{d}\theta')d_{\mathrm{BL}}\big\{g_{T}^{\omega}(\cdot|\theta'),g(\cdot|\theta')\big\}\\
 & =2\int_{B(\bar{\theta})}\pi_{T}^{\omega}q_{T}(\mathrm{d}\theta')d_{\mathrm{BL}}\big\{g_{T}^{\omega}(\cdot|\theta'),g(\cdot|\theta')\big\}+2\int_{B(\bar{\theta})^{\complement}}\pi_{T}^{\omega}q_{T}(\mathrm{d}\theta')d_{\mathrm{BL}}\left(g_{T}^{\omega}(\cdot|\theta'),g(\cdot|\theta')\right),
\end{align*}
where $B(\bar{\theta})$ is given in Assumption$\,$\ref{ass3}. Since
the bounded Lipschitz norm metrizes weak convergence (for non-random
probability measures) we know that for $\theta'\in B(\bar{\theta})$
\[
d_{\mathrm{BL}}\left(g_{T}^{\omega}(\cdot|\theta'),g(\cdot|\theta')\right)=\sup_{f\in\mathrm{BL}(\mathbb{R}),\thinspace\|f\|_{\mathrm{BL}}\leq1}\quad\Big|\int f(z')g_{T}^{\omega}(\mathrm{d}z'\mid\theta')-\int f(z')g(\mathrm{d}z'\mid\theta')\Big|
\]
vanishes in $\mathbb{P}^{Y}$-probability by Assumption \ref{ass3}.
From Lemma \ref{prop:prep} we know that the marginal distribution
of the proposal at stationarity $\pi_{T}^{\omega}q_{T}(\mathrm{d}\theta')=\int\pi_{T}^{\omega}(\mathrm{d}\theta)q(\theta,\mathrm{d}\theta')$
concentrates around the true parameter value. Since the bounded Lipschitz
metric cannot exceed 1 we have
\[
\int\pi_{T}^{\omega}q_{T}(\mathrm{d}\theta')\mathrm{\mathbb{I}}_{B(\bar{\theta})^{\complement}}(\theta')d_{\mathrm{BL}}\left(g_{T}^{\omega}(\cdot|\theta'),g(\cdot|\theta')\right)\leq \pi_{T}^{\omega}q_{T}\big\{B(\bar{\theta})^{\complement}\big\}\overset{\mathbb{P}^{Y}}{\longrightarrow}\delta_{\bar{\theta}}\big\{B(\bar{\theta})^{\complement}\big\}=0.
\]

In addition from Assumption$\,$\ref{ass3}
\[
\left|\int_{B(\bar{\theta})}\pi_{T}^{\omega}q_{T}(\mathrm{d}\theta')d_{\mathrm{BL}}\big\{g_{T}^{\omega}(\cdot|\theta'),g(\cdot|\theta')\big\}\right|\leq\sup_{\theta\in B(\bar{\theta})}\left|d_{\mathrm{BL}}\big\{g_{T}^{\omega}(\cdot|\theta),g(\cdot|\theta)\big\}\right|\overset{\mathbb{P}^{Y}}{\longrightarrow}0.
\]
Finally, using a similar argument for (\ref{eq:end2}) we have
\begin{align}
 & \iiint\tilde{\pi}_{T}^{\omega}(\mathrm{d}\tilde{\theta})e^{z}\tilde{g}_{T}^{\omega}(\mathrm{d}z\mid\tilde{\theta})\tilde{q}(\tilde{\theta},\mathrm{d}\tilde{\theta}^{\prime})\left|\int g(z'\mid\hat{\theta}_{T}^{\omega}+\tilde{\theta}^{\prime}/\surd{T})f(\tilde{\theta}^{\prime},z')\tilde{\alpha}\big\{(\tilde{\theta},z),(\tilde{\theta}^{\prime},z')\big\}f(\tilde{\theta}^{\prime},z')\right.\nonumber \\
 & \qquad-\left.\tilde{q}(\tilde{\theta},\mathrm{d}\tilde{\theta}^{\prime})\tilde{\alpha}\big\{(\tilde{\theta},z),(\tilde{\theta}^{\prime},z')\big\}g(z'\mid\bar{\theta})f(\tilde{\theta}^{\prime},z')\right|\mathrm{d}z'\nonumber \\
 & \quad \leq 2 \iint\pi_{T}^{\omega}(\theta)q_{T}(\theta,\theta')\mathrm{d}\theta d_{\mathrm{BL}}\left(g(\cdot|\theta'),g(\cdot|\bar{\theta})\right)\mathrm{d}\theta'.\label{eq:dbl_expectation}
\end{align}
By Lemma \ref{lem:Gaussian_continuity} the bounded Lipschitz metric, $d_{\mathrm{BL}}\big\{g(\cdot|\theta'),g(\cdot|\bar{\theta})\big\}$, is bounded and continuous at $\bar{\theta}$. Thus (\ref{eq:dbl_expectation})
converges to zero by Lemma \ref{prop:prep}.
\end{proof}
\begin{proposition}
Under Assumption \ref{ass2},
the map $(\theta,z)\mapsto\tilde{P}f(\theta,z)$ is continuous for every $f\in C_b(\mathbb{R}^{d+1})$.
\end{proposition}

\begin{proof}[Proof of Proposition \ref{prop:4.3}]
Without loss of generality let $\|f\|_{\infty}\leq1$, consider $(\theta^{*},z^{*})\in\Theta\times\mathbb{R}$
and denote $(\theta_{n},z_{n})_{n\in\mathbb{N}}$ a sequence converging
to $(\theta^{*},z^{*})$ as $n\rightarrow\infty$. Using the decomposition \eqref{eq:op_decomp} we have
\begin{align*}
&\left|\tilde{P}f(\theta_{n},z_{n})-\tilde{P}f(\theta^{*},z^{*})\right|\\
&=\left|\Pi f(\theta_{n},z_{n})+f(\theta_{n},z_{n})\left\{1-\Pi 1(\theta_{n},z_{n})\right\}-\Pi f(\theta^{*},z^{*})-f(\theta^{*},z^{*})\left\{1-\Pi1(\theta^{*},z^{*})\right\}\right|\\
&\leq\left|\Pi f(\theta_{n},z_{n})-\Pi f(\theta^{*},z^{*})\right|+\left|f(\theta_{n},z_{n}) - f(\theta^{*},z^{*})\right| + \left|\Pi 1(\theta_{n},z_{n})-\Pi 1(\theta^{*},z^{*})\right|
\end{align*}
By continuity of $f$ we have $f(\theta_n, z_n) \rightarrow f(\theta^*, z^*)$ as $n\rightarrow \infty$. Since $1 \in C_b(\mathbb{R}^{d+1})$ it remains to show that $\Pi f$ is continuous for every $f\in C_b(\mathbb{R}^{d+1})$. Now
\begin{align}
 & \left|\Pi f(\theta_{n},z_{n})-\Pi f(\theta^{*},z^{*})\right|\nonumber \\
 & =\bigg|\int f(\theta',z')\min\left\{1,\frac{\varphi(\theta';0,\Sigma)}{\varphi(\theta_{n};0,\Sigma)}\frac{\nu(\theta_{n}-\theta')}{\nu(\theta'-\theta_{n})}e^{z'-z_{n}}\right\}\nu(\theta'-\theta_{n})g(\mathrm{d}z'\mid\overline{\theta})\mathrm{d}\theta'\\
 & \qquad-\int f(\theta',z')\min\left\{1,\frac{\varphi(\theta';0,\Sigma)}{\varphi(\theta^{*};0,\Sigma)}\frac{\nu(\theta^{*}-\theta')}{\nu(\theta'-\theta^{*})}e^{z'-z^{*}}\right\}\nu(\theta'-\theta^{*})g(\mathrm{d}z'\mid\overline{\theta})\mathrm{d}\theta'\bigg|\nonumber \\
 & \leq\int\left|\nu(\theta'-\theta_{n})-\nu(\theta'-\theta^{*})\right|\mathrm{d}\theta'\label{eq:densitytv}\\
\begin{split}
 & \quad+\int\Bigg|\min\left\{1,\frac{\varphi(\theta';0,\Sigma)}{\varphi(\theta_{n};0,\Sigma)}\frac{\nu(\theta_{n}-\theta')}{\nu(\theta'-\theta_{n})}e^{z'-z_{n}}\right\} \\ 
& \qquad\quad -\min\left\{1,\frac{\varphi(\theta';0,\Sigma)}{\varphi(\theta^{*};0,\Sigma)}\frac{\nu(\theta^{*}-\theta')}{\nu(\theta'-\theta^{*})}e^{z'-z^{*}}\right\}\Bigg|\nu(\theta'-\theta^{*})g(\mathrm{d}z'\mid\overline{\theta})\mathrm{d}\theta'.\label{eq:mintv}
\end{split}
\end{align}
For (\ref{eq:densitytv}), Assumption \ref{ass2} implies $\nu(\theta'-\theta_{n})\rightarrow\nu(\theta'-\theta^{*})$
as $n\rightarrow\infty$ and hence Scheff\'e's lemma yields

\[
\int\left|\nu(\theta'-\theta_{n})-\nu(\theta'-\theta^{*})\right|\mathrm{d}\theta'\rightarrow0.
\]
For (\ref{eq:mintv}), the map
\[
(\theta,z)\mapsto\min\left\{1,\frac{\varphi(\theta';0,\Sigma)}{\varphi(\theta;0,\Sigma)}\frac{\nu(\theta-\theta')}{\nu(\theta'-\theta)}e^{z'-z}\right\}
\]
is continuous for all $\theta',z'$ since it is just a composition
of continuous functions. Hence,
\[
\left|\min\left\{1,\frac{\varphi(\theta';0,\Sigma)}{\varphi(\theta_{n};0,\Sigma)}\frac{\nu(\theta_{n}-\theta')}{\nu(\theta'-\theta_{n})}e^{z'-z_{n}}\right\}-\min\left\{1,\frac{\varphi(\theta';0,\Sigma)}{\varphi(\theta^{*};0,\Sigma)}\frac{\nu(\theta^{*}-\theta')}{\nu(\theta'-\theta^{*})}e^{z'-z^{*}}\right\}\right|\rightarrow0
\]
for every $(\theta',z')$ and an application of dominated convergence
shows that (\ref{eq:mintv}) goes to zero.
\end{proof}

\renewcommand{\thesection}{S3}
\section{Proofs of Section 5\label{app:clt}}
\subsection{Central Limit Theorem for Likelihood Estimators}
We detail here the proof of Theorem \ref{thm:uniform_CLT}. For clarity we explicitly state the probability space supporting all
random variables that are used to prove our limit theorem. For integers
$N,T,k$ we introduce the space $E_{T}=\Theta\times\mathbb{R}^{NTk}$
where $\Theta\subset\mathbb{R}^{d}$ is the parameter space equipped
with the Borel $\sigma$-algebra and probability measure $\mathbb{P}_{T}(\mathrm{d}\theta,\mathrm{d}u)=\pi_{T}^{\omega}(\mathrm{d}\theta)m_{T,\theta}(\mathrm{d}u).$
Finally, we will work with the Borel probability measure $\mathbb{P}$
on $E$ where $E=\mathsf{Y}^{\mathbb{N}}\times\prod_{T=1}^{\infty}E_{T},~\mathbb{P}=\mathbb{P}^{Y}\otimes\bigotimes_{T=1}^{\infty}\mathbb{P}_{T}.$

We are interested in the asymptotic distribution of the relative error
of the log-likelihood
\[
Z_{T}(\theta)=\log\widehat{p}(Y_{1:T}\mid\theta,U)-\log p(Y_{1:T}\mid\theta),
\]
where $U\sim m_{T,\theta}(\cdot)$ or $U\sim\pi_{T}^{\omega}(\cdot\mid\theta).$
Indeed, we have $\mathcal{L}\mathrm{aw}\left\{ Z_{T}(\theta)\right\} =g_{T}^{\omega}\left(\cdot\mid\theta\right)$
when $U\sim m_{T,\theta}(\cdot)$ and $\mathcal{L}\mathrm{aw}\left\{ Z_{T}(\theta)\right\} =\bar{g}_{T}^{\omega}\left(\cdot\mid\theta\right)$
when $U\sim\pi_{T}^{\omega}(\cdot\mid\theta).$ Weak convergence results for $Z_T(\theta)$ have been established in \textcite[Theorem 1]{deligiannidis2015} using a Taylor expansion. However, the CLTs introduced therein do not provide a bound on the Lipschitz metric $d_{\mathrm{BL}}$ and are not uniform in the parameter $\theta$ as required
in Assumption \ref{ass3}. In order to obtain a uniform bound for
all functions in $\mathrm{BL(\mathbb{R})}$ with $\|f\|_{\mathrm{BL}}\leq1$ and all parameter values for some neighbourhood $B(\bar{\theta})$
we need to introduce further assumptions. We follow the approach in \textcite{deligiannidis2015} and write
\begin{align*}
Z_{T}(\theta) & =\sum_{t=1}^{T}\log\left\{ 1+\frac{\widehat{p}(Y_{t}\mid\theta,U_{t})-p(Y_{t}\mid\theta)}{p(Y_{t}\mid\theta)}\right\} \\
 & =\sum_{t=1}^{T}\log\left\{ 1+\frac{\epsilon_{N}(Y_{t},\theta)}{\surd N}\right\}
\end{align*}
where
\begin{align*}
\epsilon_{N}(Y_{t},\theta) & =\frac{1}{\surd N}\sum_{i=1}^{N}\left\{ \overline{w}(Y_{t},U_{t,i},\theta)-1\right\} ,
\end{align*}
$\overline{w}(Y_{t},U_{t,i},\theta)$ being a normalized importance
weight defined in (\ref{eq:normalizedweight}).
Recall that
\begin{align*}
\sigma^{2}(y,\theta) & =E\left\{\epsilon_{T}(y,\theta)^{2}\right\}=\mathrm{Var}\left\{\overline{w}(y,U_{1,1},\theta)\right\}, \quad \sigma^{2}(\theta) =E\left\{\sigma^{2}(Y_{1},\theta)\right\}.
\end{align*}
Here the number of particles, $N$, is scaled proportionally to the
number of observations, that is $N=\left\lceil \gamma T\right\rceil $
for some $\gamma>0.$ In the following we will take $\gamma=1$ (that
is $N=T$) for simplicity and without loss of generality. In order
to show convergence of the bounded Lipschitz metric uniformly in $\theta$,
we will exploit the relation
\begin{equation*}
	\log(1+x) = x - \frac{x^2}{2} + \int_0^x \frac{u^2}{1+u} \mathrm{d}u,
\end{equation*}
where for $x<0$ we use the convention
$$\int_0^x \frac{u^2}{1+u} \mathrm{d}u = -\int_x^0 \frac{u^2}{1+u} \mathrm{d}u.$$
We thus obtain
\begin{align}\label{eq:ZT}
Z_{T}(\theta)&=\frac{1}{\surd T}\sum_{t=1}^{T}\epsilon_{T}(Y_{t},\theta)-\frac{1}{2T}\sum_{t=1}^{T}\epsilon_{T}(Y_{t},\theta)^{2}+\sum_{t=1}^{T}R_{T}(Y_{t},\theta),
\intertext{with}
R_T(y, \theta) &= \int_{0}^{\epsilon_{T}(y,\theta)/\surd T}\frac{u^{2}}{1+u}\mathrm{d}u.
\end{align}
We recall the following assumptions regarding the normalized weights.
\setcounter{assumption}{3}
\begin{assumption}
There exists a closed $\varepsilon$-ball $B(\bar{\theta})$
around $\bar{\theta}$ and a function $g$ such that the normalized
weight $\overline{w}(y,U_{1,1},\theta)$ defined in (\ref{eq:normalizedweight})
satisfies for some $0<\Delta<1$
\begin{equation*}
\sup_{\theta\in B(\bar{\theta})}E\left\{\overline{w}(y,U_{1,1},\theta)^{2+\Delta}\right\} \leq g(y),\label{eq:unif:bound}
\end{equation*}
where $U_{1,1}\sim h(\,\cdot\mid y,\theta)$ and $\mu(g)<\infty$. Additionally, $\theta\mapsto \sigma^2(y,\theta)$ is continuous in
$\theta$ on $B(\bar{\theta})$ for all $y\in\mathsf{Y}$.
\end{assumption}

We can relate expectations of powers of $\epsilon_T(y, \theta)$ to that of $\overline{w}(y,U_{1, 1},\theta)$ in the following way.

\begin{lemma}
For any $k\geq2$ and any $T\geq 1$
\begin{equation*}
	E\left\{\left|\epsilon_T(y, \theta) \right|^k \right\} \leq c(k)\left[E\left\{  \overline{w}(y,U_{1,1},\theta))^{k}\right\}+1\right]
\end{equation*}
where $c(k)$ is a constant only depending on $k$.
\end{lemma}

\begin{proof}
This is Lemma 2 in \textcite{deligiannidis2015}. We repeat it here for convenience. It holds
\begin{align*}
 E\left\{\left|\epsilon_T(y, \theta) \right|^k \right\} & = E\left[\left|\frac{1}{\surd T}\sum_{i=1}^{T}\left\{ \overline{w}(y,U_{1,i},\theta)-1\right\} \right|^{k}\right]\\
 & \leq c_{1}(k)E\left[ \left|\frac{1}{T}\sum_{i=1}^{T}\left\{  \overline{w}(y,U_{1,i},\theta)-1\right\} ^{2}\right|^{k/2}\right] \\
 & \leq c_{1}(k)\frac{1}{T}\sum_{i=1}^{T}E\left\{ \left| \overline{w}(y,U_{1,i},\theta)-1\right|^{k}\right\} \\
 & \leq c_{1}(k)c_{2}(k)\left[E\left\{  \overline{w}(y,U_{1,1},\theta))^{k}\right\}+1\right]
\end{align*}
for some constants $c_{1}(k),c_{2}(k)$ by application of the Marcinkiewicz--Zygmund, Jensen and $c_r$-inequalities.
\end{proof}
As a result we have thus
\begin{equation}\label{eq:momentepsilon}
	\sup_{\theta\in B(\bar{\theta})}E\left\{\left|\epsilon_T(y, \theta) \right|^k \right\} \leq c(k)\sup_{\theta\in B(\bar{\theta})}\left[E\left\{  \overline{w}(y,U_{1,1},\theta))^{k}\right\} +1\right]
\end{equation}
and the left-hand-side is finite whenever the right-hand-side is finite.

\subsection{Moment Conditions for Weak Convergence}

Denote $\mathcal{Y}_{T}$ the $\sigma-$algebra spanned by the data $Y_{1:T}=(Y_{1},\ldots,Y_{T})$
observed up to $T$.

\setcounter{theorem}{2}
\begin{theorem}[Moment conditions for UCLT]
\label{thm:moment_uclt}
Under Assumption \ref{ass4} we have the following uniform central limit theorems
\begin{itemize}
\item[a)]
\[
\sup_{\theta\in B(\bar{\theta})}d_\mathrm{BL}\left[g_{T}^{\omega}(\cdot\mid\theta),\varphi\left\{ \cdot;-\sigma^{2}(\theta)/2,\sigma^{2}(\theta)\right\} \mid\mathcal{Y}_{T}\right]\overset{\mathbb{P}}{\rightarrow}0,
\]
and
\item[b)]
\[
\sup_{\theta\in B(\bar{\theta})}d_\mathrm{BL}\left[\bar{g}_{T}^{\omega}(\cdot\mid\theta),\varphi\left\{ \cdot;\sigma^{2}(\theta)/2,\sigma^{2}(\theta)\right\} \mid\mathcal{Y}_{T}\right]\overset{\mathbb{P}}{\rightarrow}0.
\]
\end{itemize}
\end{theorem}

We will need the following auxiliary results.
\begin{lemma}
\label{lem:bl_bound} Let $S_{T}(\theta)=\sum_{i=1}^{T}\xi_{i}(\theta)$
denote the sum of zero mean independent random variables $\xi_{1}(\theta),\ldots,\xi_{T}(\theta)$
such that $\mathrm{Var}(S_{T})=1$.
Then for any Lipschitz function $f$ with Lipschitz constant $L$ and $Z\sim \mathcal{N}(0,1)$
\[
\left|E\left[f\left\{ S_{T}(\theta)\right\} -f\left(Z\right)\right]\right|\leq L\left(4E\left[\sum_{i=1}^{T}\xi_{i}^{2}(\theta)1_{\{\left|\xi_{i}(\theta)\right|>1\}}\right]+3E\left[\sum_{i=1}^{T}|\xi_{i}(\theta)|^{3}1_{\{\left|\xi_{i}(\theta)\right|\leq1\}}\right]\right).
\]
\end{lemma}

\begin{proof}
This is Theorem 3.2 in \textcite{chen2010normal}.
\end{proof}
The above result reduces the problem of showing weak convergence uniformly over some neighbourhood $B(\bar{\theta})$ to uniform laws of large numbers for conditional higher order moments.
Conditions to ensure uniformity in the convergence of averages are
widely established. We will use the following result given in
\autocite[Theorem 2]{jennrich1969}.
\begin{lemma}
\label{lem:uniformconv} Let $A\subset\mathbb{R}^{d}$ be compact and
let $f\colon\mathbb{R}^{k}\times A\rightarrow\mathbb{R}$ be continuous
in $\theta$ for each $y\in\mathbb{R}^{k}$ and measurable in y for
each $\theta\in A$. Further assume that there exists an integrable
function $g$, such that $|f(y,\theta)|\leq g(y)$ for all $y$ and
$\theta$. For independent random variables $Y_{i}\sim\mu$ $(i = 1, \ldots, T)$ then \textup{$\mathbb{P}^{Y}$}-almost surely
\[
\sup_{\theta\in A}\left|\frac{1}{T}\sum_{t=1}^{T}f\left(Y_{t},\theta\right)-E\left\{f(Y_{1},\theta)\right\}\right|\rightarrow0,
\]
as $T\rightarrow \infty$.
\end{lemma}
Before we proceed with the proof of Theorem \ref{thm:moment_uclt}, we note that Lemma \ref{lem:bl_bound} is not formulated in terms of conditional laws. However, considering conditionally (upon $\mathcal{Y}_{T}$)
centred and independent random variables ${\xi}_{T,1},\ldots,{\xi}_{T,T}$ such that $\sum_{i=1}^{T}\mathrm{Var}\left\{ \xi_{i}(\theta)|Y_{1:T}\right\} =1$, we can apply the above lemma for every realization $Y_{1:T}=y_{1:T}$.
Denote $P_{T}^{y}$ a regular conditional distribution associated
with the law of $S_{T}={\xi}_{T,1}+\ldots+{\xi}_{T,T}$ given
$Y_{1:T}=y_{1:T}.$
By applying Lemma \ref{lem:bl_bound}, we get
\begin{align}
 \notag & d_{\mathrm{BL}}\big\{P_{T}^{y},\varphi(\,\cdot\,;0,1)\mid Y_{1:T} = y_{1:T}\big\} \\
 & \leq 4E\left[\sum_{i=1}^{T}\xi_{i}^{2}(\theta)1_{\{\left|\xi_{i}(\theta)\right|>1\}}\mid Y_{1:T}=y_{1:T}\right]+3E\left[\sum_{i=1}^{T}|\xi_{i}(\theta)|^{3}1_{\{\left|\xi_{i}(\theta)\right|\leq1\}}\mid Y_{1:T}=y_{1:T}\right].
\end{align}
Thus, if the terms on the r.h.s. go to zero in $\mathbb{P}^Y$-probability then $d_{\mathrm{BL}}\big\{P_{T}^{Y},\varphi(\,\cdot\,;0,1)\big\}\overset{\mathbb{P}^{Y}}{\longrightarrow}0$.
With this reasoning we can apply Lemma \ref{lem:bl_bound} to
prove Theorem \ref{thm:moment_uclt}.
\begin{proof}[Proof of Theorem \ref{thm:moment_uclt}, part \emph{a)}]
Define
\[
\xi_{T,t}(\theta)=\frac{\epsilon_{T}(Y_{t},\theta)}{\surd T \sigma_T(Y_{1:T}, \theta)},\quad S_{T}(\theta)=\sum_{t=1}^T \xi_{T,t}(\theta),
\]
where
\begin{align}\label{eq:condvarST}
\sigma_{T}^{2}(Y_{1:T},\theta) & =\frac{1}{T}\sum_{t=1}^{T}\mathrm{Var}\left\{ \epsilon_{T,t}(\theta)\mid\mathcal{\mathcal{Y}}_{T}\right\} .
\end{align}
Thus
\[
\mathrm{Var}  \left\{S_{T}(\theta)\mid\mathcal{\mathcal{Y}}_{T}\right\}=\sum_{t=1}^T \mathrm{Var} \left\{\xi_{T,t}(\theta)\right\}=1.
\]

In the following we will use the shorthands $\sigma_{T}(Y_{1:T},\theta)=\left\{ \sigma_{T}^{2}(Y_{1:T},\theta)\right\} ^{1/2}$
and $\sigma_{T}^{r}(Y_{1:T},\theta)=\left\{ \sigma_{T}^{2}(Y_{1:T},\theta)\right\} ^{r/2}$
for any real value $r$.

Then $S_{T}(\theta)$ fulfils the conditions of Lemma \ref{lem:bl_bound} conditionally on $\mathcal{Y}_T$.
The random variable $Z_{T}(\theta)$ defined in (\ref{eq:ZT}) can be rewritten as
\[
Z_{T}(\theta)=S_{T}(\theta)\sigma_{T}(Y_{1:T},\theta)-\frac{1}{2T}\sum_{t=1}^{T}\epsilon_{T}(Y_{t},\theta)^{2}+\sum_{t=1}^{T}R_{T}(Y_{t},\theta).
\]
We have for $Z\sim\mathcal{N}(0,1)$
\begin{align}
 & \sup_{\theta\in B(\bar{\theta})}d_\mathrm{BL}\left[\mathcal{L}aw\left\{ Z_{T}(\theta)\right\} ,\varphi\left\{ \cdot;-\sigma^{2}(\theta)/2,\sigma^{2}(\theta)\right\}\mid \mathcal{Y}_T \right]\nonumber \\
 & =\sup_{\theta\in B(\bar{\theta})}d_\mathrm{BL}\left[\mathcal{L}aw\left\{ Z_{T}(\theta)\right\} ,\mathcal{L}aw\left\{ Z\sigma(\theta)-\frac{\sigma^{2}(\theta)}{2}\right\}\mid \mathcal{Y}_T \right]\nonumber \\
 & =\sup_{\theta\in B(\bar{\theta})}\sup_{\substack{f\in\mathrm{BL}(\mathbb{R})\\
\|f\|_{\mathrm{BL}}\leq1 }}\Bigg|E\left[f\left\{ S_{T}(\theta)\sigma_{T}(Y_{1:T},\theta)-\frac{1}{2T}\sum_{t=1}^{T}\epsilon_{T}(Y_{t},\theta)^{2}+\sum_{t=1}^{T}R_{T}(Y_{t},\theta)\right\} \mid\mathcal{Y}_{T}\right] \\
& \qquad -E\left[f\left\{ Z\sigma(\theta)-\frac{\sigma^{2}(\theta)}{2}\right\} \right]\Bigg|\nonumber \\
 & \leq\sup_{\theta\in B(\bar{\theta})}\sup_{\substack{f\in\mathrm{BL}(\mathbb{R})\\
\|f\|_{\mathrm{BL}}\leq1}}\Bigg|E\left[f\left\{ S_{T}(\theta)\sigma_{T}(Y_{1:T},\theta)-\frac{1}{2T}\sum_{t=1}^{T}\epsilon_{T}(Y_{t},\theta)^{2}+\sum_{t=1}^{T}R_{T}(Y_{t},\theta)-\frac{\sigma^{2}(\theta)}{2}+\frac{\sigma^{2}(\theta)}{2}\right\} \mid\mathcal{Y}_{T}\right]\label{eq:clt_in1}\\
 & \text{\ensuremath{\quad\quad-E\left[f\left\{ S_{T}(\theta)\sigma_{T}(Y_{1:T},\theta)+\sum_{t=1}^{T}R_{T}(Y_{t},\theta)-\frac{\sigma^{2}(\theta)}{2}\right\} \mid\mathcal{Y}_{T}\right]}}\Bigg|\nonumber \\
 & \quad+\sup_{\theta\in B(\bar{\theta})}\sup_{\substack{f\in\mathrm{BL}(\mathbb{R})\\
\|f\|_{\mathrm{BL}}\leq1 }}\left|E\left[f\left\{ S_{T}(\theta)\sigma_{T}(Y_{1:T},\theta)+\sum_{t=1}^{T}R_{T}(Y_{t},\theta)-\frac{\sigma^{2}(\theta)}{2}\right\} \mid\mathcal{Y}_{T}\right]\right. \notag\\
&\qquad\qquad \left.-E\left[f\left\{ S_{T}(\theta)\sigma_{T}(Y_{1:T},\theta)-\frac{\sigma^{2}(\theta)}{2}\right\} \mid\mathcal{Y}_{T}\right]\right|\label{eq:clt_in2}\\
 & \quad+\sup_{\theta\in B(\bar{\theta})}\sup_{\substack{f\in\mathrm{BL}(\mathbb{R})\\
\|f\|_{\mathrm{BL}}\leq1}}\left|E\left[f\left\{ S_{T}(\theta)\sigma_{T}(Y_{1:T},\theta)-\frac{\sigma^{2}(\theta)}{2}\right\} \mid\mathcal{Y}_{T}\right]-E\left[f\left\{ Z\sigma(\theta)-\frac{\sigma^{2}(\theta)}{2}\right\} \right]\right|\label{eq:clt_in3}
\end{align}

Now we have for (\ref{eq:clt_in1})
\begin{align*}
\eqref{eq:clt_in1} & \leq\sup_{\theta\in B(\bar{\theta})}\sup_{\substack{f\in\mathrm{BL}(\mathbb{R})\\
\|f\|_{\mathrm{BL}}\leq1 }}\Bigg|E\Bigg[f\Bigg\{ S_{T}(\theta)\sigma_{T}(Y_{1:T},\theta)-\frac{1}{2T}\sum_{t=1}^{T}\epsilon_{T}(Y_{t},\theta)^{2}\\
 & \qquad\qquad +\sum_{t=1}^{T}R_{T}(Y_{t},\theta)-\frac{\sigma^{2}(\theta)}{2}+\frac{\sigma^{2}(\theta)}{2}\Bigg\} \mid\mathcal{Y}_{T}\Bigg]\\
 & \text{\ensuremath{\qquad-E\left[f\left\{ S_{T}(\theta)\sigma_{T}(Y_{1:T},\theta)+\sum_{t=1}^{T}R_{T}(Y_{t},\theta)-\frac{\sigma^{2}(\theta)}{2}\right\} \mid\mathcal{Y}_{T}\right]}}\Bigg|\\
 & \leq \sup_{\theta\in B(\bar{\theta})}E\left[\min\left\{ 1,\left|\frac{\sigma^{2}(\theta)}{2}-\frac{1}{2T}\sum_{t=1}^{T}\epsilon_{T}(Y_{t},\theta)^{2}\right|\right\} \mid\mathcal{Y}_{T}\right] %\\
\end{align*}
where we use that $f$ is bounded and Lipschitz.

We can bound this term by
\begin{align}
	\notag& \sup_{\theta\in B(\bar{\theta})}E\left(\min\left\{1,\left|\frac{\sigma^{2}(\theta)}{2}-\frac{1}{2T}\sum_{t=1}^{T}\epsilon_{T}(Y_{t},\theta)^{2}\right|\right\}\Bigg|\mathcal{Y}_{T}\right) \\
	\label{decomposition1stterm}
\begin{split}
	& \leq \sup_{\theta\in B(\bar{\theta})}E\left(\min\left\{1,\left|\frac{1}{2T}\sum_{t=1}^{T}\left\{\sigma^2(Y_{t},\theta) - \sigma^2(\theta)\right\}\right|\right\}\Bigg|\mathcal{Y}_{T}\right)\\
	&\qquad\qquad  + \sup_{\theta\in B(\bar{\theta})}E\left(\min\left\{1,\left|\frac{1}{2T}\sum_{t=1}^{T}\left\{\epsilon_{T}(Y_{t},\theta)^{2}-\sigma^2(Y_t, \theta)\right\}\right|\right\}\Bigg|\mathcal{Y}_{T}\right).
\end{split}
\end{align}
For any $0<\delta<1$, we can bound the first term on the r.h.s. of \eqref{decomposition1stterm} by
\begin{align*}
 & \sup_{\theta\in B(\bar{\theta})}E\left[\min\left\{1,\left|\frac{1}{2T}\sum_{t=1}^{T}\left\{\epsilon_{T}(Y_{t},\theta)^{2}-\sigma^2(Y_t, \theta)\right\}\right| \right\}\Bigg|\mathcal{Y}_T\right]\\
 & \leq \sup_{\theta\in B(\bar{\theta})}\left(E\left[\min\left\{1,\left|\frac{1}{2T}\sum_{t=1}^{T}\left\{\epsilon_{T}(Y_{t},\theta)^{2}-\sigma^2(Y_t, \theta)\right\}\right|^{1+\delta}\right\}\Bigg|\mathcal{Y}_{T}\right]\right)^{\frac{1}{1+\delta}} \\
 & \leq \left[\frac{C}{2^{1+\delta}T^{1+\delta}}\sum_{t=1}^{T}\sup_{\theta\in B(\bar{\theta})}E\left\{\left|\epsilon_{T}(Y_{t},\theta)^{2}-\sigma^2(Y_t, \theta)\right|^{1+\delta}\Bigg| \mathcal{Y}_T\right\}\right]^{\frac{1}{1+\delta}}\\
 &\leq  C
\left[\frac{C'}{2^{1+\delta}T^{1+\delta}}\sum_{t=1}^{T}
\left\{1+g(Y_t)\right\}\right]^{\frac{1}{1+\delta}}
  \rightarrow 0
\end{align*}
in $\mathbb{P}^Y$-probability by the law of large numbers using, in turn, Jensen's inequality, von Bahr--Esseen inequality \autocite{vonbahr1965} as $E\left\{\epsilon_T(Y_t, \theta)^2\mid \mathcal{Y}_T\right\} = \sigma^2(Y_t, \theta)$, $c_r$-inequality, \eqref{eq:momentepsilon} and Assumption \ref{ass4} for $\Delta=2\delta$, noting that
\begin{align*}
	\sigma^2(Y_t, \theta) &= E\left\{\epsilon_T(Y_t, \theta)^2\mid \mathcal{Y}_T \right\} \leq E\left\{|\epsilon_T(Y_t, \theta)|^{2+\Delta} \right\}^{2/(2+\Delta)} \\
	& \leq C \cdot \left\{g(Y_t) + 1\right\}^{2/(2+\Delta)},
\end{align*}
where the last inequality is due to \eqref{eq:momentepsilon}.
The second term on  the right-hand side of \eqref{decomposition1stterm} can be bounded
\begin{align*}
 & \sup_{\theta\in B(\bar{\theta})}E\left(\min\left\{1,\left|\frac{1}{2T}\sum_{t=1}^{T}\left\{\sigma^2(Y_{t},\theta) - \sigma^2(\theta)\right\}\right|\right\}\mid \mathcal{Y}_T\right) \\
 & \leq E\left(\min\left\{1,\sup_{\theta\in B(\bar{\theta})}\left|\frac{1}{2T}\sum_{t=1}^{T}\left\{\sigma^2(Y_{t},\theta) - \sigma^2(\theta)\right\}\right|\right\}\mid \mathcal{Y}_T\right).
\end{align*}
Noting that $\sigma^2(y,\theta)$ is continuous in $\theta$ for all $y$ by Assumption \ref{ass4} and $\sigma^2(y, \theta) \leq C\cdot \left\{1 + g(y) \right\}^{2/(2+\Delta)}$ we can apply Lemma \ref{lem:uniformconv} to get
\begin{align*}
	\sup_{\theta\in B(\bar{\theta})}\left|\frac{1}{2T}\sum_{t=1}^{T}\left\{\sigma^2(Y_{t},\theta) - \sigma^2(\theta)\right\}\right| \overset{\mathbb{P}^Y}{\rightarrow} 0
\end{align*}
and we can use dominated convergence to conclude that
\begin{equation*}
	E\left(E\left[\min\left\{1,\sup_{\theta\in B(\bar{\theta})}\left|\frac{1}{2T}\sum_{t=1}^{T}\left\{\sigma^2(Y_{t},\theta) - \sigma^2(\theta)\right\}\right|\right\}\mid \mathcal{Y}_T \right]\right) \rightarrow 0
\end{equation*}
and thus
\begin{equation*}
		E\left[\min\left\{1,\sup_{\theta\in B(\bar{\theta})}\left|\frac{1}{2T}\sum_{t=1}^{T}\left\{\sigma^2(Y_{t},\theta) - \sigma^2(\theta)\right\}\right|\right\}\mid \mathcal{Y}_T\right] \overset{\mathbb{P}^Y}{\rightarrow}0.
\end{equation*}

The quantity (\ref{eq:clt_in2}) can be upper bounded by
\begin{equation}
	\label{triangleRT}\eqref{eq:clt_in2}\leq E\left[\min\left\{ 1,\left|\sum_{t=1}^{T}R_{T}(Y_{t},\theta)\right|\right\} \mid\mathcal{Y}_{T}\right] \leq \sum_{t=1}^{T}E\left[\min\left\{ 1,\left|R_{T}(Y_{t},\theta)\right|\right\} \mid\mathcal{Y}_{T}\right].
\end{equation}
We will split the expectation into two terms
\begin{align}\label{eq:RTdecomposition}
	\notag & E\left[\min\left\{ 1,\left|R_{T}(Y_{t},\theta)\right|\right\} \mid\mathcal{Y}_{T}\right] \\
	& = E\left[\min\left\{ 1,\left|R_{T}(Y_{t},\theta)\right|\right\}1_{\left\{\left|\frac{\epsilon_T(Y_t,\theta)}{\surd T}\right| \leq 1 \right\}} \mid\mathcal{Y}_{T}\right] + E\left[\min\left\{ 1,\left|R_{T}(Y_{t},\theta)\right|\right\}1_{\left\{\left|\frac{\epsilon_T(Y_t,\theta)}{\surd T}\right| > 1 \right\}} \mid\mathcal{Y}_{T}\right].
\end{align}
Recall
\begin{align*}
	R_T(y, \theta) = \int_{0}^{\epsilon_{T}(y,\theta)/\surd T}\frac{u^{2}}{1+u}\mathrm{d}u.
\end{align*}
We investigate the integral
\begin{equation}\label{eq:psifunction}
	\Psi(x)=\int_0^x \frac{u^2}{1+u} \mathrm{d}u
\end{equation}
in more detail (see also Figure \ref{fig:func}), where in the case $x<0$, we interpret the above as an integral over the interval $[x,0]$. On the interval $(-1,1]$ we can bound the function
\begin{equation*}
	\frac{u^2}{1+u} \leq \frac{|u|^{1+\Delta}}{1+u},
\end{equation*}
as $0<\Delta<1$ where we show $\Delta = 0.1$ as an example in Figure \ref{fig:func}. Subsequently, we bound for $x\in (-1,1]$
\begin{equation*}
	\left|\int_0^x \frac{u^2}{1+u} \mathrm{d}u \right| \leq \left|\int_0^x \frac{|u|^{1+\Delta}}{1+u} \mathrm{d}u\right| \leq \left| x \cdot \frac{|x|^{1+\Delta}}{1+x} \right|,
\end{equation*}
i.e. the box containing the area under the curve. This is visualized in Figure \ref{fig:func}. The integral (shaded blue) is bounded by the striped box.
Hence, on the set $|\epsilon_T(y, \theta)/\surd T| \leq 1$
\begin{equation*}
	\left|\int_0^{\epsilon_T(y, \theta)/\surd T} \frac{u^2}{1+u} \mathrm{d}u\right| \leq
	\left| \frac{|\epsilon_T(y, \theta)|^{2+\Delta}}{T^{1+\Delta/2}}\frac{1}{1+\epsilon_T(y, \theta)/\surd T}  \right|.
\end{equation*}

\begin{figure}
	\centering
	\includegraphics[width=.5\textwidth]{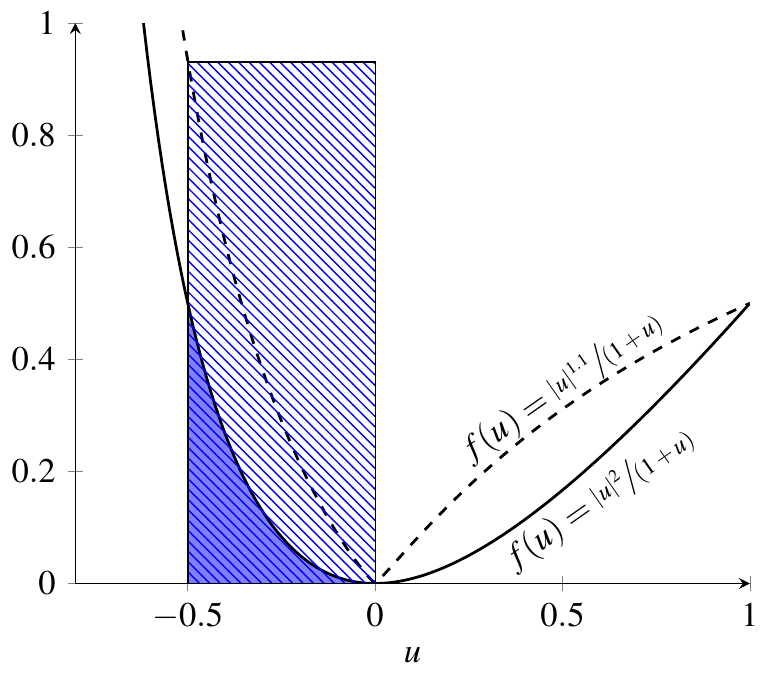}
	\caption{For $x\in (-1,1]$, $x=-0.5$ on the graph, the remainder of our expansion is estimated by the integral under the solid curve (blue shaded area). We bound this integral first by dashed line and then we approximate the integral by the box containing the area (lines).}
	\label{fig:func}
\end{figure}

For any non-negative random variable $X$ and event $A$, we have the identity $$E\left\{\min(1,X) 1_A\right\}\leq E\left\{X 1_{X\leq 1} 1_A\right\}+\mathbb{P}(X>1),$$ so we can bound the first term on the right-hand side of \eqref{eq:RTdecomposition} for every $t=1,\ldots, T$
\begin{align*}
	& E\left[\min\left\{ 1,\left|R_{T}(Y_t,\theta)\right|\right\}1_{\left\{\left|\frac{\epsilon_T(Y_t,\theta)}{\surd T}\right| \leq 1 \right\}} \mid \mathcal{Y}_T \right] \\
	& \leq E\left[\left|\frac{\epsilon_T^{2+\Delta}(Y_t, \theta)/T^{1+\Delta/2}}{1 + \epsilon_T(Y_t, \theta)/\surd T}\right| 1_{\left\{\left|\frac{\epsilon_T^{2+\Delta}(Y_t, \theta)/T^{1+\Delta/2}}{1 + \epsilon_T(Y_t, \theta)/\surd T}\right| \leq 1 \right\}}1_{\left\{\left|\frac{\epsilon_T(Y_t,\theta)}{\surd T}\right| \leq 1 \right\}}\Bigg| \mathcal{Y}_T\right]  \\
	& \qquad + \mathbb{P}\left\{\left|\frac{\epsilon_T^{2+\Delta}(Y_t, \theta)/T^{1+\Delta/2}}{1 + \epsilon_T(Y_t, \theta)/\surd T}\right| > 1\Bigg| \mathcal{Y}_T\right\}.
\end{align*}
By inspection of the function, similarly to before,
\begin{equation*}
	u \mapsto \frac{|u|^{2+\Delta}}{1+u}
\end{equation*}
one can easily verify that there exist $0 < \delta_1 < 1$  and $\delta_2 >0$ such that
\begin{equation*}
	\frac{|u|^{2+\Delta}}{1+u} \leq 1 \Leftrightarrow -\delta_1 \leq u \leq \delta_2.
\end{equation*}
Thus we have
\begin{align}
	\notag & E\left[\left|\frac{\epsilon_T^{2+\Delta}(Y_t, \theta)/T^{1+\Delta/2}}{1 + \epsilon_T(Y_t, \theta)/\surd T}\right| 1_{\left\{\frac{\epsilon_T^{2+\Delta}(Y_t, \theta)/T^{1+\Delta/2}}{1 + \epsilon_T(Y_t, \theta)/\surd T} \leq 1 \right\}} 1_{\left\{\left|\frac{\epsilon_T(Y_t,\theta)}{\surd T}\right| \leq 1 \right\}}\Bigg| \mathcal{Y}_T\right] \\
 \notag & \leq E\left[\left|\frac{\epsilon_T^{2+\Delta}(Y_t, \theta)/T^{1+\Delta/2}}{1 + \epsilon_T(Y_t, \theta)/\surd T}\right| 1_{\left\{-\delta_1 \leq \epsilon_T(Y_t, \theta)/\surd T \leq \delta_2 \right\}}\Bigg| \mathcal{Y}_T\right] \\
	\label{boundRT1}& \leq \frac{1}{(1-\delta_1)T^{1+\Delta/2}} E\left[\left|\epsilon_T(Y_t, \theta)\right|^{2+\Delta}\Big| \mathcal{Y}_T\right],
\end{align}
while
\begin{align}
	\notag& \mathbb{P}\left\{\left|\frac{\epsilon_T^{2+\Delta}(Y_t, \theta)/T^{1+\Delta/2}}{1 + \epsilon_T(Y_t, \theta)/\surd T}\right| > 1\Bigg| \mathcal{Y}_T\right\} \\
	\notag& \leq \mathbb{P}\left\{\left|\frac{\epsilon_T(Y_t, \theta)}{T^{1/2}}\right| > \min\{\delta_1, \delta_2\}\Big| \mathcal{Y}_T\right\} \\
	\label{boundRT2} & \leq \frac{1}{\min\{\delta_1, \delta_2\}^{2+\Delta} T^{1+\Delta/2}} E\left[|\epsilon_T(Y_t, \theta)|^{2+\Delta}\Big| \mathcal{Y}_T\right].
\end{align}
The second term on the right-hand side of \eqref{eq:RTdecomposition} is bounded by
\begin{align}
	\notag& E\left[\min\left\{ 1,\left|R_{T}(Y_{t},\theta)\right|\right\}1_{\left\{\left|\frac{\epsilon_T(Y_t,\theta)}{\surd T}\right| > 1 \right\}} \mid\mathcal{Y}_{T}\right] \\
	\label{eq:RT2}& \leq E\left[\left|R_{T}(Y_{t},\theta)\right|1_{\left\{\left|\frac{\epsilon_T(Y_t,\theta)}{\surd T}\right| > 1 \right\}} \mid\mathcal{Y}_{T}\right].
\end{align}
As $\epsilon_T(Y_t,\theta)/\surd T\geq-1$, \eqref{eq:RT2} is null for $\epsilon_T(Y_t,\theta)<-1$ so writing $X^+ = \max\{0, X\}$ this can be rewritten as
\begin{align*}
& E\left\{\int_{0}^{\epsilon_{T}(Y_t,\theta)/\surd T}\frac{u^{2}}{1+u}\mathrm{d}u 1_{\{\epsilon_{T}(Y_t,\theta)/\surd T \geq 1\}} \Big|\mathcal{Y}_T \right\} \\
& \leq E\left[ \int_0^{(\epsilon_{T}(Y_t,\theta)/\surd T)^+} \frac{u^2}{1+u} \mathrm{d}u \Big|\mathcal{Y}_T\right] \\
& =\int_{0}^{\infty}\frac{u^2}{1+u}\mathbb{P}\left\{\epsilon_{T}(Y_t,\theta)^+ > \surd Tu \Big|\mathcal{Y}_T\right\}\mathrm{d}u,
\end{align*}
where we have used that for the function \eqref{eq:psifunction} is increasing and differentiable on its domain so
\[
E\left\{ \Psi\left(\left|X\right|\right)\right\} =\Psi(0)+\int_{0}^{\infty}\Psi'(u)P(\left|X\right|>u)\mathrm{d}u.
\]
For $\Delta\in(0,1)$, we bound the remainder using
\begin{align}
 \notag & = \int_{0}^{\infty}\frac{u^{2}}{1+u}\mathbb{P}\left\{\epsilon_{T}(Y_t,\theta)^+ > \surd Tu \Big|\mathcal{Y}_T \right\}\mathrm{d}u\\
 \notag  & \leq \int_{0}^{\infty}\frac{u^{2}}{1+u}\frac{E\left\{ \left|\varepsilon_{T}(Y_t,\theta)\right|^{2+\Delta}\Big|\mathcal{Y}_T\right\} }{T^{(2+\Delta)/2}u^{2+\Delta}}\mathrm{d}u\\
  \notag & =\int_{0}^{\infty}\frac{1}{\left(1+u\right)u^{\Delta}}\mathrm{d}u\frac{1}{T^{1+\Delta/2}}E\left\{ \left|\varepsilon_{T}(Y_t,\theta)\right|^{2+\Delta}\Big|\mathcal{Y}_T\right\}\\
  \label{boundRT3} & =C(\Delta)\frac{E\left\{ \left|\varepsilon_{T}(Y_t,\theta)\right|^{2+\Delta}\Big|\mathcal{Y}_T\right\}}{T^{1+\Delta/2}}
\end{align}
noting that
\[
\int_{0}^{\infty}\frac{1}{\left(1+u\right)u^{\Delta}}du=C(\Delta)<\infty
\]
for $\Delta\in(0,1)$. Hence we can bound \eqref{eq:RTdecomposition} by the sum of \eqref{boundRT1}, \eqref{boundRT2} and \eqref{boundRT3} so, by using \eqref{triangleRT}, we obtain a bound for \eqref{eq:clt_in2}
\begin{align*}
	\eqref{eq:clt_in2} & \leq \frac{1}{(1-\delta_1)T^{(1+\Delta)/2}} \sum_{t=1}^{T}\sup_{\theta \in B(\bar{\theta})}E\left\{ \left|\varepsilon_{T}(Y_t,\theta)\right|^{2+\Delta}
	\Big|\mathcal{Y}_T\right\} \\
	& \qquad + \frac{1}{\min\{\delta_1, \delta_2\}^{2+\Delta} T^{(1+\Delta)/2}}\sum_{t=1}^T \sup_{\theta \in B(\bar{\theta})}E\left\{ \left|\varepsilon_{T}(Y_t,\theta)\right|^{2+\Delta}
		\Big|\mathcal{Y}_T\right\}\\
	& \qquad + C(\Delta)\frac{1}{T^{1+\Delta/2}}\sum_{t=1}^{T}\sup_{\theta \in B(\bar{\theta})}E\left\{ \left|\varepsilon_{T}(Y_t,\theta)\right|^{2+\Delta}	\Big|\mathcal{Y}_T\right\}
	\rightarrow 0
\end{align*}
which all converge in $\mathbb{P}^Y$-probability  by \eqref{eq:momentepsilon}, Assumption \ref{ass4} and the law of large numbers.

We are now going to bound \eqref{eq:clt_in3}. We will use the fact that any constant $c$ and any two random variables $X_1, X_2$ we have for $c>0$
\begin{align*}
	\sup_{\substack{f\in\mathrm{BL}(\mathbb{R})\\
\|f\|_{\mathrm{BL}}\leq1}} \left|E\left[f(cX_1) - f(cX_2)\right]\right|&\leq \sup_{\substack{f\in\mathrm{BL}(\mathbb{R})\\
\|f\|_{\mathrm{BL}}\leq c}} \left|E\left[f(X_1) - f(X_2)\right]\right| \\
	& = \sup_{\substack{f\in\mathrm{BL}(\mathbb{R})\\
\|f\|_{\mathrm{BL}}\leq c}} \left|E\left[c\left\{\frac{f(X_1)}{c} - \frac{f(X_2)}{c}\right\}\right]\right| \\
	& \leq c\cdot \sup_{\substack{f\in\mathrm{BL}(\mathbb{R})\\
\|f\|_{\mathrm{L}}\leq 1}} \left|E\left[\left\{f(X_1) - f(X_2)\right\}\right] \right|.
\end{align*}
Note that we only require $\|f\|_\mathrm{L} \leq 1$ ($\|f\|_\mathrm{L}$ denoting the Lipschitz constant) in the last line alleviating the bound on the supremum $\|f\|_\infty$. The aim of the following paragraphs is to apply the above inequality and Lemma \ref{lem:bl_bound} to find a bound on (\ref{eq:clt_in3}). Omitting for the moment the supremum over the set $B(\bar{\theta})$ we compute for (\ref{eq:clt_in3})
\begin{align*}
 &\sup_{\substack{f\in\mathrm{BL}(\mathbb{R})\\
\|f\|_{\mathrm{BL}}\leq1}}\left|E\left[f\left\{ S_{T}(\theta)\sigma_{T}(Y_{1:T},\theta)\right\} \mid\mathcal{Y}_{T}\right]-E\left[f\left\{ Z\sigma(\theta)\right\} \mid\mathcal{Y}_{T}\right]\right|\\
 & \leq \sup_{\substack{f\in\mathrm{BL}(\mathbb{R})\\
\|f\|_{\mathrm{BL}}\leq1}}\left|E\left[f\left\{ S_{T}(\theta)\sigma_{T}(Y_{1:T},\theta)\right\} \mid\mathcal{Y}_{T}\right]-E\left[f\left\{ Z\sigma_{T}(Y_{1:T},\theta)\right\} \mid\mathcal{Y}_{T}\right]\right| \\
 & \qquad +\sup_{\substack{f\in\mathrm{BL}(\mathbb{R})\\
\|f\|_{\mathrm{BL}}\leq1}}\left|E\left[f\left\{ Z\sigma_{T}(Y_{1:T},\theta)\right\} \mid\mathcal{Y}_{T}\right]-E\left[f\left\{ Z\sigma(\theta)\right\} \right]\right|\\
 & \leq\sigma_{T}(Y_{1:T},\theta)\sup_{\substack{f\in\mathrm{BL}(\mathbb{R})\\ \|f\|_{\mathrm{L}}\leq1}}\left|E\left[f\left\{ S_{T}(\theta)\right\} \mid\mathcal{Y}_{T}\right]-E\left[f\left\{ Z\right\} \right]\right|+E\left[|Z|\right]\left|\sigma_{T}(Y_{1:T},\theta)-\sigma(\theta)\right|\\
 & \leq\sigma_{T}(Y_{1:T},\theta)\sup_{\substack{f\in\mathrm{BL}(\mathbb{R})\\
\|f\|_{\mathrm{L}}\leq1 }}\left|E\left[f\left\{ S_{T}(\theta)\right\} \mid\mathcal{Y}_{T}\right]-E\left[f\left\{ Z\right\} \right]\right| + \left(\frac{2}{\pi}\right)^{1/2}\left|\sigma_{T}(Y_{1:T},\theta)-\sigma(\theta)\right|, \numberthis \label{eq:three_eight}
\end{align*}
We have already shown
\begin{align*}
  \sup_{\theta\in B(\bar{\theta})}\left|\sigma_{T}^{2}(Y_{1:T},\theta)-\sigma^{2}(\theta)\right|
  =\sup_{\theta\in B(\bar{\theta})}\left|\sum_{t=1}^{T}\frac{\sigma^{2}\left(Y_{t},\theta\right)}{T}-\sigma^{2}(\theta)\right|\overset{\mathbb{P}^Y}{\longrightarrow}0,
\end{align*}
by the uniform law of large numbers (Lemma \ref{lem:uniformconv}). Using $|\surd{a}-\surd{b}| \leq \surd{|a-b|}$, we have
\[
\sup_{\theta\in B(\bar{\theta})}\left|\sigma_{T}(Y_{1:T},\theta)-\sigma(\theta)\right|\overset{\mathbb{P}^{Y}}{\longrightarrow}0.
\]
For the first part of \eqref{eq:three_eight}, by Lemma~\ref{lem:bl_bound} applied conditionally on $\mathcal{Y}_T$
\begin{align}
 & \sup_{\theta\in B(\bar{\theta})}\sigma_{T}(Y_{1:T},\theta)\sup_{\substack{f\in\mathrm{BL}(\mathbb{R})\\
\|f\|_{\mathrm{L}}\leq1}
}\left|E\left[f\left\{ \frac{1}{\surd T}\sum_{t=1}^{T}\frac{\epsilon_{T}(Y_{t},\theta)}{\sigma_{T}(Y_{1:T},\theta)}\right\} \mid\mathcal{Y}_{T}\right]-E\left[f\left\{ Z\right\} \right]\right|\nonumber \\
 & \leq 4\sup_{\theta\in B(\bar{\theta})}\sigma_{T}(Y_{1:T},\theta)\sum_{t=1}^{T}E\left[\left\{ \frac{\epsilon(Y_{t},\theta)}{\surd T\sigma_{T}(Y_{1:T},\theta)}\right\}^{2}1_{\left\{ \frac{\left|\epsilon_{T}(Y_{t},\theta)\right|}{\surd T\sigma_{T}(Y_{1:T},\theta)}>1\right\} }\mid\mathcal{Y}_{T}\right]\label{eq:stein_est_1} \\
 & +3\sup_{\theta\in B(\bar{\theta})}\sigma_{T}(Y_{1:T},\theta)\sum_{i=1}^{T}E\left[\left|\frac{\epsilon(Y_{t},\theta)}{\surd T\sigma_{T}(Y_{1:T},\theta)}\right|^{3}1_{\left\{ \left|\frac{\epsilon_{T}(Y_{t},\theta)}{\surd T\sigma_{T}(Y_{1:T},\theta)}\right|\leq1\right\} }\mid\mathcal{Y}_{T}\right].\label{eq:stein_est_2}
\end{align}
In order to control the $\sigma^2(Y_{1:T}, \theta)$ term consider the set
\begin{equation*}
	A_T(\delta) = \left\{y_{1:T} : \sup_{\theta\in B(\bar{\theta})}\left|\sigma_{T}^2(y_{1:T},\theta)-\sigma^2(\theta)\right|\leq \delta \right\}.
\end{equation*}
The uniform convergence of $\sigma_{T}^{2}(Y_{1:T},\theta)$ means that for any $\delta>0$
\begin{equation*}
\mathbb{P}^{Y}\left\{A_T(\delta)^\complement\right\} \rightarrow 0
\end{equation*}
as $T\rightarrow\infty$.
Choosing $\delta > 0$  for any family of  random variables $\gamma_{T}(Y_{1:T},\theta)$ we have %for any $\eta>0$
\begin{align*}
 & \mathbb{P}^{Y}\left(\left|\sup_{\theta\in B(\bar{\theta})}\gamma_{T}(Y_{1:T},\theta)\right|>\delta\right)\\
 & =\mathbb{P}^{Y}\left(\left\{ \left|\sup_{\theta\in B(\bar{\theta})}\gamma_{T}(Y_{1:T},\theta)\right|>\delta\right\} \cap A_T(\delta)\right)+\mathbb{P}^{Y}\left(\left\{ \left|\sup_{\theta\in B(\bar{\theta})}\gamma_{T}(Y_{1:T},\theta)\right|>\delta\right\} \cap A_T(\delta)^{\complement}\right)
\end{align*}
where we have already shown
\begin{equation}\label{eq:convergence_complement}
\mathbb{P}^{Y}\left[\left\{ \left|\sup_{\theta\in B(\bar{\theta})}\gamma_{T}(Y_{1:T},\theta)\right|>\eta\right\} \cap A_T(\delta)^\complement\right]\leq\mathbb{P}^{Y}\left\{ A_T(\delta)^\complement\right\} \rightarrow0.
\end{equation}
Hence, for showing the convergence in probability for a random variable
$\gamma_{T}(Y_{1:T},\theta)$ it suffices to ensure convergence on
the set $A(\delta)$.
On the set $A(\delta)$ we can estimate $\sigma_{T}^{2}\left\{ Y_{1:T}(\omega),\theta\right\}\geq\sigma^{2}(\theta)-\delta$ for all $\theta$. By continuity of $\sigma^2(\theta)$---and by shrinking $B(\bar{\theta})$ if necessary---we further have $\sigma^{2}(\theta)\geq\sigma^{2}(\bar{\theta})-\delta$ for all $\theta\in B(\bar{\theta})$ and we get for \eqref{eq:stein_est_1}, ignoring the constant for now
\begin{align*}
 & \sup_{\theta\in B(\bar{\theta})}\sigma_{T}(Y_{1:T},\theta)\sum_{t=1}^{T}E\left[\left\{ \frac{\epsilon(Y_{t},\theta)}{\surd T\sigma_{T}(Y_{1:T},\theta)}\right\}^{2}1_{\left\{ \frac{\left|\epsilon_{T}(Y_{t},\theta)\right|}{\surd T\sigma_{T}(Y_{1:T},\theta)}>1\right\} }\mid\mathcal{Y}_{T}\right] \mathbf{1}_{A_T(\delta)}\\
 & \leq\sup_{\theta\in B(\bar{\theta})}\frac{1}{\sigma_{T}^{1+\Delta}(Y_{1:T},\theta)}\sum_{t=1}^{T}E\left[\left\{ \frac{\epsilon(Y_{t},\theta)}{\surd T}\right\} ^{2+\Delta}1_{\left\{ \frac{\left|\epsilon_{T}(Y_{t},\theta)\right|}{\surd T\sigma_{T}(Y_{1:T},\theta)}>1\right\} }\mid\mathcal{Y}_{T}\right]\mathbf{1}_{A_T(\delta)}\\
 & \leq\sup_{\theta\in B(\bar{\theta})}\frac{1}{\left\{\sigma^2(\theta)-\delta\right\}^{(1+\Delta)/2}}\frac{1}{T^{1+\Delta/2}}\sum_{t=1}^{T}E\left\{\left|\epsilon_{T}(Y_{t},\theta)\right|^{2+\Delta}\mid\mathcal{Y}_{T}\right\}\mathbf{1}_{A_T(\delta)}\\
 & \leq\frac{C}{\left\{ \sigma^{2}(\bar{\theta})-2\delta\right\}^{(1+\Delta)/2}T^{1+\Delta/2}}\sum_{t=1}^{T}\left\{g(Y_{t})+1\right\}\overset{\mathbb{P}^Y}{\longrightarrow}0
\end{align*}
independently of $\theta$ by the Marcinkiewicz-Zygmund law of large
numbers \autocite[Theorem 4.23]{Kallenberg2006}. Together with \eqref{eq:convergence_complement} we can conclude that
\eqref{eq:stein_est_1}, vanishes in probability.

The second part,  \eqref{eq:stein_est_2}, can be controlled similarly via
\begin{align*}
 & \sigma_{T}(Y_{1:T},\theta)\sum_{i=1}^{T}E\left[\left|\frac{\epsilon(Y_{t},\theta)}{\surd T\sigma_{T}(Y_{1:T},\theta)}\right|^{3}1_{\left\{ \left|\frac{\epsilon_{T}(Y_{t},\theta)}{\surd T\sigma_{T}(Y_{1:T},\theta)}\right|\leq1\right\} }\mid\mathcal{Y}_{T}\right]\mathbf{1}_{A_T(\delta)}\\
 & \leq\frac{1}{\sigma_{T}^{1+\Delta}(Y_{1:T},\theta)}\sum_{t=1}^{T}E\left[\left|\frac{\epsilon(Y_{t},\theta)}{\surd T}\right|^{2+\Delta}1_{\left\{ \left|\frac{\epsilon_{T}(Y_{t},\theta)}{\surd T\sigma_{T}(Y_{1:T},\theta)}\right|\leq1\right\} }\mid\mathcal{Y}_{T}\right]\mathbf{1}_{A_T(\delta)} \\
 & \leq\frac{1}{\left\{\sigma^{2}(\theta)-\delta\right\}^{(1+\Delta)/2}}\sum_{t=1}^{T}E\left[\left|\frac{\epsilon(Y_{t},\theta)}{\surd T}\right|^{2+\Delta}1_{\left\{ \left|\frac{\epsilon_{T}(Y_{t},\theta)}{\surd T\sigma_{T}(Y_{1:T},\theta)}\right|\leq1\right\} }\mid\mathcal{Y}_{T}\right]\mathbf{1}_{A_T(\delta)}\\
 & \leq\frac{1}{\left\{ \sigma^{2}(\bar{\theta})-2\delta\right\} ^{(1+\Delta)/2}}\frac{1}{T^{1+\Delta/2}}\sum_{t=1}^{T}E\left\{ \left|\epsilon(Y_{t},\theta)\right|^{2+\Delta}\mid\mathcal{Y}_{T}\right\}\mathbf{1}_{A_T(\delta)} \\
 & \leq\frac{C}{ \left\{\sigma^{2}(\bar{\theta})-2\delta\right\} ^{(1+\Delta)/2}T^{1+\Delta/2}}\sum_{t=1}^{T}\left\{g(Y_{t})+1\right\}\overset{\mathbb{P}^Y}{\longrightarrow}0,
\end{align*}
which also does not depend on $\theta$. A similar argument to the one used to conclude in the case of \eqref{eq:stein_est_1} suffices also in this case.
\end{proof}
Turning to part \emph{b)}, we analyse $Z_{T}(\theta)$ under stationarity.
Therefore we need to introduce the probability measure of the auxiliary
variables under stationarity, i.e. the distribution of the auxiliary
variables conditional on the current state $\theta.$ The conditional
density is given by
\[
\pi(u\mid\theta)=\frac{\pi(u,\theta)}{\pi(\theta)}=\pi(\theta)\frac{\hat{p}(y\mid\theta,u)}{p(y\mid\theta)}m(u)/\pi(\theta)=\frac{\hat{p}(y\mid\theta,u)}{p(y\mid\theta)}m(u)
\]
which gives us the Radon-Nikodym derivative
\[
\frac{d\pi(\cdot\mid\theta)}{dm}=\prod_{t=1}^{T}\frac{\hat{p}(y_{t}\mid\theta,u_{t})}{p(y_{t}\mid\theta)}=\exp\left\{ Z_{T}(\theta)\right\}
\]
or alternatively
\begin{align*}
\prod_{t=1}^{T}\frac{\hat{p}(y_{t}\mid\theta,u_{t})}{p(y_{t}\mid\theta)} & =\prod_{t=1}^{T}\left\{ \frac{\hat{p}(y_{t}\mid\theta,u_{t})-p(y_{t}\mid\theta)}{p(y_{t}\mid\theta)}+1\right\} \\
 & =\prod_{t=1}^{T}\left\{ \frac{\epsilon_{T}(y_{t},\theta)}{\surd T}+1\right\} .
\end{align*}
The limiting distribution will now be Gaussian with a shifted mean, i.e. $\varphi(\cdot; \sigma^2(\theta)/2, \sigma^2(\theta))$. For $Z\sim \mathcal{N}(0, 1)$ we will make use of the following identity
\begin{equation*}
	E\left\{f\left(Z\sigma + \frac{\sigma^2}{2}\right)\right\} = E\left\{f\left(Z\sigma - \frac{\sigma^2}{2}\right)\exp\left(Z\sigma - \frac{\sigma^2}{2}\right)\right\}
\end{equation*}
for every bounded Lipschitz function $f$. The identity is not restricted to this case, but we will only consider bounded Lipschitz functions. Before we present the proof, we have the following useful result.

\setcounter{proposition}{6}
\begin{proposition}
\label{prop:radon_nikodym_bounded}The Radon-Nikodym derivative is
asymptotically uniformly bounded in its second moment,
\[
\limsup_{T\rightarrow\infty}\sup_{\theta\in B(\bar{\theta})}E\left[\prod_{t=1}^{T}\left\{\frac{\epsilon_{T}(Y_{t},\theta)}{\surd T}+1\right\}^{2}\mid\mathcal{Y}_{T}\right]<\infty.
\]
\end{proposition}
\begin{proof}
Using independence of $\left(U_{t,1:T}\right)_{t\geq1}$ we compute
for all $\theta\in B(\bar{\theta})$
\begin{align*}
\prod_{t=1}^{T}E\left\{ \frac{\epsilon_{T}(Y_{t},\theta)^{2}}{T}+2\frac{\epsilon_{T}(Y_{t},\theta)}{\surd T}+1\mid\mathcal{Y}_{T}\right\}  & =\prod_{t=1}^{T}\left\{\frac{\sigma^{2}(Y_{t},\theta)}{T}+1\right\}\\
 & \leq\exp\left\{\sum_{t=1}^{T}\frac{\sigma^{2}(Y_{t},\theta)}{T}\right\}\\
 & \leq\exp\left\{\sum_{t=1}^{T}\frac{C(1+g(Y_t))^{2/(2+\Delta)}}{T}\right\}\\
 & \rightarrow\exp\left[C E\left\{(1+g\left(Y_{1})\right)^{2/(2+\Delta)}\right\} \right]
\end{align*}
in $\mathbb{P}^Y$-probability, which is clearly finite by Assumption \ref{ass4}.
\end{proof}

In the following we denote $E$ the expectation under $m$ and $\tilde{E}$
the expectation under $\pi(\cdot\mid\theta)$. Using the Radon-Nikodym
derivative, it is possible to relate the expectation of $\epsilon_T(y,\theta)^{k}$
under $U$ at stationarity (conditional on $\theta$) to the expectation
under $U\sim m(\cdot)$ by
\begin{align*}
E_{U\sim\pi(\cdot\mid\theta)}\left\{ \epsilon_T(y,\theta)^{k}\right\}  & =\frac{1}{\surd T}E_{U\sim m(\cdot)}\left\{ \epsilon_T(y,\theta)^{k+1}\right\} +E_{U\sim m(\cdot)}\left\{ \epsilon_T(y,\theta)^{k}\right\};
\end{align*}
see \cite[Lemma 4]{deligiannidis2015} for a proof. We are now able to prove the second part of Theorem \ref{thm:moment_uclt}.
\begin{proof}[Proof of Theorem \ref{thm:moment_uclt}, part \emph{b)}]
 Again we take $Z\sim\mathcal{N}(0,1)$ and use the same decomposition as before, but with all expectations replaced by $\tilde{E}$, the expectation at stationarity:
\begin{align}
 & \sup_{\theta\in B(\bar{\theta})}d_\mathrm{BL}\left[\bar{g}_{T}^{\omega}(\cdot\mid\theta),\varphi\left\{ \cdot;\sigma^{2}(\theta)/2,\sigma^{2}(\theta)\right\} \right]\nonumber \\
 & =\sup_{\theta\in B(\bar{\theta})}\sup_{\substack{f\in\mathrm{BL}(\mathbb{R})\\
\|f\|_{\mathrm{BL}}\leq1}} \Bigg|\tilde{E}\left[f\left\{ S_{T}(\theta)\sigma_{T}(Y_{1:T},\theta)-\frac{1}{2T}\sum_{t=1}^{T}\epsilon_{T}(Y_{t},\theta)^{2}+\sum_{t=1}^{T}R_{T}(Y_{t},\theta)\right\} \mid\mathcal{Y}_{T}\right]\nonumber \\
 & \qquad-\tilde{E}\left[f\left\{ Z\sigma(\theta)+\frac{\sigma^{2}(\theta)}{2}\right\} \right]\Bigg|\nonumber \\
 & \leq\sup_{\theta\in B(\bar{\theta})}\sup_{\substack{f\in\mathrm{BL}(\mathbb{R})\\
\|f\|_{\mathrm{BL}}\leq1}} \Bigg|\tilde{E}\left[f\left\{ S_{T}(\theta)\sigma_{T}(Y_{1:T},\theta)-\frac{1}{2T}\sum_{t=1}^{T}\epsilon_{T}(Y_{t},\theta)^{2}+\sum_{t=1}^{T}R_{T}(Y_{t},\theta)-\frac{\sigma^{2}(\theta)}{2}+\frac{\sigma^{2}(\theta)}{2}\right\} \mid\mathcal{Y}_{T}\right]\label{eq:clt_in1-1}\\
 & \text{\ensuremath{\qquad-\tilde{E}\left[f\left\{ S_{T}(\theta)\sigma_{T}(Y_{1:T},\theta)+\sum_{t=1}^{T}R_{T}(Y_{t},\theta)-\frac{\sigma^{2}(\theta)}{2}\right\} \mid\mathcal{Y}_{T}\right]}}\Bigg|\nonumber \\
 & \quad+\sup_{\theta\in B(\bar{\theta})}\sup_{\substack{f\in\mathrm{BL}(\mathbb{R})\\
\|f\|_{\mathrm{BL}}\leq1}} \Bigg|\tilde{E}\left[f\left\{ S_{T}(\theta)\sigma_{T}(Y_{1:T},\theta)+\sum_{t=1}^{T}R_{T}(Y_{t},\theta)-\frac{\sigma^{2}(\theta)}{2}\right\} \mid\mathcal{Y}_{T}\right]\label{eq:clt_in2-1}\\
 & \qquad-\tilde{E}\left[f\left\{ S_{T}(\theta)\sigma_{T}(Y_{1:T},\theta)-\frac{\sigma^{2}(\theta)}{2}\right\} \mid\mathcal{Y}_{T}\right]\Bigg|\nonumber \\
 & \quad+\sup_{\theta\in B(\bar{\theta})}\sup_{\substack{f\in\mathrm{BL}(\mathbb{R})\\
\|f\|_{\mathrm{BL}}\leq1}} \left|\tilde{E}\left[f\left\{ S_{T}(\theta)\sigma_{T}(Y_{1:T},\theta)-\frac{\sigma^{2}(\theta)}{2}\right\} \mid\mathcal{Y}_{T}\right]-\tilde{E}\left[f\left\{ Z\sigma(\theta)+\frac{\sigma^{2}(\theta)}{2}\right\} \right]\right|.\label{eq:clt_in3-1}
\end{align}
For (\ref{eq:clt_in1-1}) we have
\begin{align*}
\eqref{eq:clt_in1-1} & \leq\sup_{\theta\in B(\bar{\theta})}\tilde{E}\left[\min\left\{ 1,\left|\frac{\sigma^{2}(\theta)}{2}-\frac{1}{2T}\sum_{t=1}^{T}\epsilon_{T}(Y_{t},\theta)^{2}\right|\right\}\mid \mathcal{Y}_T \right].
\end{align*}
An application of Cauchy-Schwartz yields
\begin{align*}
 & \tilde{E}\left[\min\left\{ 1,\left|\frac{\sigma^{2}(\theta)}{2}-\frac{1}{2T}\sum_{t=1}^{T}\epsilon_{T}(Y_{t},\theta)^{2}\right|\right\}\mid \mathcal{Y}_T \right] \\
 & =E\left[\min\left\{ 1,\left|\frac{\sigma^{2}(\theta)}{2}-\frac{1}{2T}\sum_{t=1}^{T}\epsilon_{T}(Y_{t},\theta)^{2}\right|\right\}\prod_{t=1}^{T}\left\{ \frac{\epsilon_{T}(Y_{t},\theta)}{\surd T}+1\right\} \mid \mathcal{Y}_T\right]  \\
 & \leq E\left[\min\left\{ 1,\left|\frac{\sigma^{2}(\theta)}{2}-\frac{1}{2T}\sum_{t=1}^{T}\epsilon_{T}(Y_{t},\theta)^{2}\right|^2\right\}\mid \mathcal{Y}_T\right]^{1/2}E\left[\prod_{t=1}^{T}\left\{ \frac{\epsilon_{T}(Y_{t},\theta)}{\surd T}+1\right\} ^{2}\mid \mathcal{Y}_T\right]^{1/2}\\
 & \leq E\left[\min\left\{ 1,\left|\frac{\sigma^{2}(\theta)}{2}-\frac{1}{2T}\sum_{t=1}^{T}\epsilon_{T}(Y_{t},\theta)^{2}\right|\right\}\mid \mathcal{Y}_T \right]^{1/2}E\left[\prod_{t=1}^{T}\left\{ \frac{\epsilon_{T}(Y_{t},\theta)}{\surd T}+1\right\}^{2}\mid \mathcal{Y}_T\right]^{1/2}.
\end{align*}
By Proposition \ref{prop:radon_nikodym_bounded}
\[
\limsup_{T\rightarrow\infty}\sup_{\theta\in B(\bar{\theta})}E\left[\prod_{t=1}^{T}\left\{ \frac{\epsilon_{T}(Y_{t},\theta)}{\surd T}+1\right\} ^{2}\mid \mathcal{Y}_T\right]^{1/2}<\infty
\]
and we have previously shown that
\[
\sup_{\theta\in B(\bar{\theta})}E\left[\min\left\{ 1,\left|-\frac{1}{2T}\sum_{t=1}^{T}\epsilon_{T}(Y_{t},\theta)^{2}+\frac{\sigma^{2}(\theta)}{2}\right|\right\}\mid \mathcal{Y}_T \right]\overset{}{\rightarrow}0.
\]
As for the remainder \eqref{eq:clt_in2-1} we argue analogously
\begin{align*}
 & \tilde{E}\left[\min\left\{ 1,\left|\sum_{t=1}^{T}R_{T}(Y_{t},\theta)\right|\right\} \mid\mathcal{Y}_{T}\right]\\
 & =E\left[\min\left\{ 1, \left|\sum_{t=1}^{T}R_{T}(Y_{t},\theta)\right|\right\} \prod_{t=1}^{T}\left\{\frac{\epsilon_{T}(Y_{t},\theta)}{\surd T}+1\right\}\mid\mathcal{Y}_{T}\right]\\
 & \leq E\left[\min\left\{ 1,\left|\sum_{t=1}^{T}R_{T}(Y_{t},\theta)\right|\right\} ^{2}\mid\mathcal{Y}_{T}\right]^{1/2}\cdot E\left[\prod_{t=1}^{T}\left\{ \frac{\epsilon_{T}(Y_{t},\theta)}{\surd T}+1\right\} ^{2}\mid\mathcal{Y}_{T}\right]^{1/2}\\
 & \leq E\left[\min\left\{ 1,\left|\sum_{t=1}^{T}R_{T}(Y_{t},\theta)\right|\right\} \mid\mathcal{Y}_{T}\right]^{1/2}\cdot E\left[\prod_{t=1}^{T}\left\{ \frac{\epsilon_{T}(Y_{t},\theta)}{\surd T}+1\right\} ^{2}\mid\mathcal{Y}_{T}\right]^{1/2}.
\end{align*}
The first factor vanishes in probability as we have shown in the proof of Theorem~\ref{thm:moment_uclt}(a), where as the second factor is bounded by Proposition \ref{prop:radon_nikodym_bounded}.

 For \eqref{eq:clt_in3-1}, note first that
\begin{equation*}
	E\left[\prod_{t=1}^T\left\{\frac{\epsilon_{T}(Y_{t},\theta)}{\surd T}+1\right\} \mid \mathcal{Y}_T\right] = 1 \quad \text{and} \quad E\left[e^{Z\sigma(\theta)-\frac{\sigma^{2}(\theta)}{2}}\right] = 1.
\end{equation*}
Hence, we can write
\begin{align*}
	& \tilde{E}\left[f\left\{ S_{T}(\theta)\sigma_{T}(Y_{1:T},\theta)-\frac{\sigma^{2}(\theta)}{2}\right\} \mid\mathcal{Y}_{T}\right] \\
	& = \tilde{E}\left[f\left\{ S_{T}(\theta)\sigma_{T}(Y_{1:T},\theta)-\frac{\sigma^{2}(\theta)}{2}\right\}\prod_{t=1}^T\left\{\frac{\epsilon_{T}(Y_{t},\theta)}{\surd T}+1\right\}  \mid\mathcal{Y}_{T}\right]E\left[e^{Z\sigma(\theta)-\frac{\sigma^{2}(\theta)}{2}}\right] \\
	& = \tilde{E}\left[f\left\{ S_{T}(\theta)\sigma_{T}(Y_{1:T},\theta)-\frac{\sigma^{2}(\theta)}{2}\right\}\prod_{t=1}^T\left\{\frac{\epsilon_{T}(Y_{t},\theta)}{\surd T}+1\right\} e^{Z\sigma(\theta)-\frac{\sigma^{2}(\theta)}{2}} \mid\mathcal{Y}_{T}\right]
\end{align*}
and similarly
\begin{align*}
	\tilde{E}\left[f\left\{ Z\sigma(\theta)+\frac{\sigma^{2}(\theta)}{2}\right\} \right] & = E\left[f\left\{ Z\sigma(\theta)+\frac{\sigma^{2}(\theta)}{2}\right\}e^{Z\sigma(\theta)-\frac{\sigma^{2}(\theta)}{2}} \right] E\left[\prod_{t=1}^T\left\{\frac{\epsilon_{T}(Y_{t},\theta)}{\surd T}+1\right\} \mid \mathcal{Y}_T \right]\\
	& = E\left[f\left\{ Z\sigma(\theta)+\frac{\sigma^{2}(\theta)}{2}\right\}e^{Z\sigma(\theta)-\frac{\sigma^{2}(\theta)}{2}} \prod_{t=1}^T\left\{\frac{\epsilon_{T}(Y_{t},\theta)}{\surd T}+1\right\} \mid \mathcal{Y}_T\right],
\end{align*}
where we used that $Z$ is independent of all other random variables in both cases. Using these identities we obtain
\begin{align*}
	& \left|\tilde{E}\left[f\left\{ S_{T}(\theta)\sigma_{T}(Y_{1:T},\theta)-\frac{\sigma^{2}(\theta)}{2}\right\} \mid\mathcal{Y}_{T}\right]-\tilde{E}\left[f\left\{ Z\sigma(\theta)+\frac{\sigma^{2}(\theta)}{2}\right\} \right]\right| \\
	& \leq \left|\tilde{E}\left[f\left\{ S_{T}(\theta)\sigma_{T}(Y_{1:T},\theta)-\frac{\sigma^{2}(\theta)}{2}\right\} - f\left\{ Z\sigma(\theta)+\frac{\sigma^{2}(\theta)}{2}\right\}\mid \mathcal{Y}_T \right]\right| \\
	& \leq \left|E\left[
	\left(f\left\{ S_{T}(\theta)\sigma_{T}(Y_{1:T},\theta)-\frac{\sigma^{2}(\theta)}{2}\right\} - f\left\{ Z\sigma(\theta)-\frac{\sigma^{2}(\theta)}{2}\right\}\right) \prod_{t=1}^T\left\{\frac{\epsilon_{T}(Y_{t},\theta)}{\surd T}+1\right\} e^{Z\sigma(\theta)-\frac{\sigma^{2}(\theta)}{2}}\mid \mathcal{Y}_T \right]\right| \\
	& \leq \left|E\left(\left[f\left\{ S_{T}(\theta)\sigma_{T}(Y_{1:T},\theta)-\frac{\sigma^{2}(\theta)}{2}\right\} - f\left\{ Z\sigma(\theta)-\frac{\sigma^{2}(\theta)}{2}\right\}\right]^2\mid \mathcal{Y}_T \right)^{\frac{1}{2}}\right| \\
	& \quad\times  \left|E\left[\prod_{t=1}^T\left\{\frac{\epsilon_{T}(Y_{t},\theta)}{\surd T}+1\right\}^2 e^{2\left\{Z\sigma(\theta)-\frac{\sigma^{2}(\theta)}{2}\right\}} \mid \mathcal{Y}_T \right]^{\frac{1}{2}}\right|.
\end{align*}
We investigate the two factors of the product separately. First we use the fact that $\|f\|_\infty \leq 1$ when $\|f\|_{\mathrm{BL}}\leq1$ (see \eqref{def:BL}) and thus
\begin{align*}
	& \left|\sup_{\theta\in B(\bar{\theta})}\sup_{\substack{f\in\mathrm{BL}(\mathbb{R})\\
\|f\|_{\mathrm{BL}}\leq1}} E\left(\left[f\left\{ S_{T}(\theta)\sigma_{T}(Y_{1:T},\theta)-\frac{\sigma^{2}(\theta)}{2}\right\} - f\left\{ Z\sigma(\theta)-\frac{\sigma^{2}(\theta)}{2}\right\}\right]^2\mid \mathcal{Y}_T \right)^{\frac{1}{2}}\right| \\
	& \leq \sup_{\theta\in B(\bar{\theta})}\sup_{\substack{f\in\mathrm{BL}(\mathbb{R})\\
\|f\|_{\mathrm{BL}}\leq1}}\left|E\left(\left[f\left\{ S_{T}(\theta)\sigma_{T}(Y_{1:T},\theta)-\frac{\sigma^{2}(\theta)}{2}\right\} - f\left\{ Z\sigma(\theta)-\frac{\sigma^{2}(\theta)}{2}\right\}\right]\mid \mathcal{Y}_T \right)^{\frac{1}{2}}\right| \rightarrow 0
\end{align*}	
in $\mathbb{P}^Y$-probability as established in the previous part.
For the second factor note that $Z$ is independent of all other random variables and hence
\begin{align*}
	& \sup_{\theta\in B(\bar{\theta})}\left|E\left[\prod_{t=1}^T\left\{\frac{\epsilon_{T}(Y_{t},\theta)}{\surd T}+1\right\}^2 e^{2\left\{Z\sigma(\theta)-\frac{\sigma^{2}(\theta)}{2}\right\}} \mid \mathcal{Y}_T \right]^{\frac{1}{2}}\right| \\
	& = \sup_{\theta\in B(\bar{\theta})}\left|E\left[\prod_{t=1}^T\left\{\frac{\epsilon_{T}(Y_{t},\theta)}{\surd T}+1\right\}^2 \mid \mathcal{Y}_T \right]^{\frac{1}{2}} E\left[ e^{2\left\{Z\sigma(\theta)-\frac{\sigma^{2}(\theta)}{2}\right\}}\right]^{\frac{1}{2}}\right|.
\end{align*}
We know
\begin{equation*}
\sup_{\theta\in B(\bar{\theta})}E\left[\prod_{t=1}^T\left\{\frac{\epsilon_{T}(Y_{t},\theta)}{\surd T}+1\right\}^2 \mid \mathcal{Y}_T \right]^{\frac{1}{2}}
\end{equation*}
converges to a constant in $\mathbb{P}^Y$-probability and
\begin{equation*}
	\sup_{\theta\in B(\bar{\theta})}E\left[ e^{2\left\{Z\sigma(\theta)-\frac{\sigma^{2}(\theta)}{2}\right\}}\right]^{\frac{1}{2}} = \sup_{\theta\in B(\bar{\theta})} \exp\left\{\sigma(\theta)^{2}\right\}^{1/2} < \infty.
\end{equation*}
\end{proof}

\section{Generalized Linear Mixed Models}
\label{sec:glmm_appendix}
\subsection{Exponential Families and Random Effects}
In this section we introduce a class of random effects models for which all assumptions required for Theorem \ref{theorem} are satisfied.
We analyse the latent variable model introduced in Section \ref{sec:clt} for the popular class of generalized linear mixed models \autocite[see e.g][]{mcculloch2005generalized}, where the observation density is of the form of an exponential family. We restrict attention here to the class of natural exponential family distributions, i.e. $T(y) = y$, with respect to the Lebesgue measure
\begin{equation}\label{eq:exp_fam}
  p(y\mid \eta) = m(y)\exp\left\{\eta ^\T y - A(\eta) \right\},
\end{equation}
where $y$ is the natural sufficient statistic and $\eta$ denotes the natural parameter, which will be set equal to the linear predictor in a generalized linear model. The function $m(y)$ is a base measure, which can be absorbed into the dominating measure. $A(\eta)$ is commonly referred to as the $\log$-partition function and we assume that $A$ is strictly convex and increasing in $\eta$ so that the $\log$-likelihood will be strictly concave. This assumption will be satisfied in the most common natural exponential family models including Poisson and Binomial models.
In the following we will allow for multiple measurements for each group, which means we have one random effect associated with multiple observations. This corresponds to the logistic mixed model of Section \ref{sec:sim_rem}. For the conditional exponential family with $J$ repeated measurements $y_t=(y_{t, 1},\ldots,y_{t, J})^{\T}$ where $\eta_{t, j}= c_{t, j}^{\T}\beta+X_t, j = 1, \ldots, J, t = 1, \ldots, T$ and the random effects are centred Gaussian variables $X \sim \mathcal{N}(0,\tau^{2})$ independent for each set of repeated measurements $y$. We will simplify the notation by dropping the subscript $t$ as the importance sampler for each $t$ can be considered in isolation.
Assume here that
\begin{equation}\label{eq:exp_model}
	g(y \mid x, \theta) = \prod_{j=1}^J m(y_j) \exp\left[\eta_j(x) y_j - A\{\eta_j(x)\}\right], \quad f(x \mid \theta) = \varphi(x; 0, \tau^2),
\end{equation}
where $\eta_j(x) = c_j^{\T}\beta + x$ and $c$ is a vector of covariates with corresponding parameter vector $\beta$.
The (full) model likelihood for every observation is now given by
\begin{equation*}
	p(y,x\mid\theta) \propto \prod_{j=1}^J m(y_j) \exp\left[\eta_j(x) y_j - A\{\eta_j(x)\}\right] \varphi(x,0,\text{\ensuremath{\tau^2)}}.
\end{equation*}
Since $X$ is unobserved, we are interested in the marginal likelihood
\begin{align*}
	p(y\mid\theta)
	&=\int p(y,x\mid\theta)\mathrm{d}x \\
	&=\int \prod_{j=1}^J m(y_j) \exp\left[\eta_j(x) y_j - A\{\eta_j(x)\}\right] \varphi(x,0,\tau^2)\mathrm{d}x.
\end{align*}
Consequently, the likelihood of a set of observations $y_{1:T}$, with $y_i = (y_{i, 1}, \ldots, y_{i, J})$ is
\begin{equation*}
	p(y_{1:T} \mid \theta) = \prod_{t=1}^T \int \prod_{j=1}^J m(y_{t, j}) \exp\left[\eta_{t, j}(x_t) y_{t, j} - A\{\eta_{t, j}(x_t)\}\right] \varphi(x_t,0,\tau^2)\mathrm{d}x_t.
\end{equation*}
We list the $\log$-partition function as well as it's first derivative $A'(x) = \partial_x A(x)$ (which will be important later) below together with the base measure.

\noindent\textbf{Binomial.} Denote $n$ the number of trials, then
\begin{equation*}
	A(\eta) = n\log\left(1 + e^{\eta} \right), \quad A'(\eta) = \frac{ne^\eta}{1+e^\eta}, \quad m(y) = {n \choose{y}}.
\end{equation*}
\textbf{Poisson.} For the Poisson family
\begin{equation*}
	A(\eta) = e^\eta, \quad A'(\eta) = e^\eta, \quad m(y) = \frac{1}{y!}.
\end{equation*}

\subsection{Asymptotic Posterior Normality}
This section establishes the Bernstein-von Mises theorem
for priors having exponentially decaying tails. Denote $\Theta\subset\mathbb{R}^{d}$
a subset of the Euclidean space, where we take $d=1$ without loss of generality. Consider the case of i.i.d.\ observations
$Y_{1},Y_{2},\ldots$ drawn from a density $Y_{i}\sim f(\cdot\mid\bar{\theta})$,
where $\bar{\theta}\in\Theta$ is assumed to be the ``true parameter''.
The measure describing the distribution of the data vector $Y_{1:T}=(Y_{1},\ldots,Y_{T})$
is written as $P_{T,\bar{\theta}}.$ Writing $\pi(\theta)$ for the
prior distribution we denote the posterior density  as
\begin{align*}
\pi_{T}(\theta) & =\pi(\theta\mid Y_{1:T})
 =\frac{\prod_{i=1}^{T}f(y_{i}\mid\theta)\pi(\theta)}{\int_{\Theta}\prod_{i=1}^{T}f(y_{i}\mid\theta)\pi(\theta)\mathrm{d}\theta}.
\end{align*}

\setcounter{theorem}{5}
\begin{theorem}
Let the experiment be differentiable in quadratic mean at $\bar{\theta}$
with non-singular Fisher information matrix $I_{\bar{\theta}},$ and
suppose that for every $\varepsilon>0$ there exist an increasing
sequence of sets $K_{1}\subset K_{2}\subset\ldots$ with $\cup_{i=1}^{\infty}K_{i}=\Theta$ with $K_T$ growing at rate $T$.
Assume there exists a sequence of tests such that
\[
E(\phi_{T})\rightarrow0,\quad\sup_{\left\{ \|\theta-\bar{\theta}\|\geq\varepsilon\right\} \cap K_{T}}E_{\theta}^{n}\left(1-\phi_{T}\right)\rightarrow0.
\]
Furthermore, let the prior measure be absolutely continuous in a neighbourhood
of $\bar{\theta}$ with a continuous positive density at $\bar{\theta}$
s.t. for $T$ large enough, we have
\[
\pi\left(\left[-T,T\right]^{\complement}\right)\leq c_{1}\exp\left(-c_{2}T\right),
\]
where $c_1$ and $c_2$ are positive constants.
Then the corresponding posterior distributions satisfy
\begin{equation}
\int\left|\tilde{\pi}_{T}(h)-\varphi\left(h,\surd T\left(\hat{\theta}_{T} - \theta_0\right),I_{\bar{\theta}}^{-1}\right)\right|\mathrm{d} h\rightarrow 0\label{eq:BVW1}
\end{equation}
in $P_{T, \bar{\theta}}$-probability where
\[
\Delta_{T}(\bar{\theta})=\frac{1}{\sqrt{T}}\sum_{i=1}^{T}\tilde{I}_{\bar{\theta}}^{-1}\frac{\partial\ell(\bar{\theta},Y_{i})}{\partial\theta}
\]
and
\[
\tilde{\pi}_{T}(h)=\frac{\pi_{T}(\bar{\theta}+h/\sqrt{T})}{T^{1/2}}
\]
is a measure on $H=\left\{ h=\surd{T}\left(\theta-\bar{\theta}\right):\theta\in\Theta\right\} .$
\end{theorem}

\begin{proof}
The proof follows \textcite{VanderVaart2000}, Theorem 10.1, see also the lecture notes by \textcite{nickl2012statistical}.
We will show that it is enough to show convergence of the measures restricted on some arbitrarily large compact
set. In order to do so, denote
\[
P^{C}(A)=\frac{P(A\cap C)}{P(C)}
\]
for any measurable set $A$ the restriction of the probability measure
$P$ to the set $C$.
Denote $h=\surd{T}\left(\theta-\bar{\theta}\right)$. We will write $P_{T,h}$ for the posterior
distribution with data $Y_{1:T}$ and parameter \mbox{$\bar{\theta}+h/\surd{T}(=\theta)$}.
Define the prior-weighted mixture measure over a set $C$
as
\[
P_{T,C}=\int P_{T,h}\tilde{\pi}_{T}^{C}(h)\mathrm{d}h.
\]
The expectation with respect to $P_{T,C}$ is calculated as
\begin{align*}
E_{P_{T,C}}\left\{f(Y_{1:T})\right\} & =\iint f(y_{1:T})\mathrm{d}P_{T,h}(y_{1:T})\tilde{\pi}_{T}^{C}(h)\mathrm{d}h.
\end{align*}
For any sequence of sets $A_{T}$ with $P_{T,\bar{\theta}}(A_{T})\rightarrow 0$
it follows that $P_{T,B}(A_{T})\rightarrow 0$ and vice versa, where
$B$ denotes a closed ball around 0. (Two measures with this relationship
are called mutually contiguous.)

This means that we can interchange convergence in probability under
the measures $P_{T,B}$ and $P_{T,0}$. Let $C$ now denote a ball
of size $M_{T}$ around 0 where $M_{T}\rightarrow\infty$ as $T\rightarrow\infty$.
We can show that the total variation between distance between the
posterior and the posterior restricted on the set $C$ vanishes by
estimating
\[
\left\Vert \tilde{\pi}_{T}\left(B\right)-\tilde{\pi}_{T}^{C}\left(B\right)\right\Vert _{\mathrm{tv}}\leq 2\tilde{\pi}_{T}\left(C^{\complement}\right),
\]
where $\|\cdot\|_\mathrm{tv}$ denotes the total variation norm.
We will show that the left-hand side converges to zero under $P_{T,B}$
for $B$ a closed ball around 0. We can now use the tests $\phi_{T}$ to bound
\begin{align*}
E_{T,B}\left\{\tilde{\pi}_{T}(C^{\complement})\right\} & =E_{T,B}\left\{\tilde{\pi}_{T}\left(C^{\complement}\right)\left(1-\phi_{T}+\phi_{T}\right)\right\}\\
 & \leq E_{T,B}\left[\tilde{\pi}_{T}\left(C^{\complement}\right)\left(1-\phi_{T}\right)\right]+E_{T,B}\left(\phi_{T}\right),
\end{align*}
where $E_{T,B}\left(\phi_{T}\right)=o_{P_{T,B}}(1)$ by assumption. Now
\begin{align*}
 & E_{T,B}\left\{P_{H_{T}\mid Y_{1:T}}\left(C^{\complement}\right)\left(1-\phi_{T}\right)\right\}\\
 & =\int_{B}\int_{\mathbb{R}^{T}}\int_{C^{\complement}}\left(1-\phi_{T}\right)\frac{\prod_{i=1}^{T}f(\bar{\theta}+g/\surd{T},y_{i})}{\int\prod_{i=1}^{T}f(\bar{\theta}+m/\surd{T},y_{i})\mathrm{d}\tilde{\pi}(m)}\prod_{i=1}^{T}f\left(\bar{\theta}+\frac{h}{\surd{T}},y_{i}\right)\mathrm{d}y_{i}\frac{\mathrm{d}\tilde{\pi}(h)}{\tilde{\pi}(B)}\\
 & =\frac{\tilde{\pi}(C^{\complement})}{\tilde{\pi}(B)}\int_{C^{\complement}}\int_{\mathbb{R}^{T}}\int_{B}\left(1-\phi_{T}\right)\frac{\prod_{i=1}^{T}f(\bar{\theta}+\frac{h}{\surd{T}},y_{i})}{\int\prod_{i=1}^{T}f(\bar{\theta}+m/\surd{T},y_{i})\mathrm{d}\tilde{\pi}(m)}\mathrm{d}\tilde{\pi}(h)\mathrm{d}P_{g}^{T}(y)\mathrm{d}\tilde{\pi}^{C^{\complement}}(g)\\
 & =\frac{\tilde{\pi}(C^{\complement})}{\tilde{\pi}(B)}E_{T,C^{\complement}}\left\{\tilde{\pi}_{T}(B)(1-\phi_{T})\right\}.
\end{align*}
The upper bound is
\begin{align*}
\frac{\tilde{\pi}(C^{\complement})}{\tilde{\pi}(B)}E_{C^{\complement}}^{T}\tilde{\pi}_{T}(B)(1-\phi_{T}) & =\frac{\tilde{\pi}(C^{\complement})}{\tilde{\pi}(B)}\int_{C^{\complement}}\tilde{\pi}_{T}(B)(1-\phi_{T})P_{h}^{T}\frac{\mathrm{d}\tilde{\pi}(h\cap C^{\complement})}{\tilde{\pi}(C^{\complement})}\\
 & =\frac{1}{\tilde{\pi}(B)}\int_{C^{\complement}}\int_{\mathbb{R}^d}\tilde{\pi}_{T}(B)(1-\phi_{T})dP_{h}^{T}(y)d\tilde{\pi}(h\cap C^{\complement})\\
 & \leq\frac{1}{\tilde{\pi}(B)}\int_{C^{\complement}}E(1-\phi_{T})\mathrm{d}\tilde{\pi}(h\cap C^{\complement})\\
 & =\frac{1}{\tilde{\pi}(B)}\int_{C^{\complement}}E(1-\phi_{T})\mathrm{d}\tilde{\pi}(h)\\
 & =\frac{1}{\tilde{\pi}(B)}\int_{C^{\complement}\cap \widetilde{K}_{T}}E(1-\phi_{T})\mathrm{d}\tilde{\pi}(h)+\frac{1}{\tilde{\pi}(B)}\int_{C^{\complement}\cap \tilde{K}_T^{\complement}}E(1-\phi_{T})\mathrm{d}\tilde{\pi}(h),
\end{align*}
where $\tilde{K}_T = \{h = \surd{T}(\theta - \bar{\theta}): \theta \in K_T\}$. For simplicity and without loss of generality we assume $K_T = [-T,T]$ in the following.
By \textcite[][Lemma 10.3]{VanderVaart2000} the tests converge exponentially fast so with $\theta=\bar{\theta}+h/\surd{T},h=\surd{T}\left(\theta-\bar{\theta}\right),\mathrm{d}\theta=\mathrm{d}h/\surd{T}$
\begin{align*}
\frac{1}{\tilde{\pi}(B)}\int_{C^{\complement}\cap \tilde{K}_{T}}E(1-\phi_{T})\mathrm{d}\tilde{\pi}(h) & =\frac{1}{\tilde{\pi}(B)}\int_{\left\{ \|\theta-\bar{\theta}\|\geq M_{T}/\surd{T}\right\} \cap K_{T}}E_{\theta}\left(1-\phi_{T}\right)\pi(\theta)\mathrm{d}\theta\\
 & =\frac{1}{\tilde{\pi}(U)}\int_{\left\{ \|\theta-\bar{\theta}\|\geq M_{T}/\surd{T}\right\} \cap K_{T}}E_{\theta}\left(1-\phi_{T}\right)\pi(\theta)\mathrm{d}\theta\\
 & =\frac{1}{\tilde{\pi}(B)}\int_{\left\{ D'\geq\|\theta-\bar{\theta}\|\geq M_{T}/\surd{T}\right\} \cap K_{T}}E_{\theta}\left(1-\phi_{T}\right)\pi(\theta)\mathrm{d}\theta\\
 & \quad+\frac{1}{\tilde{\pi}(B)}\int_{\left\{ \|\theta-\bar{\theta}\|\geq D'\right\} \cap K_{T}}E_{\theta}\left(1-\phi_{T}\right)\pi(\theta)\mathrm{d}\theta\\
 & =c_2\int_{\left\{ D'\geq\|\theta-\bar{\theta}\|\geq M_{T}/\surd{T}\right\} \cap K_{T}}\exp\left(-DT\left\Vert \theta-\bar{\theta}\right\Vert ^{2}\right)\mathrm{d}\theta\\
 & \quad+\frac{1}{\tilde{\pi}(B)}\int_{\left\{ \|\theta-\bar{\theta}\|\geq D'\right\} \cap K_{T}}\exp(-c_3T)\pi(\theta)\mathrm{d}\theta\\
 & \leq c_2\int_{\left\{ h:h\geq M_{T}\right\} \cap K_{T}}\exp\left(-DT\left\Vert \theta-\bar{\theta}\right\Vert ^{2}\right)T^{1/2}\mathrm{d}\theta\\
 & \quad+2c_3T^{1/2}\exp(-c_4 T),
\end{align*}
where we used $\tilde{\pi}(B) \geq 1/(c_3T^{1/2})$ for some constant $c_3$ because the prior is positive and continuous at $\bar{\theta}$.
For the second part
\begin{align*}
\frac{1}{\tilde{\pi}(B)}\int_{C^{\complement}\cap \tilde{K}_{T}^{c}}E_{h}(1-\phi_{T})d\tilde{\pi}(h) & \leq\frac{1}{\tilde{\pi}(B)}\int_{C^{\complement}\cap \tilde{K}_{T}^{\complement}}\mathrm{d}\tilde{\pi}(h)\\
 & =\frac{1}{\tilde{\pi}(B)}\int_{C^{\complement}\cap \tilde{K}_{T}^{\complement}}\pi(\bar{\theta}+h/\surd{T})\mathrm{d}h\\
 & \leq c_3T^{1/2}\int_{K_{T}^{\complement}}\pi(\theta)\mathrm{d}\theta\\
 & \leq c_3T^{1/2}\cdot c_1 \exp\left(-c_2T\right).
\end{align*}
As $T\rightarrow\infty$ we have
\[
\|\tilde{\pi}_{T}-\tilde{\pi}_{T}^{C}\|_{\mathrm{tv}}\rightarrow 0
\]
in $P_{T,B}$-probability and by contiguity also in $P_{T,\bar{\theta}}$.

Similarly, for a Gaussian distribution with means $\sup\left|\mu_{T}\right|<\infty$
and variance $\sigma^{2}$ we have
\[
\left\Vert \mathcal{N}\left(\mu_{T},\sigma^{2}\right)-\mathcal{N}^{C}\left(\mu_{T},\sigma^{2}\right)\right\Vert \leq 2\mathcal{N}\left(\mu_{T},\sigma^{2}\right)\left(C^{\complement}\right).
\]
We know that $\Delta_{T,\bar{\theta}}$ is uniformly tight, i.e. for
any $\varepsilon>0$ there exists $K$ such that $\sup_{T}P\left(\left|\Delta_{T,\bar{\theta}}\right|\leq K\right)=1-\varepsilon.$
Hence, with probability $1-\varepsilon$
\[
\left\Vert \mathcal{N}\left(\Delta_{T,\bar{\theta}},I_{\bar{\theta}}^{-1}\right)-\mathcal{N}^{C}\left(\Delta_{T,\bar{\theta}},I_{\bar{\theta}}^{-1}\right)\right\Vert \leq 2\mathcal{N}\left(\Delta_{T,\bar{\theta}},I_{\bar{\theta}}^{-1}\right)(C^{\complement})
\]
by choosing $M$ (the radius of $C$) sufficiently large. Hence, by
the triangle inequality we have to show that
\[
\int\left|\tilde{\pi}_{T}^{C}(h)-\mathcal{\varphi}^{C}\left(h;\Delta_{T,\bar{\theta}},I_{\bar{\theta}}^{-1}\right)\right|\mathrm{d}h\rightarrow 0
\]
in $P_{T,0}$-probability. Denoting $x^+ = \max\{0, x\}$
\begin{align*}
 & \frac{1}{2} \int\left|\tilde{\pi}_{T}^{C}(h)-\mathcal{\varphi}^{C}\left(h;\Delta_{T,\bar{\theta}},I_{\bar{\theta}}^{-1}\right)\right|\mathrm{d}h\\
 & =\int\left(1-\frac{\mathcal{\varphi}^{C}\left(h;\Delta_{T,\bar{\theta}},I_{\bar{\theta}}^{-1}\right)}{\tilde{\pi}_{T}^{C}(h)}\right)^{+}\tilde{\pi}_{T}^{C}(h)\mathrm{d}h \\
 & =\int\left(1-\frac{\mathcal{\varphi}^{C}\left(h;\Delta_{T,\bar{\theta}},I_{\bar{\theta}}^{-1}\right)\int 1_{C}f_{T,g}^{C}(g)\pi(g)\mathrm{d}g}{1_{C}\tilde{\pi}_{T}^{C}(h)}\right)^{+}\tilde{\pi}_{T}^{C}(h)\mathrm{d}h\\
 & =\int\left(1-\int\frac{1_{C}(g)f_{T,g}^{C}(g)\pi(g)\mathcal{\varphi}^{C}\left(h;\Delta_{T,\bar{\theta}},I_{\bar{\theta}}^{-1}\right)}{1_{C}(h)f_{T,h}^{C}(h)\pi(h)\mathcal{\varphi}^{C}\left(g;\Delta_{T,\bar{\theta}},I_{\bar{\theta}}^{-1}\right)}\mathcal{\varphi}^{C}\left(g;\Delta_{T,\bar{\theta}},I_{\bar{\theta}}^{-1}\right)\mathrm{d}g\right)^{+}\tilde{\pi}_{T}^{C}(h)\mathrm{d}h\\
 & \leq\iint\left(1-\frac{f_{T,g}^{C}(g)\pi(g)\mathcal{\varphi}^{C}\left(h;\Delta_{T,\bar{\theta}},I_{\bar{\theta}}^{-1}\right)}{f_{T,h}^{C}(h)\pi(h)\mathcal{\varphi}^{C}\left(g;\Delta_{T,\bar{\theta}},I_{\bar{\theta}}^{-1}\right)}\right)^{+}\mathcal{\varphi}^{C}\left(g;\Delta_{T,\bar{\theta}},I_{\bar{\theta}}^{-1}\right)\mathrm{d}g\tilde{\pi}_{n}^{C}(h)\mathrm{d}h\\
 & \leq\left\{ \sup_{x\in C}\mathcal{\varphi}^{C}\left(x;\Delta_{T,\bar{\theta}},I_{\bar{\theta}}^{-1}\right)\right\} \iint\left(1-\frac{f_{T,g}^{C}(g)\pi(g)\mathcal{\varphi}^{C}\left(h;\Delta_{T,\bar{\theta}},I_{\bar{\theta}}^{-1}\right)}{f_{T,h}^{C}(h)\pi(h)\mathcal{\varphi}^{C}\left(g;\Delta_{T,\bar{\theta}},I_{\bar{\theta}}^{-1}\right)}\right)^{+}\mathrm{d}g\tilde{\pi}_{T}^{C}(h)\mathrm{d}h.
\end{align*}
By dominated convergence it is enough to conclude that this quantity
goes to 0 in
\begin{align*}
P_{T,C}(\mathrm{d}y)\tilde{\pi}_{T}^{C}(\mathrm{d}h)\lambda_{C}(\mathrm{d}g) & =\int P_{T,x}(\mathrm{d}y)\tilde{\pi}_{T}^{C}(h)\mathrm{d}h\lambda_{C}(\mathrm{d}g)\\
 & =\int\prod_{i=1}^{T}f\left(\theta+s/\surd{T},y_{i}\right)\frac{\prod_{i=1}^{T}f\left(\theta+h/\surd{T},y_{i}\right)\tilde{\pi}^{C}(h)\mathrm{d}h}{\int\prod_{i=1}^{T}f\left(\theta+u/\surd{T},y_{i}\right)\tilde{\pi}^{C}(u)\mathrm{d}u}\mathrm{d}s\lambda_{C}(\mathrm{d}g)\\
 & =\prod_{i=1}^{T}f\left(\theta+h/\surd{T},y_{i}\right)\tilde{\pi}^{C}(h)\mathrm{d}h\lambda_{C}(\mathrm{d}g)\\
 & =P_{T,C}(\mathrm{d}y)\tilde{\pi}^{C}(h)\mathrm{d}h\lambda_{C}(\mathrm{d}g)
\end{align*}
probability. Under Theorem 7.2 in \textcite{VanderVaart2000} mean-square differentiability
of the likelihood implies that the likelihood ratio allows for the
LAN \autocite[Definition 7.14]{VanderVaart2000} expansion
\begin{align*}
 & \frac{\prod_{i=1}^{T}f(\theta+g/\surd{T},y_{i})}{\prod_{i=1}^{T}f(\theta+h/\surd{T},y_{i})}=\\
 & =\prod_{i=1}^{T}\frac{f(\theta+g/\surd{T},y_{i})}{f(\theta,y_{i})}\bigg/\prod_{i=1}^{T}\frac{f(\theta+h/\surd{T},y_{i})}{f(\theta,y_{i})}\\
 & =\exp\left(\frac{1}{\surd{T}}\sum_{i=1}^{T}g^{T}\ell'_{\theta}(y_{i})-\frac{1}{2}g^{T}I_{\theta}g-\frac{1}{\surd{T}}\sum_{i=1}^{T}h^{T}\ell'_{\theta}(y_{i})-\frac{1}{2}h^{T}I_{\theta}h+o_{P_{\theta}}(1)\right)
\end{align*}
and thus as $T\rightarrow\infty$ and using continuity of the prior
$\pi$ at $\bar{\theta}$ we have
\[
1-\frac{f_{T,g}^{C}(g)\pi(g)\mathcal{\varphi}^{C}\left(h;\Delta_{T,\bar{\theta}},I_{\bar{\theta}}^{-1}\right)}{f_{T,h}^{C}(h)\pi(h)\mathcal{\varphi}^{C}\left(g;\Delta_{T,\bar{\theta}},I_{\bar{\theta}}^{-1}\right)}\rightarrow0
\]
which yields the result.
\end{proof}

\noindent\begin{remark}
\begin{itemize}
\item[\emph{i)}]
The centring sequence $\Delta_{T,\theta}$ can be replaced by any best regular estimator. To see this note that following \textcite[Theorem 8.14]{VanderVaart2000} any best regular estimator, $\hat{\theta}_T$, satisfies the expansion
\begin{equation*}
 \surd{T}(\hat{\theta}_T - \bar{\theta})	=\frac{1}{\surd{T}}\sum_{i=1}^{T}\tilde{I}_{\bar{\theta}}^{-1}\frac{\partial\ell(\bar{\theta},Y_{i})}{\partial\theta} + o_{P_{T,\bar{\theta}}}(1)
\end{equation*}
and thus
\[
\Delta_{T}(\bar{\theta})-\surd{T}\left(\hat{\theta}_{T}-\bar{\theta}\right)\rightarrow 0
\]
in $P_{T, \theta_0}$-probability as $T\rightarrow\infty$. Since
\[
\left\Vert \mathcal{N}\left(\Delta_{T,\bar{\theta}},\tilde{I}_{\bar{\theta}}^{-1}\right)-\mathcal{N}\left\{\surd{T}\left(\hat{\theta}_{T}-\bar{\theta}\right),\tilde{I}_{\bar{\theta}}^{-1}\right\}\right\Vert \lesssim \left\Vert \surd{T}\left(\hat{\theta}_{T}-\bar{\theta}\right)-\Delta_{T,\bar{\theta}}\right\Vert \rightarrow 0
\]
in probability.
\item[\emph{ii)}] Under regularity conditions \cite[Theorem 5.39]{VanderVaart2000} the maximum likelihood estimator is best regular and can be used as a centring sequence following the argument in \emph{i)}.
\end{itemize}
\end{remark}

We will now apply this Bernstein-von Mises result to our exponential family models. Hence, consider again the likelihood contribution of every observation $y$,
\begin{equation}\label{eq:marginal_iid}
	p(y \mid \beta, \tau) = \int \prod_{j=1}^J m(y_{j}) \exp\left\{(c_{j}^\T\beta + x) y_{j} - A(c_{j}^\T\beta + x)\right\} \varphi(x,0,\tau^2)\mathrm{d}x.
\end{equation}
%Since the calculations can be quite cumbersome we will only provide a sketchy outline. Thus,
For simplicity we assume that the exogenous variables $c_j$ are all identical and that $\Theta$ is a subset of $\mathbb{R}$. Let $A$ be continuously differentiable (e.g. the Binomial and Poisson models introduced above).
The prior can be easily chosen to fulfil the conditions of the updated Bernstein--von Mises theorem. The other conditions need further analysis.
In order to show differentiability in quadratic mean it is sufficient to prove that the map $\theta \mapsto p(y \mid \theta)^{1/2}$ is continuously differentiable.
By Lemma 7.6 in \textcite{VanderVaart2000} we need to show that
\[
\theta\mapsto p(y\mid\theta)^{1/2}=\left[\int m(y)\exp\left\{ \left(c^{\T}\beta+x\right)\cdot y - A(c^{\T}\beta+x)\right\} \varphi(x,0,\text{\ensuremath{\tau^2)\mathrm{d}x}}\right]^{1/2}
\]
is continuously differentiable for all $y$. Firstly,
\begin{equation*}
	\frac{\partial}{\partial\theta} p(y \mid \theta)^{1/2} = \frac{1}{2p(y\mid\theta)^{1/2}}\partial_\theta p(y\mid \theta).
\end{equation*}
It is easy to see that $\theta \mapsto p(y\mid\theta)$ and
$\theta\mapsto \partial_\theta p(y \mid\theta)$ are continuous.
The fisher information is well defined, continuous in $\theta$ and positive since
\begin{align*}
	I_\theta & = E\left[\left\{\partial_\theta\log p(Y\mid \theta)\right\}^2\right] \\
	& = \int \left\{ \partial_\theta\log p(y\mid \theta)\right\}^2 p(y \mid \theta)\mathrm{d}y > 0
\end{align*}
whenever $\partial_\theta\log p(y\mid\theta)$ is not identically 0 for all $y$.
The multivariate case is more involved and treated for example in \textcite{mukerjee2002positive} for the Binomial and Poisson case.
In order to ensure the existence of the tests consider $K_{1}\subset K_{2}\subset\ldots$ an increasing sequence of compact sets with $\cup_{i=1}^{\infty}K_{i}=\Theta$. Then, if the model is identifiable and continuous in total variation norm, Lemma 10.6 in \textcite{VanderVaart2000}, and a diagonal argument similar to that in the proof of
\cite[Lemma~10.6]{VanderVaart2000},  ensures the existence of a sequence of estimators $\hat{\theta}_{T}$ such that $\sup_{\theta\in K_{T}}P_{\theta}(|\hat{\theta}_{T}-\theta|\geq\varepsilon)\rightarrow 0$
whence we have, see for example \cite[Lemmas~1,2 in Section 2.2.3]{nickl2012statistical},
\[
E_{\bar{\theta}}(\phi_{T})\rightarrow 0,\quad\sup_{\left\{ \|\theta-\bar{\theta}\|\geq\varepsilon\right\} \cap K_{T}}E_{T,\theta}\left(1-\phi_{T}\right)\rightarrow 0.
\]
Since our model has a density with respect to the Lebesgue measure continuity in total variation is trivially the case as we can write the total variation distance as
\begin{equation*}
	\|P_\theta - P_{\theta'}\|_\mathrm{tv} = \int \left|p(y \mid \theta) - p(y \mid \theta') \right| \mathrm{d}y.
\end{equation*}
Therefore, by Scheff\'e's lemma, continuity in the parameter already implies convergence of the integral and therefore continuity in the total variation distance. To conclude that our models are indeed identifiable it is enough to ensure that
\begin{itemize}
	\item[$i)$] the integral
	$$E(Y) = E\left[A'(k + X)\right] = \int_{\mathbb{R}} A'(k + \tau x)\varphi(x; 0, 1) \mathrm{d}x< \infty,
	$$
	for all $k, \tau$ and
	\item[$ii)$] the equation
	$$
		\frac{A'(c^T\beta_1 + \tau_1 x)}{\tau_1} = \frac{A'(c^T\beta_2 + \tau_2 x)}{\tau_2} \quad \text{for all $c$ and $x$}
	$$
	has no solution,
\end{itemize}
see \textcite{labouriau2014note}. These conditions are fulfilled for the Binomial case, $A'(\eta) = ne^\eta/(1+e^\eta)$, and Poisson case $A'(\eta) = e^\eta$.

\subsection{Importance Sampling with Univariate Random Effects}
We will now consider Assumption \ref{ass3} in the context of generalized linear mixed models, which we will prove using Assumption \ref{ass4} and Theorem \ref{thm:uniform_CLT}.
In the following we will first consider a univariate random effect and a Gaussian importance sampling proposal. This will include the example of Section \ref{sec:sim_rem}. In addition we will show how fatter tails in the proposal affect the existence of moments by considering a univariate $t$-proposal. 
Recall that we are interested in bounds on
\begin{equation}\label{eq:is_weight}
E^Y\left[\sup_{\theta\in B(\overline{\theta})} E^{X \mid Y}\left\{\overline{w}(Y, X, \theta)^a \right\}\right] = E^{Y}\left[\sup_{\theta\in B(\overline{\theta})}\frac{E^{X \mid Y}\{w(Y, X, \theta)^a \}}{p(Y\mid\theta)^a}\right],
\end{equation}
where $a > 0$, $\theta = (\beta, \tau)$ and $B(\bar{\theta}) \subset \Theta$ denotes a closed $\varepsilon$-ball around $\bar{\theta}$.
For additional clarity, we write $E^Y$ and $E^{X \mid Y}$ for the expectations over $Y$ and $X$ given $Y$, respectively. Consider the Gaussian proposal centred at the mode
\begin{equation}\label{eq:gaussian_prop}
	q(x \mid y) = \varphi(x; \widehat{x}, \tau_q^2),
\end{equation}
where $\tau_q^2$ denotes the proposal variance and $\widehat{x}$ is the mode of $h(x;y) = g(y\mid x)f(x)$ and fulfils the first order condition
\begin{equation}\label{eq:first_order_condition}
	\widehat{x} = \tau^2\left\{S - \widetilde{A}'(\widehat{x})\right\},
\end{equation}
where $\widetilde{A}'(x) = \sum_{j=1}^J A'(c_j^{\T}\beta + x)$ with $A'(z) = \partial_z A(z)$ and $S = \sum_{j=1}^J y_j$. For later convenience we define the unnormalized proposal density
\[
\widetilde{q}(x;y)=\frac{q(x\mid y)}{q(\widehat{x}\mid y)},
\]
where $q(x\mid y)$ is the proposal density. For a symmetric proposal distribution centred at $\widehat{x}$ the term $q(\widehat{x}\mid y)$ is simply an inverse normalizing constant, which only involves the proposal parameters. For the Gaussian proposal
\begin{equation}\label{eq:q_gaussian}
\widetilde{q}(x; y) = \exp\left\{- \frac{(x - \widehat{x})^2}{2\tau^2}\right\}, \quad q(\widehat{x} \mid y)=\frac{1}{(2\pi\tau_q^2)^{1/2}}.
\end{equation}
Associated with this we introduce the modified weight which is
defined as
\begin{equation}
\widetilde{w}(x,y)=\frac{g(y\mid x)f(x)}{g(y\mid \widehat{x})f(\widehat{x})}%
\frac{1}{\widetilde{q}(x;y)}=\frac{h(x;y)}{h(\widehat{x};y)}\frac
{1}{\widetilde{q}(x;y)},\label{eq:mod_weight}%
\end{equation}
where $h(x;y)=g(y\mid x)f(x)$. These weights are easier to work with as
$\widetilde{w}(x,y)=1$ when $x=\widehat{x}$. It is easily seen that%
\[
\widetilde{w}(x,y)=\frac{h(x;y)}{q(x\mid y)}\frac{q(\widehat{x}%
\mid y)}{h(\widehat{x};y)}=w(x,y)\frac{q(\widehat{x}\mid y)}{h(\widehat{x};y)},
\]
so that the modified weights are proportional to the standard weights
$w(x,y)$ as a function of $x$. We can recast the expectation (\ref{eq:is_weight}) as
\begin{align}
E^Y\left[\sup_{\theta\in B(\overline{\theta})} E^{X \mid Y}\left\{\overline{w}(Y, X, \theta)^a \right\}\right]
& = E^{Y}\left[\sup_{\theta\in B(\overline{\theta})} \frac{E^{X \mid Y}\{w(X,Y, \theta)^{a}\}}{p(Y \mid \theta)^{a}}\right] \label{mod_weight_expect}\\
& = E^{Y}\left[\sup_{\theta\in B(\overline{\theta})} \frac{E^{X \mid Y}\left\{\widetilde{w}(X,Y, \theta)^{a}\right\}}{E^{X \mid Y}\left\{\widetilde{w}(X,Y, \theta)\right\}^{a}}\right].\nonumber
\end{align}
The $\log$-density of the observations is given by
\begin{align}\nonumber
\log g(y \mid x) & = \sum_{j=1}^{J}\{\log m(y_{j})+y_j\eta
_{j}-A(\eta_{j})\}\label{eq:loglik_exponentialfam}\\
&  =\sum_{j=1}^{J}\{\log m(y_{j})+y_{j}c_j^{\T}\beta+y_{j}x-A(c_j^{\T}\beta+x)\}\nonumber\\
&  =k(y)+x(J\overline{y})-\widetilde{A}(x),\nonumber
\end{align}
where $k(y)$ represents constant values (which do not depend upon
$x$), $J\overline{y}= \sum_{j=1}^{J}y_{j}$ and $\widetilde{A}(x)=\sum_{j=1}^{J}A(c_j^{\T}\beta+x)$. Hence, we get
\begin{equation}
\log h(x;y)=\log g(y\mid x)f(x)=c+x(J\overline{y})-\widetilde{A}(x)-\frac{x^{2}}{2\tau^{2}}.\label{eq:pi_exp_uni}%
\end{equation}
We will proceed by deriving bounds for the denominator and enumerator of \eqref{eq:is_weight} separately. We present the following lemma on the denominator without reference to the Gaussian proposal, because it holds for general proposal distribution.
\begin{lemma}\label{lem:inv_upper_bound}
Consider the exponential family model with repeated measurement $j = 1, \ldots, J$ and Gaussian random effects. For general proposal density $q(x \mid y)$ we have
\begin{equation*}
	\frac{1}{E^X\left\{\widetilde{w}(X,y)\right\}} \leq \frac{(2\pi)^{1/2}}{C}(b +1),
\end{equation*}
where $b = \tau \widetilde{A}'(\widehat{x})$ and $C = q(\widehat{x} \mid y) (2\pi\tau^2)^{1/2}$.
\end{lemma}
\begin{proof}[Proof of Lemma \ref{lem:inv_upper_bound}]
For given observation $y$, the expectation of the rescaled weights is
\begin{align*}
E^{X}\left\{\widetilde{w}(X,y)\right\} & = \int \widetilde{w}(x,y) q(x \mid y) \mathrm{d}x \\
& = \int \frac{h(x;y)}{h(\widehat{x};y)}\frac{q(x \mid y)}{\widetilde{q}(x;y)} \mathrm{d}x \\
& = q(\widehat{x}\mid y)\int\frac{h(x; y)}{h(\widehat{x}; y)}\mathrm{d}x.
\end{align*}
Write again $S = J\bar{y} = \sum_{j=1}^Jy_j$. Since $\widetilde{A}$ is an increasing function we obtain for $x \leq \widehat{x}$,
\begin{align*}
\log h(x; y)-\log h(\widehat{x} ; y) & =-\widetilde{A}(x)+xS-\frac{1}{2}\frac{x^{2}}{\tau^{2}} +\widetilde{A}(\widehat{x})-\widehat{x}S+\frac{1}{2}\frac{\widehat{x}^{2}}{\tau^{2}}\\
& \geq(x-\widehat{x})S-\frac{1}{2}\frac{x^{2}}{\tau^{2}}+\frac{1}{2}\frac{\widehat{x}^{2}}{\tau^{2}}\\
& =R_{2}-\frac{1}{2}\frac{\{x-\tau^{2}S\}^{2}}{\tau^{2}},
\end{align*}
where
\[
R_{2}=\frac{\tau^{2}S^{2}}{2}-\widehat{x}S+\frac{1}{2}\frac
{\widehat{x}^{2}}{\tau^{2}}=\frac{\tau^{2}\widetilde{A}^{\prime
}(\widehat{x})^{2}}{2}\text{, }%
\]
by using the first order condition for the mode $\widehat{x}=\tau^{2}\{S-\widetilde{A}^{\prime}(\widehat{x})\}$. Therefore
\[
E^{X}\left\{\widetilde{w}(X,y)\right\} \geq
q(\widehat{x}\mid y)\left(2\pi\tau^2\right)^{1/2}\exp\left\{  \frac{\tau^{2}\widetilde{A}%
^{\prime}(\widehat{x})^{2}}{2}\right\}  \Phi\{-\tau\widetilde{A}^{\prime
}(\widehat{x})\}.
\]
Consider the inequality due to \textcite{birnbaum1942}
\begin{equation*}
	\frac{\exp(-b^2/2)}{1 - \Phi(b)} < \left(\frac{\pi}{2}\right)^{1/2} \left\{b + (b^2 + 4)^{1/2} \right\}
\end{equation*}
Setting $b = \tau \widetilde{A}'(\widehat{x})$ and $C = q(\widehat{x} \mid y) (2\pi\tau^2)^{1/2}$ gives
\begin{align*}
\frac{1}{E^X\left\{\widetilde{w}(X, y)\right\}} & \leq \frac{\exp(-b^2/2)}{C\left\{1 - \Phi(b)\right\}} \\
& \leq C^{-1}\left(\frac{\pi}{2}\right)^{1/2} \left\{b + (b^2 + 4)^{1/2} \right\} \\
& \leq \frac{\left(2\pi\right)^{1/2}}{C} (b + 1),
\end{align*}
as $(b^2 + 4)^{1/2} \leq b + 2$. %concavity of $x \mapsto \surd x$ to estimate
\end{proof}
Using Lemma \ref{lem:inv_upper_bound} we can set
\begin{equation*}
	\sup_{\theta \in B(\overline{\theta})} \frac{1}{E^X\left\{\widetilde{w}(X, y, \theta)\right\}} \leq \sup_{\theta \in B(\overline{\theta})} \frac{1}{q(\widehat{x} \mid y) \tau}\left\{\tau \widetilde{A}'(\widehat{x}) + 1 \right\}.
\end{equation*}
We will use this result in the following corollary.

\begin{corollary}\label{cor:upper_exp_inv}
Assume one of the following condition holds:
\begin{itemize}
	\item[(i)] $\sup_x \widetilde{A}'(x) < \infty$,
	\item[(ii)] $E^Y(Y^a) < \infty$ and $\sup_{\theta\in B(\bar{\theta})} \widetilde{A}'(0) <\infty$.
\end{itemize}
Then taking the expectation over $Y$, we have
\begin{equation*}
	E^Y\left[\sup_{\theta\in B(\overline{\theta})} \frac{1}{E^{X\mid Y}\{\widetilde{w}(X,Y)\}^a}\right] < \infty.
\end{equation*}
\end{corollary}
\begin{proof}
Applying Lemma \ref{lem:inv_upper_bound} with $b=\tau\widetilde{A}^{\prime}(\widehat{x})$ yields
\begin{equation}
\label{eq:applemmanine}
E^Y\left[\sup_{\theta\in B(\overline{\theta})}\frac{1}{E^X\{\widetilde{w}(X,Y)\}^a}\right] \leq E^Y\left[\sup_{\theta \in B(\overline{\theta})}\left\{ \frac{\tau \widetilde{A}'(\widehat{x}) + 1}{C \tau} \right\}^a\right],
\end{equation}
where we write $C = q(\widehat{x} \mid y)$ which only involved parameters of the proposal distribution.
The right-hand side of \eqref{eq:applemmanine} is finite provided $E^{Y}\{\sup_{\theta \in B(\overline{\theta})}\widetilde{A}^{\prime}(\widehat{x})^a\}<\infty$. This concludes the proof for \emph{(i)}. For \emph{(ii)} we need to control the function $\widetilde{A}'(\widehat{x})$. Therefore, it is useful to establish the behaviour of $\widetilde{A}^{\prime}(\widehat{x})$ in terms of the random variables $y = (y_1, \ldots, y_J)$. Recall the first order condition \eqref{eq:first_order_condition}
\begin{equation*}
\widehat{x}=\tau^{2}\{J\overline{y}-\widetilde{A}^{\prime}(\widehat{x})\},
\end{equation*}
where the sufficient statistic is $S=J\overline{y}=$ $\sum_{j=1}^{J}y_{j}$. It is easily established that $\widetilde{A}^{\prime}(\widehat{x})\leq\max\{\widetilde{A}^{\prime}(0),J\overline{y}\}$.
To see this note
\begin{equation}\label{eq:foc}
{\partial_x}\log h(x; y)=J\overline{y}-\widetilde{A}^{\prime}(x)-\frac{x}{\tau^{2}}.
\end{equation}
The function $\widetilde{A}^{\prime}(x)$ is monotonically increasing. If $\widetilde{A}^{\prime}(0)<J\overline{y}$, then at $x=0$, $\partial_{x}\log h(x; y)>0$ and at $x=\widetilde{x}$, where $\widetilde{A}^{\prime}(\widetilde{x})=J\overline{y}$, $\partial_{x}\log h(x;y)<0$ since $\widetilde{x}>0$. Similarly, if $\widetilde{A}^{\prime}(0)<J\overline{y}$ then at $x=0$, $\partial_x\log h(x; y)<0$ and at $x=\widetilde{x}$, $\partial_{x}\log h(x;y)>0$. As a consequence, the mode of the concave function $\log h(x; y)$, $\widehat{x}$ is always between $0$ and $\widetilde{x}$, where
$\widetilde{A}^{\prime}(\widetilde{x})=J\overline{y}$. This yields $\widetilde{A}^{\prime}(\widehat{x})\leq\max\{\widetilde{A}^{\prime}(0),J\overline{y}\}$ so that
\begin{align*}
E^Y\left\{\sup_{\theta \in B(\overline{\theta})}\widetilde{A}^{\prime}(\widehat{x})^a\right\} & \leq  E^Y\left[\sup_{\theta \in B(\overline{\theta})}\max\{\widetilde{A}^{\prime}%
(0),S\}^{a}\right] \\
& \leq  E^Y\left[\max\left\{\sup_{\theta \in B(\overline{\theta})}\widetilde{A}^{\prime}%
(0),S\right\}^{a}\right] \\
& = \sup_{\theta \in B(\overline{\theta})}\widetilde{A}^{\prime}(0)^a\mathbb{P}^Y\left\{S<\sup_{\theta \in B(\overline{\theta})}\widetilde{A}^{\prime}(0)\right\}+\int_{\sup_{\theta \in B(\overline{\theta})}\widetilde{A}^{\prime}(0)}^{\infty}s^{a}\mathrm{d}F_S(s).
\end{align*}
The last quantity is finite whenever $\sup_{\theta \in B(\overline{\theta})} \widetilde{A}'(0)< \infty$ and $E^Y(Y^{a})<\infty$.
\end{proof}

\begin{remark}[Examples with Gaussian proposal]\label{cor:inverse_gauss}
If the proposal is a Gaussian centred at the mode \mbox{$q(x \mid y) = \varphi(x; \widehat{x}, \tau_q^2)$} and $C$ as defined in Lemma \ref{lem:inv_upper_bound}, then $C = \tau/\tau_q$. For the Binomial case, we know that $\sup_x \widetilde{A}'(x) < \infty$ and therefore condition $(i)$ of the preceding Corollary \ref{cor:upper_exp_inv} is fulfilled.
For the Poisson case $\widetilde{A}'(x)$ is not bounded, but we can use the second part of the corollary. Note that $\widetilde{A}'(x)$ is continuous and therefore $\widetilde{A}'(0)$ can be bounded in a neighbourhood small enough. In addition, if the Poisson model is true, it is straightforward to establish that the moments $E(Y^a)$ exist for all $a > 0$ and we can therefore conclude by part \emph{(ii)}.
\end{remark}
Having established conditions to ensure
\begin{equation*}
	E^Y\left[\sup_{\theta \in B(\overline{\theta})} E^{X \mid Y}\{\widetilde{w}(X,Y)\}^{-a}\right] < \infty
\end{equation*}
we can bound (\ref{mod_weight_expect}) whenever there exists a constant $K < \infty$ such that
\begin{equation*}
	\sup_{y \in \mathsf{Y}} \sup_{\theta \in B(\overline{\theta})} E^{X \mid Y}\left\{\widetilde{w}(X,y)^a\right\} < K.
\end{equation*}
In the following we will provide conditions for Gaussian and $t$-distributed proposals.

\begin{proposition}
\label{gaussian_upper} Consider the Gaussian proposal \eqref{eq:gaussian_prop} and some exponent $a > 0$. Then
\[
E^{X}\{\widetilde{w}(X,y)^{a}\}<\infty
\]
if and only if $\tau_q^2>\frac{(a-1)}{a}\tau^{2}$, where $\tau^{2}$
is the variance of the random effects term. If this condition is satisfied then
\[
E^{X}\{\widetilde{w}(X,y)^{a}\}\leq\left\{\frac{a\tau_q^2-(a-1)\tau^{2}}{\tau^{2}}\right\}^{-\frac{1}{2}},
\]
independent of $y$. 
\end{proposition}

\begin{proof}
For brevity we define the sum $S=J\overline{y}=\sum_{j=1}^{J}y_{j}$ and again have $\widetilde{A}(x)=\sum_{j=1}^{J}A(c{j}^\T\beta+x)$.
Note that $x \mapsto \widetilde{A}(x)$ is convex and thus always dominates its chord
\begin{equation*}
	\widetilde{A}(x) \geq \widetilde{A}(\widehat{x}) + \widetilde{A}'(\widehat{x})(x-\widehat{x})
\end{equation*}
for any values $x, \widehat{x}$. Then the modified proposal form $\widetilde{q}(x;y)$ is given by (\ref{eq:q_gaussian}), so
\begin{align*}
\log\widetilde{w}(x,y) &  =\log h(x; y)-\log h(\widehat{x}; y) - \log\widetilde{q}(x;y)\\
&  =xS-\widetilde{A}(x)-\frac{1}{2}\frac{x^{2}}{\tau^{2}}\\
&  \quad -\widehat{x}S+\widetilde{A}(\widehat{x})+\frac{1}{2}\frac{\widehat{x}^{2}}{\tau^{2}}+\frac{1}{2}\frac{(x-\widehat{x})^{2}}{\tau_q^{2}}.
\end{align*}
This is, by design, zero at $x=\widehat{x}$ and can be bounded as
\begin{align*}
\log\widetilde{w}(x,y) &  \leq\frac{1}{2}\frac{\widehat{x}^{2}}{\tau^{2}}+\{S-\widetilde{A}^{\prime}(\widehat{x})\}(x-\widehat{x})-\frac{1}{2}\frac{x^{2}}{\tau^{2}}+\frac{1}{2}\frac{(x-\widehat{x})^{2}}{\tau_q^{2}}\\
&  =\frac{1}{2}(x-\widehat{x})^{2}d,
\end{align*}
by noting the first order condition that ${\widehat{x}}/{\tau^2}=S-\widetilde{A}^{\prime}(\widehat{x})$. The constant $d$ is defined to be $\ $
\[
d=\frac{1}{\tau_q^{2}}-\frac{1}{\tau^{2}},
\]
and $d>0$ if we choose $\tau_q^{2}< \tau^{2}$.
Hence
\begin{equation}
E^{X}\{\widetilde{w}(X,y)^{a}\}\leq E^{X}\left[\exp\left\{\frac{ad}{2}(X-\widehat{x})^{2}\right\}  \right]
,\label{eq:expect_num}%
\end{equation}
where again the expectation is with respect to $q(x \mid y)=\varphi(x\mid\widehat{x},\tau_q^{2})$. As $a>0$, clearly the above expectation exists if $d\leq0$
which would imply choosing $\tau_q^{2}\geq\tau^{2}$. To obtain a precise condition we note that
\begin{align*}
\frac{(X-\widehat{x})^{2}}{\tau_q^{2}}  &  \sim\chi_{1}^{2}.										
\end{align*}
Considering the moment generating function of the $\chi^2$-distribution we know that the expectation (\ref{eq:expect_num}) exists provided
\begin{equation}
ad\tau_q^{2}=a\bigg(1-\frac{\tau_q^{2}}{\tau^{2}}\bigg)
<1,\quad \text{i.e.}\quad\tau_q^{2}>\frac{(a-1)}{a}\tau^{2}%
.\label{eq:condition_upper}%
\end{equation}
If this inequality holds, the moment generating function of the $\chi^{2}$-distribution exists and we have
\begin{equation*}
	E^X\left[\exp\left\{\frac{ad}{2}(X-\widehat{x})^{2}\right\}  \right] = \left(1 - ad\tau_q^2 \right)^{-1/2}.
\end{equation*}
Finally we obtain
\begin{equation}
E^X\{\widetilde{w}(X,y)^{a}\}\leq\left\{1-a\bigg(
1-\frac{\tau_q^{2}}{\tau^{2}}\bigg)\right\}^{-1/2}=\left\{
\frac{a\tau_q^{2}-(a-1)\tau^{2}}{\tau^{2}}\right\}^{-1/2}.\label{eq:upper_bnd}%
\end{equation}
as required.
\end{proof}
Note that by the upper bound in Proposition \ref{gaussian_upper} still depends on parameters via the variance term $\tau$. However, since the dependence is continuous we can find an upper bound over any compact set. Thus, we have the simple corollary.

\begin{corollary}\label{cor:gauss_uniform_upper}
Under the conditions of Proposition \ref{gaussian_upper} there exists a constant $K_1 < \infty$ such that
\begin{equation*}
	\sup_{\theta \in B(\overline{\theta})} E^{X}\{\widetilde{w}(X,y,\theta)^{a}\} \leq K_1
\end{equation*}
independent of $y$.
\end{corollary}
We can summarize the results so far in the following theorem.
\begin{theorem}\label{theorem:7}
Consider the random effects model \eqref{eq:exp_model} and assume we have an importance sampling estimator with proposal distribution
\begin{equation*}
	q(x \mid y) = \varphi(x; \widehat{x}, \tau_q^2)
\end{equation*}
and proposal variance $\tau_q^2>\frac{(a-1)}{a}\tau^{2}$. Assume additionally that either
\begin{itemize}
\item[$i)$]
$\sup_x \widetilde{A}'(x) < \infty$ or
\item[$ii)$]
$E^Y(Y^a) < \infty$ and $\sup_{\theta\in B(\bar{\theta})} \widetilde{A}'(0) < \infty$.
\end{itemize}
Then
\begin{equation*}
E^Y\left[\sup_{\theta \in B(\overline{\theta})} E^X\left\{\overline{w}(Y, X, \theta)^a \right\}\right] < \infty.
\end{equation*}
\end{theorem}

\begin{proof}
We have
\begin{align*}
	E^Y\left[\sup_{\theta \in B(\overline{\theta})} E^X\left\{\overline{w}(Y, X, \theta)^a \right\}\right] & = E^{Y}\left[\sup_{\theta\in B(\overline{\theta})} \frac{E^{X}\left\{\widetilde{w}(X,Y, \theta)^{a}\right\}}{E^{X}\left\{\widetilde{w}(X,Y, \theta)\right\}^{a}}\right] \\
	& \leq E^{Y}\left[\sup_{\theta\in B(\overline{\theta})} \frac{K_1}{E^{X}\left\{\widetilde{w}(X,Y, \theta)\right\}^{a}}\right] < \infty.
\end{align*}
where the first inequality is by Corollary \ref{cor:gauss_uniform_upper} and the second by Corollary \ref{cor:upper_exp_inv}.
\end{proof}

For the logistic model of Section 7, $\widetilde{A}^{\prime}(x)$ is bounded above by a constant. Indeed
\begin{equation*}
	\widetilde{A}'(x) = \sum_{j=1}^J A'(c_j^\T\beta + x) = \sum_{j=1}^J \frac{e^{c_j^{\T}\beta + x}}{1 + e^{c_j^{\T}\beta + x}}.
\end{equation*}
Hence, we know (see Remark \ref{cor:inverse_gauss}) that
\begin{equation*}
E^Y\left[\sup_{\theta\in B(\overline{\theta})} E^{X\mid Y}\left\{\overline{w}(X,Y)^{a}\right\}\right] < \infty
\end{equation*}
for all $a$ if we take, for example, $\tau_q^{2}=\tau^{2}$. We note, however, that the proposal
may not be particularly efficient as the proposal variance would ideally be made to be proportional to $1/J$, where $J$ represents the number of observations associated with each latent variate. Hence, taking $\tau_q^{2}=\tau^{2}$, for example, may be much too large as a choice for $\tau_q^{2}$. This naturally leads to consideration of the $t$-distribution
which has heavier tails, see for example \cite[Chapter 9]{mcbook} and so
controls the numerator term. We consider the $t$-distribution proposal centred at the mode, with scaling
$\tau_q^2$, so that $q(x\mid y)=t_{\nu}(x\mid\widehat{x},\tau_q^{2})$. For the $t$-proposal, we have
\begin{equation}\label{eq:t_prop_def}
\widetilde{q}(x; y) = \bigg\{1+\frac{(x-\widehat{x})^{2}}{\nu\tau_q^{2}}\bigg\}^{-(\nu+1)/2}, \quad q(\widehat{x} \mid y)=\frac{\sqrt{\nu\pi}\,\Gamma\left(\nu/2\right)  \tau_q}{\Gamma\left\{(\nu+1)/2\right\}}.
\end{equation}
We proceed in the same manner as in the Gaussian case. First we compute the bound from Lemma \ref{lem:inv_upper_bound} for the $t$-distribution.
Assume the proposal is a $t$-distribution centred at the mode $q(x \mid y, \theta) = t_{\nu}(x\mid\widehat{x},\tau_q^{2})$, then
\begin{equation*}
C = \frac{\tau}{\tau_q}\sqrt{\frac{2}{\nu}}\frac{\Gamma\left(
\frac{\nu+1}{2}\right)  }{\Gamma\left(  \frac{\nu}{2}\right)  }
\end{equation*}
and thus
\begin{equation*}
	\frac{1}{E^X\left\{\widetilde{w}(X,y)\right\}} \leq \frac{\tau_q}{\tau}\sqrt{\frac{\nu}{2}}\frac{\Gamma\left(  \frac{\nu}{2}\right)  }{\Gamma\left(
\frac{\nu+1}{2}\right)  }(b +1).
\end{equation*}

\begin{proposition}
\label{prop:t_upper} For the target $h(x;y)$ of (\ref{eq:pi_exp_uni}) with
$q(x\mid y, \theta)=t_{\nu}(x\mid \widehat{x},\tau_q^{2})$ specified above we shall
assume that the function $x \mapsto A(x)$ is a monotonically non-decreasing convex function. Then,
\[
E^{X\mid Y}\left\{\widetilde{w}(X,Y)^{a}\right\}  \leq K_{2}^{a},
\]
where
\[
K_{2}=\left\{\frac{\tau^{2}}{\tau_q^{2}}\frac{(\nu+1)}{\nu}\right\}
^{\frac{(\nu+1)}{2}}\exp\left\{\frac{\nu}{2}\bigg(\frac{\tau_q^{2}}{\tau^{2}}-1-\frac{1}{\nu}\bigg)  \right\},
\]
for $\tau_{q}^{2}<\frac{(\nu+1)}{\nu}\tau^{2}$ and $K_{2}=1$ for $\tau_{q}^{2}\geq\frac{(\nu+1)}{\nu}\tau^{2}$.
\end{proposition}
Unlike the Gaussian proposal above, the $t$-distributed proposal does not have any restriction on how small
the variance $\tau_q^{2}$ can be. This might be chosen, for
example, according to the second derivative of $\log h(x;y)$ at $0$ so that
$\tau_q^{-2}=\tau^{-2}+\widetilde{A}^{\prime\prime}(0)$. This would
reflect the influence of a large number of repeated observations, $J$.

\begin{proof}[Proof of Proposition \ref{prop:t_upper}]
Recall that $x \mapsto A(x)$ is convex and thus always dominates its chord
\begin{equation*}
	\widetilde{A}(x) \geq \widetilde{A}(\widehat{x}) + \widetilde{A}'(\widehat{x})(x-\widehat{x})
\end{equation*}
for any values $x, \widehat{x}$. For the modified $\log$ weight this yields
\begin{align*}
\log\widetilde{w}(x,y) &  =\log h(x; y) - \log h(\widehat{x}; y) - \log\widetilde{q}(x;y)\\
&  =xS-\widetilde{A}(x)-\frac{1}{2}\frac{x^{2}}{\tau^{2}}\\
&  \quad -\widehat{x}S+\widetilde{A}(\widehat{x})+\frac{1}{2}\frac{\widehat{x}^{2}}{\tau^{2}}+\frac{(\nu+1)}{2}\log\bigg\{  1+\frac{(x-\widehat{x})^{2}}{\nu\tau_q^{2}}\bigg\}  \\
&  \leq\frac{1}{2}\frac{\widehat{x}^{2}}{\tau^{2}}+\{S-\widetilde{A}%
^{\prime}(\widehat{x})\}(x-\widehat{x})-\frac{1}{2}\frac{x^{2}}%
{\tau^{2}} +\frac{(\nu+1)}{2}\log\bigg\{1+\frac{(x-\widehat{x})^{2}}{\nu\tau_q^{2}}\bigg\}.
\end{align*}
We recall that ${\widehat{x}}/{\tau^2}=S-\widetilde{A}^{\prime}(\widehat{x})$. Hence
\begin{align*}
\log\widetilde{w}(x,y) &  \leq -\frac{(x - \widehat{x})^{2}}{2\tau^{2}} + \frac{(\nu+1)}{2}\log\bigg\{1+\frac{(x-\widehat{x})^{2}}{\nu\tau_q^{2}}\bigg\}.
\end{align*}
Writing $\widetilde{x}=(x-\widehat{x})/\tau_q$ we obtain
\[
\log\widetilde{w}(x,y)\leq-\frac{1}{2}\frac{\tau_q^{2}}%
{\tau^{2}}\widetilde{x}^{2}+\frac{(\nu+1)}{2}\log\bigg(1+\frac{\widetilde{x}^{2}}{\nu}\bigg).
\]
The resulting symmetric function can be verified to be maximized at $\widetilde{x}^{2}%
=(\nu+1)\tau^{2}/\tau_q^{2}-\nu$, provided this expression is positive, otherwise the only maximising root is at
$\widetilde{x}=0$ and so $\log\widetilde{w}(x,y) \leq 0$. If the expression is positive we obtain an upper bound
\begin{align*}
\log\widetilde{w}(x,y) &  \leq-\frac{1}{2}\left\{  (\nu+1)-\nu
\frac{\tau_q^{2}}{\tau^{2}}\right\}  +\frac{(\nu+1)}{2}\log\left\{
\frac{(\nu +1)}{\nu}\frac{\tau^{2}}{\tau_q^{2}}\right\}.  \\
\end{align*}

\end{proof}

\begin{corollary}\label{cor:t_uniform_bound}
Under the conditions of Proposition \ref{prop:t_upper} there exists a constant $K_3 < \infty$ such that
\begin{equation*}
	\sup_{\theta \in B(\overline{\theta})} E^{X}\{\widetilde{w}(X,y)^{a}\} \leq K_3
\end{equation*}
independent of $y$.
\end{corollary}

\begin{proof}
The constant in Proposition $\ref{prop:t_upper}$ depends on $\theta$ only through $\tau$. Moreover, the upper bound in Proposition $\ref{prop:t_upper}$ is continuous in $\tau$ and thus can be bounded over the compact set $B(\overline{\theta})$.
\end{proof}

We can summarize the results regarding the $t$-distribution in the following theorem.

\begin{theorem}\label{theorem:8}
Consider the random effects model \eqref{eq:exp_model} and assume we have an importance sampling estimator with proposal distribution
\begin{equation*}
	q(x \mid y) = t_{\nu}(x\mid\widehat{x},\tau_q^{2})
\end{equation*}
with $\tau_q^2>0$. Assume additionally that either
\begin{itemize}
\item[$i)$]
$\sup_x \widetilde{A}'(x) < \infty$ or
\item[$ii)$]
$E(Y^a) < \infty$ and $\sup_{\theta\in B(\bar{\theta})} \widetilde{A}'(0) < \infty$.
\end{itemize}
Then
\begin{equation*}
E^Y\left[\sup_{\theta \in B(\overline{\theta})} E^{X\mid Y}\left\{\overline{w}(Y, X, \theta)^a \right\}\right] < \infty.
\end{equation*}
\end{theorem}

\begin{proof}
We can bound
\begin{align*}
	E^{Y}\left[\sup_{\theta \in B(\overline{\theta})} E^{X\mid Y}\left\{\overline{w}(Y, X, \theta)^a \right\}\right] & = E^{Y}\left[\sup_{\theta\in B(\overline{\theta})} \frac{E^{X\mid Y}\left\{\widetilde{w}(X,Y, \theta)^{a}\right\}}{E^{X\mid Y}\left\{\widetilde{w}(X,Y, \theta)\right\}^{a}}\right] \\
	& \leq E^{Y}\left[\sup_{\theta\in B(\overline{\theta})} \frac{K_3}{E^{X\mid Y}\left\{\widetilde{w}(X,Y, \theta)\right\}^{a}}\right] < \infty.
\end{align*}
where the first inequality is by Corollary \ref{cor:t_uniform_bound} and the second by Corollary \ref{cor:upper_exp_inv}.
\end{proof}

Theorem \ref{theorem:7} and Theorem \ref{theorem:8} provide simple and verifiable conditions for Assumption \ref{ass4} to hold in the case of generalized linear mixed models when using a Gaussian proposal or a $t$-distribution. We have established these conditions by formulating assumptions on the models and the proposal.
The assumptions that are required for the model are fulfilled in the Binomial and Poisson cases as pointed out in Remark \ref{cor:inverse_gauss}.
Gaussian proposals require that the variance is large enough, namely
$$\tau_q^2>\frac{1+\Delta}{2+\Delta}\tau^{2},$$
where $0 < \Delta < 1$ corresponds to the quantity in Assumption \ref{ass4}. When one proposes from a $t$-distribution instead, no such restriction is required.

\renewcommand{\thesection}{S4}
\section{Further Simulation studies}\label{sec:furthersim}

\subsection{Toy example}
We consider first a simple Gaussian latent variable model where
\begin{align*}
X_{t}\sim\mathcal{N}(\theta,1) & ,\qquad Y_{t}\mid X_{t}=x\sim\text{\ensuremath{\mathcal{N}}}(x,1).
\end{align*}
Here $X_t, (t=1, \ldots, T)$ are assumed to be independent. In this case, the likelihood associated to $T$ observations can be
computed exactly as $p(y_{1:T}\mid\theta)=\prod_{t=1}^{T}\varphi(y_{t};\theta,2)$.
This makes it an easy example to examine Assumption \ref{ass1}. The maximum likelihood estimator and Fisher information are given by
\begin{align*}
	\hat{\theta}_{T}^{\omega} & =\frac{1}{T}\sum_{t=1}^{T}Y_{t},\qquad I_{T}\left(\theta\right)=I_{T}=\frac{T}{2}.
\end{align*}
If we assign a zero mean Gaussian prior to $\theta$ of variance $\sigma_{0}^{2}$
then the posterior is also normal with mean $\mu_{\mathrm{post}}$
and variance $\sigma_{\mathrm{post}}^{2}$ given by
\begin{align*}
\mu_{\mathrm{post}} & =\bigg(\frac{1}{\sigma_{0}^{2}}+\frac{T}{2}\bigg)^{-1}\bigg(\frac{\sum_{t=1}^{T}Y_{t}}{2}\bigg),\qquad\sigma_{\mathrm{post}}^{2}=\bigg(\frac{1}{\sigma_{0}^{2}}+\frac{T}{2}\bigg)^{-1}.
\end{align*}
Assume the data are arising from the model with true parameter value
$\bar{\theta}$.  It follows readily from Pinsker's inequality that
the Bernstein-von Mises theorem holds for $\Sigma=2$ as we have as
$T\rightarrow\infty$
\begin{align*}
 & \int\left|\pi_{T}^{\omega}(\theta)-\varphi\left(\theta;\hat{\theta}_{T}^{\omega},I_{T}^{-1}\right)\right|\mathrm{d}\theta=\int\left|\varphi\left(\theta;\mu_{\mathrm{post}},\sigma_{\mathrm{post}}^{2}\right)-\varphi\left(\theta,\hat{\theta}_{T}^{\omega},\frac{2}{T}\right)\right|\mathrm{d}\theta\overset{\mathbb{P}^{Y}}{\longrightarrow}0.
\end{align*}
Hence this model fulfils Assumption 1.
To estimate the likelihood we simulate data from the model with $\bar{\theta}=0.5$ and use $\sigma_{0}^{2}=10^{10}$.
The likelihood is estimated using importance sampling%
\[
\hat{p}(y_{1:T}\mid\theta,U)=\prod_{t=1}^{T}\frac{1}{N}\sum_{i=1}^{N}\varphi(y_{t}-U_{t,i};\theta,1),\quad U_{t,i}\sim\mathcal{N}(0,1).
\]
In order to prove that Assumption 3 is fulfilled we show the stronger Assumption 4, i.e. for some $\Delta > 0$
\begin{align*}
	E\left[\sup_{\theta\in B(\bar{\theta})} E\left\{\overline{w}(y,U,\theta)^{2+\Delta}\right\}\right] & = E\left[\sup_{\theta\in B(\bar{\theta})} E\left\{\frac{\varphi(y-U;\theta,1)^{2+\Delta}}{\varphi(y; \theta, 2)^{2+\Delta}}\right\}\right] < \infty
\end{align*}
In a first step we compute for $a>0$
\begin{align*}
	E\left\{\frac{\varphi(y-U;\theta,1)^{a}}{\varphi(y; \theta, 2)^{a}}\right\} & = \left(\frac{2^{a-1}}{\pi}\right)^{1/2}\int_{-\infty}^{\infty}\exp\left(-\frac{a(y-x-\theta)^{2}}{2}+\frac{a(y-\theta)^{2}}{4}-\frac{x^{2}}{2}\right)\mathrm{d}x\\
	& = \left(\frac{2^{a-1}}{\pi}\right)^{1/2}\int_{-\infty}^{\infty}\exp\left(-\frac{2a(y-x-\theta)^{2}-a(y-\theta)^{2}+2x^{2}}{4}\right)\mathrm{d}x.
\end{align*}
Completing the square yields
\begin{align*}
	& 2a(y-x-\theta)^{2}-a(y-\theta)^{2}+2x^{2} \\
	& = 2(a+1)\left(x-\frac{a}{\left(a+1\right)}(y-\theta)\right)^{2}-\frac{a\left(a-1\right)}{a+1}(y-\theta)^{2}
\end{align*}
and
\begin{align*}
	E\left\{\frac{\varphi(y-U;\theta,1)^{a}}{\varphi(y; \theta, 2)^{a}}\right\} & = \left(\frac{2^a}{a+1}\right)^{1/2}\exp\left\{\frac{a\left(a-1\right)}{4(a+1)}(y-\theta)^{2} \right\}.
\end{align*}
We now consider
\[
\sup_{\theta\in B(\bar{\theta})}E\left\{\overline{w}(y, U, \theta)^{a}\right\}=\left(\frac{2^a}{a+1}\right)^{1/2}\exp\left\{  \frac{a\left(a-1\right)}{4(a+1)}\sup_{\theta\in B(\bar
{\theta})}\{(y-\theta)^{2}\}\right\}.
\]
Now let us write
\begin{align*}
(y-\theta)^{2}  & =(y-\bar{\theta}+\bar{\theta}-\theta)^{2}\\
& =(y-\bar{\theta})^{2}+2(y-\bar{\theta})(\bar{\theta}%
-\theta)+(\bar{\theta}-\theta)^{2}%
\end{align*}
and consider $\theta\in B(\overline{\theta})$ corresponding to  $\left\vert
\theta-\overline{\theta}\right\vert \leq\varepsilon$, where $\varepsilon>0$. It is clear then that $(y-\theta)^{2}$ is optimised over $B(\overline{\theta
})$ at either $\theta=\overline{\theta}+\varepsilon$ or $\theta=\overline{\theta
}+\varepsilon$. Let us denote $y_{D}=y-\overline{\theta}$ and $d=\overline
{\theta}-\theta$ for simplicity so that
\[
(y-\theta)^{2}=y_{D}^{2}+2y_{D}d+d^{2},
\]
Then we consider an upper
bound on this which is quadratic in $y_{D}$ as
\[
(1+\alpha)y_{D}^{2}+(1+\varepsilon^{2}),
\]
where we need to determine $\alpha$ to achieve bounding for all values
$\left\vert d\right\vert \leq\varepsilon$. By symmetry of the left-hand side,
we need only consider the supremum case $d=\varepsilon$ so that%
\[
(1+\alpha)y_{D}^{2}+(1+\varepsilon^{2})\geq y_{D}^{2}+2y_{D}\varepsilon
+\varepsilon^{2},
\]
in which case, examining the roots of the resulting quadratic in $y_{D}$, it
is required that $4\varepsilon^{2}-4\alpha\leq0$, so $\alpha\geq
\varepsilon^{2}$. Taking $\alpha=$ $\varepsilon^{2}$ and using
the bounding quadratic expression we obtain,
\begin{align*}
\sup_{\theta\in B(\overline{\theta})}E[\overline{w}(y, U, \theta)^{a}] &
=\left(  \frac{2^{a}}{a+1}\right)  ^{\frac{1}{2}}\exp\left\{  \frac
{a(a-1)}{4(a+1)}\sup_{\theta\in B(\overline{\theta})}\{(y-\theta
)^{2}\}\right\}  \\
&  \leq\left(  \frac{2^{a}}{a+1}\right)  ^{\frac{1}{2}}\exp\left\{
\frac{a(a-1)}{4(a+1)}(y-\overline{\theta})^{2}(1+\varepsilon^{2}%
)+\frac{a(a-1)}{4(a+1)}(1+\varepsilon^{2})\right\}  \\
&  =g(y;a).
\end{align*}
So finally it is required that
\[
E_{Y}\{g(y;a)\}=
\int\nolimits_{-\infty}^{\infty}
g(y;a)\varphi(y;\overline{\theta},2)\mathrm{d}y<\infty,
\]
for $a=2 + \Delta$ for some $\Delta>0$. The above integral is finite when%
\[
\frac{a(a-1)}{(a+1)}(1+\varepsilon^{2})<1.
\]
Hence with $a=\Delta+2$,
\[
\varepsilon^{2}<\frac{(3+\Delta)}{(2+\Delta)(1+\Delta)}-1,
\]
with the right-hand side always positive provided $\Delta<\sqrt{2}-1$. 

We apply the pseudo-marginal method to this model to demonstrate how our result can approximate its characteristics. For the Markov chain, we use a random walk proposal with variance equal to the inverse Fisher
information $I_{T}^{-1}$ scaled by $\ell=2$. For each $T$, we run
a pseudo-marginal chain for various $N$ to sample the posterior for
$250000$ iterations as well as the limit Markov chain of
kernel $\tilde{P}_{\ell,\sigma}$. In Table \ref{tab:comparison2}
we summarize the simulations results. As expected, we find that both
the average acceptance probability and the integrated autocorrelation
time for $f\left(\theta\right)=\theta$ of the pseudo-marginal algorithm
converge to those of the limiting Markov chain as $T$ increases.

\begin{table}
{\centering
\begin{tabular}{clrrrrr}
Data $T$ & Particles $N$ & $\hat{\sigma}$ & $\widehat{\textsc{iat}}$ & $\hat{\mathrm{pr}}_{\mathrm{acc}}$ & $\widehat{\textsc{iat}}\left(\tilde{P}_{\ell=2,\sigma=\hat{\sigma}}\right)$ & $\hat{\mathrm{pr}}_{\mathrm{acc}}\left(\tilde{P}_{\ell=2,\sigma=\hat{\sigma}}\right)$\tabularnewline
$T=20$  & 6  & $1$$\cdot$$70$ & $17$$\cdot$$55$ & $18$$\cdot$$69\%$ & $31$$\cdot$$25$ & $15$$\cdot$$32\%$\tabularnewline
				& 8  & $1$$\cdot$$44$ & $12$$\cdot$$34$ & $23$$\cdot$$14\%$ & $17$$\cdot$$62$ & $20$$\cdot$$27\%$\tabularnewline
				& 10 & $1$$\cdot$$24$ & $10$$\cdot$$76$ & $26$$\cdot$$34\%$ & $12$$\cdot$$44$ & $24$$\cdot$$25\%$\tabularnewline
				& 12 & $1$$\cdot$$12$ &  $8$$\cdot$$98$ & $28$$\cdot$$78\%$ & $10$$\cdot$$02$ & $27$$\cdot$$19\%$\tabularnewline
$T=30$  & 8  & $1$$\cdot$$83$ & $27$$\cdot$$70$ & $15$$\cdot$$41\%$ & $46$$\cdot$$57$ & $13$$\cdot$$17\%$\tabularnewline
				& 11 & $1$$\cdot$$47$ & $16$$\cdot$$32$ & $20$$\cdot$$24\%$ & $18$$\cdot$$64$ & $19$$\cdot$$61\%$\tabularnewline
				& 14 & $1$$\cdot$$30$ & $12$$\cdot$$04$ & $24$$\cdot$$03\%$ & $12$$\cdot$$74$ & $23$$\cdot$$29\%$\tabularnewline
				& 17 & $1$$\cdot$$16$ & $10$$\cdot$$85$ & $26$$\cdot$$68\%$ & $9$$\cdot$$91$ & $26$$\cdot$$09\%$\tabularnewline
$T=50$  & 20 & $1$$\cdot$$85$ & $30$$\cdot$$46$ & $13$$\cdot$$94\%$ & $41$$\cdot$$53$ & $13$$\cdot$$10\%$\tabularnewline
				& 30 & $1$$\cdot$$48$ & $18$$\cdot$$59$ & $19$$\cdot$$58\%$ & $17$$\cdot$$53$ & $19$$\cdot$$51\%$\tabularnewline
				& 40 & $1$$\cdot$$29$ & $13$$\cdot$$30$ & $23$$\cdot$$59\%$ & $11$$\cdot$$63$ & $23$$\cdot$$34\%$\tabularnewline
				& 50 & $1$$\cdot$$16$ & $10$$\cdot$$51$ & $26$$\cdot$$86\%$ & $9$$\cdot$$91$ & $26$$\cdot$$09\%$\tabularnewline
$T=100$ & 20 & $1$$\cdot$$86$ & $34$$\cdot$$64$ & $13$$\cdot$$01\%$ & $41$$\cdot$$04$ & $12$$\cdot$$81\%$\tabularnewline
				& 30 & $1$$\cdot$$51$ & $17$$\cdot$$98$ & $19$$\cdot$$15\%$ & $18$$\cdot$$73$ & $18$$\cdot$$93\%$\tabularnewline
				& 40 & $1$$\cdot$$32$ & $14$$\cdot$$56$ & $23$$\cdot$$15\%$ & $13$$\cdot$$59$ & $22$$\cdot$$99\%$\tabularnewline
				& 50 & $1$$\cdot$$16$ & $10$$\cdot$$51$ & $26$$\cdot$$33\%$ & $9$$\cdot$$91$  & $26$$\cdot$$09\%$\tabularnewline
$T=200$ & 80 & $1$$\cdot$$83$ & $38$$\cdot$$35$ & $13$$\cdot$$11\%$ & $46$$\cdot$$57$ & $13$$\cdot$$17\%$\tabularnewline
				& 120 & $1$$\cdot$$52$ & $20$$\cdot$$65$ & $18$$\cdot$$90\%$ & $20$$\cdot$$42$ & $18$$\cdot$$58\%$\tabularnewline
				& 160 & $1$$\cdot$$30$ & $13$$\cdot$$87$ & $22$$\cdot$$94\%$ & $12$$\cdot$$74$ & $23$$\cdot$$29\%$\tabularnewline
				& 200 & $1$$\cdot$$17$ & $11$$\cdot$$15$ & $26$$\cdot$$07\%$ & $9$$\cdot$$73$ & $26$$\cdot$$05\%$\tabularnewline
\end{tabular}
\caption{For $T$ data and $N$ particles: standard deviation $\hat{\sigma}$
of the log-likelihood estimator at $\bar{\theta}$, integrated autocorrelation
time $\hat{\tau}$ and average acceptance probability $\hat{p}_{\mathrm{acc}}$
for pseudo-marginal kernel with $\ell=2$ and limiting kernel $\tilde{P}_{\ell=2,\hat{\sigma}}$.}
\label{tab:comparison2}
}
\end{table}

\subsection{Stochastic Lotka-Volterra Model}
Assumption \ref{ass3} is difficult to verify in state space models.
To illustrate the applicability of our results beyond latent variable
models we investigate here a stochastic kinetic Lotka-Volterra model
arising in systems biology. Such models are used to describe interacting
species in a predator and prey setting. In particular we consider
the model with transition equations given by

\begin{align*}
\mathbb{P}\left(X_{1,t+h}-X_{1,t}=1,X_{2,t+h}-X_{2,t}=0\mid X_{1,t}=x_{1,t},X_{2,t}=x_{2,t}\right) & =\beta_{1}x_{1,t}+o(h)\\
\mathbb{P}\left(X_{1,t+h}-X_{1,t}=-1,X_{2,t+h}-X_{2,t}=1\mid X_{1,t}=x_{1,t},X_{2,t}=x_{2,t}\right) & =\beta_{2}x_{1,t}x_{2,t}+o(h)\\
\mathbb{P}\left(X_{1,t+h}-X_{1,t}=0,X_{2,t+h}-X_{2,t}=-1\mid X_{1,t}=x_{1,t},X_{2,t}=x_{2,t}\right) & =\beta_{3}x_{2,t}+o(h),
\end{align*}

where $X_{1,t}$ and $X_{2,t}$ denotes the number of preys and predators
at time $t\in[0,T]$. This model has been previously investigated,
for example in \autocite{andrieu:doucet:holenstein2009} and \autocite{wilkinson2012}.
We assume independent gamma priors for the kinetic rate parameter
vector $\beta=(\beta_{1},\beta_{2},\beta_{3})$ with
\[
\beta_{1}\sim\Gamma(5,5),\quad\beta_{2}\sim\Gamma(1{\cdot}5,10),\quad\beta_{3}\sim\Gamma(3{\cdot}5,5).
\]
In our simulations we assume we are only able to observe predator
and prey $X_{t}=(X_{1,t},X_{2,t})$ at discrete equidistant time points
with independent measurement error $Y_{i,t}=X_{i,t}+W_{i,t},\:i=1,2,\:t=0,\ldots,50$
where $W_{i,t}\sim\mathcal{N}(0,10^{2})$.
The artificial data have been generated using the Gillespie algorithm
\autocite{Gillespie1977} for the rate constants $\beta=(1,0$$\cdot$$005,0$$\cdot$$6)$.

In this context, it is difficult to develop standard MCMC algorithms
to sample the posterior distribution while the pseudo-marginal algorithm
can be easily applied as an unbiased estimate of the likelihood can
be computed using a bootstrap particle filter; see, e.g., \autocite{andrieu:doucet:holenstein2009}
and \autocite[Chapter 10]{wilkinson2012}. We use a multivariate Gaussian
random walk proposal with scaling factor $\ell=2$$\cdot$$17$ and covariance
matrix close to the posterior covariance, which we estimated in a
short preliminary run. This can efficiently implemented in R \autocite{R_lang} using the package \textbf{smfsb} \autocite{wilkinson2012} and the example code which can be found on the author's blog.

The algorithm is then run for $250000$ iterations. We collect
acceptance rate and computing time $\textsc{ct}(N)=\textsc{iat}(N)\cdot N$
for a range of particles $N$, see Table \ref{tab:lotka_volterra}.
In practice we do not choose $\sigma\left(\bar{\theta}\right),$ but
the number of particles, $N$, which is also displayed in Table \ref{tab:lotka_volterra}.
For comparison we also give an estimate of $\sigma\left(\bar{\theta}\right)$
for given $N$.

The computing time is optimized at $N=225$ for all rates, $\beta_{1}$,
$\beta_{2}$ and $\beta_{3}$. We estimate $\sigma\left(\bar{\theta}\right)$
to be $1$$\cdot$$44$, slightly above the results of Table \ref{tab:multi_dim}
suggesting $\sigma=1$$\cdot$$24.$ The corresponding acceptance rate of $18$$\cdot$$57\%$
is in accordance with the one suggested by our theory, which for parameter
dimension $d=3$ yields an asymptotically optimal rate of around $19$$\cdot$$30\%$ $(\ell=2{\cdot}17, \sigma=1{\cdot}24)$.
We conjecture that the deviation from the results obtained in the
limiting case are due to the fact that the posterior is not very concentrated
around $\bar{\theta}$.

\begin{table}
\centering
\label{tab:lotka_volterra}
{
\begin{tabular}{cccccc}
Particles $N$ & Acceptance Rate & $\textsc{ct}(\beta_{1})$ & $\textsc{ct}(\beta_{2})$ & $\textsc{ct}(\beta_{3})$ & $\text{\ensuremath{\hat{\sigma}}(\ensuremath{\bar{\theta}})}$\tabularnewline
100 & 8$\cdot$92\%  & 7375  & 9035  & 7564  & 2$\cdot$38 \tabularnewline
125 & 11$\cdot$17\%  & 6668  & 6717  & 6580  & 2$\cdot$10 \tabularnewline
150 & 13$\cdot$44\%  & 5805  & 5903  & 6208  & 1$\cdot$84 \tabularnewline
175 & 15$\cdot$62\%  & 5688  & 6137  & 6101  & 1$\cdot$68 \tabularnewline
200 & 17$\cdot$03\%  & 5564  & 5632  & 5744  & 1$\cdot$55 \tabularnewline
225 & 18$\cdot$57\%  & 5178  & 5452  & 5122  & 1$\cdot$44 \tabularnewline
250 & 19$\cdot$54\%  & 6107  & 6958  & 5831  & 1$\cdot$36 \tabularnewline
275 & 20$\cdot$82\%  & 5473  & 6087  & 5248  & 1$\cdot$30 \tabularnewline
300 & 21$\cdot$47\%  & 6436  & 6340  & 5959  & 1$\cdot$22 \tabularnewline
325 & 22$\cdot$41\%  & 5771  & 6586  & 6178  & 1$\cdot$19 \tabularnewline
350 & 23$\cdot$20\%  & 6406  & 6234  & 6393  & 1$\cdot$13 \tabularnewline
\end{tabular}
}
\caption{Comparison of the computing time for different numbers of particles
in the stochastic Lotka-Volterra model.}
\end{table}

\begin{figure}
\centering\includegraphics[width=\linewidth]{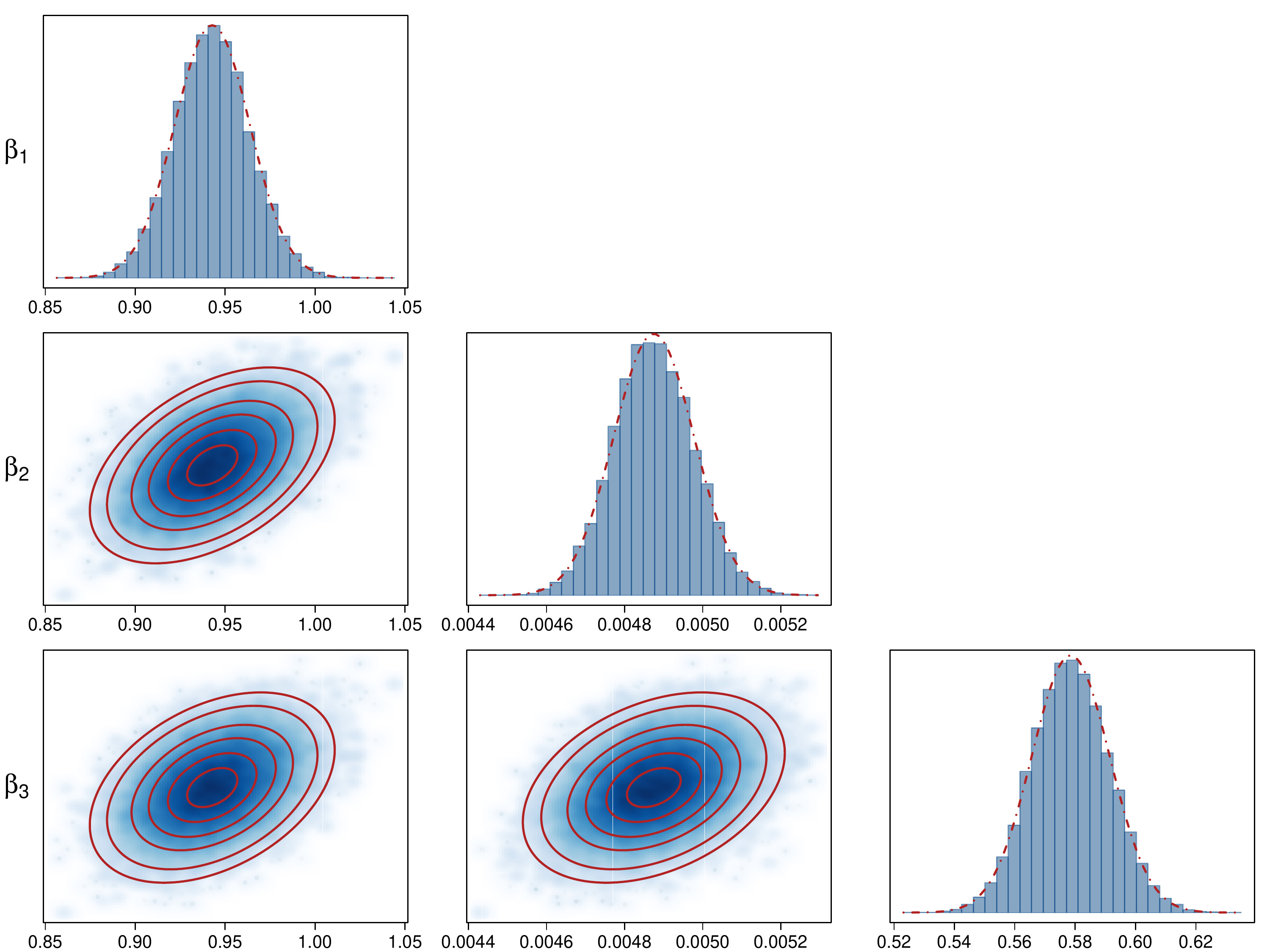}

\caption{Histogram of marginal posterior $p(\beta_{i}\mid y_{1:T}),\,i=1,2,3$
on the diagonal with Gaussian approximation (line) using sample mean
and variance. In addition, we show density estimates of the projections
to the plane. The ellipses indicate the contour lines of a Gaussian
with sample mean and sample covariance matrix. It is clear from the
plots that the posterior is very close to a Gaussian.}
\end{figure}

\textcite{Sherlock2015efficiency} carry out Bayesian inference for a 5-dimensional
stochastic Lotka-Volterra model using the pseudo-marginal
algorithm based on a data set with $T=50$ observations. The
authors optimize over a grid of values for both $\sigma$ and $\ell$.
Experimentally, it was found that the optimal standard deviation was
$\sigma\approx1$$\cdot$45 and the optimal tuning for the random walk achieved
at $\ell=2$$\cdot$048 with an associated optimal jumping rate of 15$\cdot$39\%.
This is slightly above our guidelines with the values $\hat{\sigma}_{\mathrm{opt}}=1$$\cdot$$30,\hat{\ell}_\mathrm{opt}=2$$\cdot$$17$
and $\pr_{\mathrm{acc}}(\hat{\sigma}_{\mathrm{opt}},\hat{\ell}_{\mathrm{opt}})=17$$\cdot$35\%
obtained in Table \ref{tab:multi_dim}.

\printbibliography[heading=subbibliography] % print section bibliography
\end{refsection}

\end{document}